\documentclass[a4paper,UKenglish,cleveref, autoref, thm-restate]{lipics-v2021}

\usepackage{amsmath,amssymb}
\usepackage{color}
\usepackage{enumerate}
\usepackage{graphicx}
\usepackage{hyperref}
\usepackage{microtype}
\usepackage{tabularx}
\usepackage{xspace}
\usepackage{algorithm}
\usepackage[noend]{algpseudocode}
\usepackage{varwidth}  
\usepackage{fancybox}
\usepackage{todonotes}
\usepackage{amsthm}
\usepackage[font=small,skip=0pt]{caption}
\usepackage{framed}
\usepackage{mdframed}

\colorlet{shadecolor}{gray!20}

\definecolor{ShadowColor}{RGB}{30,150,190}

\makeatletter
\newcommand\Cshadowbox{\VerbBox\@Cshadowbox}
\def\@Cshadowbox#1{%
  \setbox\@fancybox\hbox{\fbox{#1}}%
  \leavevmode\vbox{%
    \offinterlineskip
    \dimen@=\shadowsize
    \advance\dimen@ .5\fboxrule
    \hbox{\copy\@fancybox\kern.5\fboxrule\lower\shadowsize\hbox{%
      \color{ShadowColor}\vrule \@height\ht\@fancybox \@depth\dp\@fancybox \@width\dimen@}}%
    \vskip\dimexpr-\dimen@+0.5\fboxrule\relax
    \moveright\shadowsize\vbox{%
      \color{ShadowColor}\hrule \@width\wd\@fancybox \@height\dimen@}}}
\makeatother

\newcommand{\B}{\ensuremath{\mathcal{B}}}

\newcommand{\F}{\ensuremath{\mathcal{F}}}

\renewcommand{\Pr}[1]{\ensuremath{\mathbf{Pr}\left[#1\right]}}

\newcommand{\mypara}[1]{\subparagraph*{#1}}

\newcommand{\COMMENTED}[1]{{}}
\newcommand{\Reals}{\mathbb{R}}
\newcommand{\Nats}{\mathbb{N}}
\newcommand{\Ex}[1]{\ensuremath{\mathbf{E}[#1]}}
\newcommand{\Integers}{\mathbb{Z}}
\renewcommand{\geq}{\geqslant}
\renewcommand{\leq}{\leqslant}
\renewcommand{\ge}{\geqslant}
\renewcommand{\le}{\leqslant}

\newcommand{\radius}{\mathrm{radius}}
\newcommand{\dist}{\mathrm{dist}}
\newcommand{\weight}{w}

\newcommand{\ceil}[1]{\left\lceil #1 \right\rceil}
\newcommand{\floor}[1]{\left\lfloor #1 \right\rfloor}
\newcommand{\eps}{\varepsilon}

\newcommand{\etal}{\emph{et al.}\xspace}

\newcommand{\nw}[1]{\textcolor{black}{#1}}

\newcommand{\mdb}[1]{}
\newcommand{\mm}[1]{}
\newcommand{\myomit}[1]{}
\newcommand{\lb}[1]{}

\newcommand{\opt}{\mbox{{\sc opt}}\xspace}
\newcommand{\optkz}{\opt_{k,z}}

\newcommand{\et}{t_{\mathrm{exp}}}
\newcommand{\at}{t_{\mathrm{arr}}}

\newcommand{\myopt}{\mathrm{opt}}

\newcommand{\Bopt}{\B_{\myopt}}
\newcommand{\Pout}{P_{\mathrm{out}}}

\newcommand{\zeronorm}[1]{\| #1 \|_0}

\newcommand{\rrounds}{R}
\newcommand{\feasibleRadius}{\hat{r}}
\newcommand{\ball}{b}
\newcommand{\greedy}{\textsc{Greedy}\xspace}
\newcommand{\computeCoreset}{\textsc{MBCConstruction}\xspace}
\newcommand{\outliersBound}{\frac{6z}{m}+3\log{n}}
\newcommand{\peps}{\gamma}
\newcommand{\streaming}{\textsc{InsertionOnlyStreaming}\xspace}
\newcommand{\UpdateCoreset}{\textsc{UpdateCoreset}\xspace}
\newenvironment{proofof}[1]{\begin{trivlist} \item {\bf Proof
#1:~~}}
  {\qed\end{trivlist}}

\newcommand\Tstrut{\rule{0pt}{3.5mm}}         % = `top' strut

\newenvironment{myquote}%
  {\list{}{\leftmargin=5mm \rightmargin=5mm}\item[]}%
  {\endlist}
\newenvironment{claiminproof}{\begin{myquote}\noindent\emph{Claim.}}{\end{myquote}}
\newenvironment{proofinproof}{\begin{myquote}\noindent\emph{Proof.}}{\hfill $\lhd$ \end{myquote}}

\title{
$k$-Center Clustering with Outliers in the MPC and Streaming Model} %TODO Please add

\titlerunning{$k$-Center Clustering with Outliers in the MPC and Streaming Model} %TODO optional, please use if title is longer than one line

\funding{The work in this paper is supported by the Dutch Research Council (NWO) through Gravitation-grant NETWORKS-024.002.003.}

\author{Mark de Berg}{Department of Computer Science, TU Eindhoven, the Netherlands}{m.t.d.berg@tue.nl}{}{}
\author{Leyla Biabani}{Department of Computer Science, TU Eindhoven, the Netherlands}{l.biabani@tue.nl}{}{}
\author{Morteza Monemizadeh}{Department of Computer Science, TU Eindhoven, the Netherlands}{m.monemizadeh@tue.nl}{}{}

\authorrunning{M.~de Berg and L.~Biabani and M.~Monemizadeh} 
\Copyright{Mark de Berg and Leyla Biabani and Morteza Monemizadeh}

\authorrunning{L.~Biabani and M.~de Berg and M.~Monemizadeh} 
\Copyright{Leyla Biabani and Mark de Berg and Morteza Monemizadeh}

\ccsdesc[100]{Theory of computation~Design and analysis of algorithms} %TODO mandatory: Please choose ACM 2012 classifications from https://dl.acm.org/ccs/ccs_flat.cfm 

\keywords{$k$-center problem, outliers, massively parallel computing, streaming} %TODO mandatory; please add comma-separated list of keywords

\category{} %optional, e.g. invited paper

\relatedversion{} %optional, e.g. full version hosted on arXiv, HAL, or other respository/website
\EventEditors{John Q. Open and Joan R. Access}
\EventNoEds{2}
\EventLongTitle{42nd Conference on Very Important Topics (CVIT 2016)}
\EventShortTitle{CVIT 2016}
\EventAcronym{CVIT}
\EventYear{2016}
\EventDate{December 24--27, 2016}
\EventLocation{Little Whinging, United Kingdom}
\EventLogo{}
\SeriesVolume{42}
\ArticleNo{23}
\hideLIPIcs

\begin{document}
\thispagestyle{empty}
\maketitle

%--------------------------------------------------------------------------------------
%--------------------------------------------------------------------------------------
%--------------------------------------------------------------------------------------
\begin{abstract}
Given a point set $P \subseteq X$ of size $n$ in a metric space $(X,\dist)$ of doubling dimension~$d$ 
and two parameters $k \in \mathbb{N}$ and $z \in \mathbb{N}$, 
the $k$-center problem with $z$ outliers asks to return a set $\mathcal{C}^* = \{ c^*_1,\cdots,c^*_k\} \subseteq X$ of $k$ centers 
such that the maximum distance of all but $z$ points of $P$ to their nearest center in $\mathcal{C}^*$ is minimized. 
An $(\eps,k,z)$-coreset for this problem is a weighted point set $P^*$ such that 
an optimal solution for the $k$-center problem with $z$ outliers on $P^*$ gives a
$(1\pm\eps)$-approximation for the $k$-center problem with $z$ outliers on~$P$. 
We study the construction of such coresets in the Massively Parallel Computing (MPC) model, 
and in the insertion-only as well as the fully dynamic streaming model. 
We obtain the following results, for any given error parameter $0 < \eps \le 1$: 
In all cases, the size of the computed coreset is $O(k/\eps^d+z)$.
\begin{itemize}
\item \textbf{Algorithms for the MPC model.} 
      In this model, the data are distributed over $m$ machines. One of these machines
      is the coordinator machine, which will contain the final answer, the
      other machines are worker machines. 
      \begin{itemize}
      \item We present a deterministic $2$-round algorithm that uses $O(\sqrt{n})$ machines,
            where the worker machines have $O(\sqrt{nk/\eps^d} + \sqrt{n}\cdot\log(z+1))$ local memory,
            and the coordinator has $O(\sqrt{nk/\eps^d}+\sqrt{n}\cdot\log(z+1) + z)$ local memory. 
            The algorithm can handle point sets~$P$ that 
            are distributed arbitrarily (possibly adversarially) over the machines.
     \item We present a randomized algorithm that uses only a single round, under
           the assumption that the input set~$P$ is initially distributed randomly over the machines.
           This algorithm also uses $O(\sqrt{n})$ machines, where the worker machines have
           $O(\sqrt{nk/\eps^d})$ local memory and the coordinator has
           $O(\sqrt{nk/\eps^d}+ \sqrt{n}\cdot\min(\log{n}, z)+z)$ local memory. 
    \item  We present a deterministic algorithm that obtains a trade-off between the number of rounds, $R$, 
           and the storage per machine. 
      \end{itemize}
\item \textbf{Algorithms and lower bounds in the streaming model.} 
      In this model, we have a single machine, with limited storage, and the point set $P$ 
      is revealed in a streaming fashion.
      \begin{itemize}
      \item We present the first lower bound for the (insertion-only) streaming model, 
            where the points arrive one by one and no points are deleted. 
            We show that any deterministic algorithm that maintains an $(\eps,k,z)$-coreset 
            must use $\Omega(k/\eps^d+ z)$ space. We complement this with a deterministic streaming algorithm  
            using $O(k/\eps^d+ z)$ space, which is thus optimal. 
      \item We study the problem in fully dynamic data streams, where points
            can be inserted as well as deleted. Our algorithm works for point sets
            from a $d$-dimensional discrete Euclidean space $[\Delta]^d$,
            where $\Delta\in \Nats$ indicates the size of the universe from which the coordinates are taken.
            We present the first algorithm for this setting. It
            constructs an $(\eps,k,z)$-coreset. Our (randomized) algorithm uses only $O((k/\eps^d+z)\log^4 (k\Delta/\eps \delta))$ space. 
            We also present an $\Omega((k/\eps^d)\log{\Delta}+z)$ lower bound for deterministic fully dynamic streaming algorithms.
      \item Finally, for the sliding-window model, where we are interested in maintaining 
            an $(\eps,k,z)$-coreset for the last $W$ points in the stream, 
            we show that any deterministic streaming algorithm that guarantees a $(1+\eps)$-approximation 
            for the $k$-center problem with outliers in~$\Reals^d$ must use $\Omega((kz/\eps^d)\log \sigma)$ space, 
            where $\sigma$ is the ratio of the largest and smallest distance between any two points in the stream.
            This improves a recent lower bound of De~Berg, Monemizadeh, and Zhong~\cite{DBLP:conf/esa/BergMZ21,DBLP:journals/corr/abs-2109-11853} 
            and shows the space usage of their algorithm is optimal. Thus, our lower bound gives a (negative) answer to 
            a question posed by De~Berg~\etal~\cite{DBLP:conf/esa/BergMZ21,DBLP:journals/corr/abs-2109-11853}.
      \end{itemize}
\end{itemize}
\end{abstract}
%--------------------------------------------------------------------------------------
%--------------------------------------------------------------------------------------
%--------------------------------------------------------------------------------------

            % Interestingly, De~Berg, Monemizadeh, and Zhong~\cite{DBLP:conf/esa/BergMZ21} developed 
            % a sliding window algorithm that uses $O((kz/\eps^d)\log \sigma)$ space. 
            % Our lower bound shows that their algorithm is optimal and gives a (negative) answer to a question posed by De~Berg~\etal~\cite{DBLP:conf/esa/BergMZ21}, who asked whether there is a sketch for 
            % this problem whose storage is polynomial in~$d$.

%\end{titlepage}

% \newpage %%%% Commented by Leyla
\setcounter{page}{1}

%------------------------------------------------------------------------------------------
\section{Introduction}
%------------------------------------------------------------------------------------------
% \mypara{$k$-Center clustering with outliers.}
%------------------------------------------------------------------------------------------
Clustering is a classic topic in computer science and machine learning with applications 
in pattern recognition~\cite{DBLP:books/crc/aggarwal2013}, image processing~\cite{DHANACHANDRA2015764}, 
data compression~\cite{DBLP:journals/sj/PaekK17,DBLP:conf/ipcv/SchmiederCL09},   
healthcare~\cite{10.1007/978-981-16-4641-6_21,Priebe5995}, and more. 
Centroid-based clustering~\cite{DBLP:books/lib/Bishop07,everitt2011cluster} is a well-studied
type of clustering, due to its simple formulation and many applications. 
An important centroid-based clustering variant 
%, besides $k$-median and the $k$-means clustering,
is \emph{$k$-center clustering}: given a point set $P$ in a metric space,
find $k$ congruent (that is, equal-sized) balls of minimum radius that together cover~$P$. 
The $k$-center problem has been studied extensively; see for 
example~\cite{DBLP:journals/siamcomp/CharikarCFM04,DBLP:conf/stoc/FederG88,DBLP:journals/tcs/Gonzalez85,DBLP:conf/stoc/Har-PeledM04}.
% \mdb{These were called greedy, but are all of them greedy?}
The resulting algorithms are relatively easy to implement, but often
cannot be used in practice because of noise or anomalies in the data.
% Dirty data leads to dirty results, and thus, it makes hard for a human to 
% interpret the results easily\footnote{``Interpretability~\cite{DBLP:conf/nips/KimKK16} 
% is the degree to which a human can consistently predict the algorithm's result''.}. 
One may try to first clean the data, but
with the huge size of modern data sets, this would take a data scientist 
an enormous amount of time. A better approach is to take noise and anomalies
into account when defining the clustering problem.
%It is simple to see that even adding $k$ outliers can change 
%the optimal $k$-center radius significantly.
% As an example, as we see in Figure~\ref{fig:overfitting}, even adding $k$ outliers can change 
% the optimal $k$-center radius significantly.
This leads to the $k$-center problem with outliers, where we allow a certain number
of points from the input set~$P$ to remain uncovered. 
Formally, the problem is defined as
follows. 

Let $P$ be a point set in a metric space $(X,\dist)$. Let $k \in \Nats$ and $z\in \Nats$ be two parameters. 
In the \emph{$k$-center problem with $z$ outliers} we want to compute a set of $k$ congruent balls 
that together cover all points from $P$, except for at most $z$ outliers.
We denote by $\optkz(P)$ the radius of the balls in an optimal solution.
In the weighted version of the problem, we are also given a 
function $w: P \rightarrow \mathbb{Z}^+$ 
that assigns a positive integer weight to the points in~$P$.
% where we assume the total weight $\sum_{p \in P} w(p)$ is $n$. 
%  \mdb{Does not seem necessary.}
The problem is now defined as before, except that the total weight of the outliers 
(instead of their number) is at most~$z$.
We consider the setting where the metric space~$(X,\dist)$
has a constant doubling dimension, defined as follows.
For a point $p \in X$ and a radius~$r \ge 0$, 
let $\ball(p,r) = \{q \in X: \dist(p,q) \le r\}$ be the ball of radius $r$ around~$p$. 
The \emph{doubling dimension} of $(X,\dist)$ is the smallest~$d$ such 
that any ball~$b$ can be covered by $2^d$ balls of radius $\radius(b)/2$. 
%\lb{We need to define $[n]=\{1,\ldots,n\}$ somewhere.}\mm{I defined in the abstract.}\mdb{It is no longer there. Let's define it the first time we use it.}

%The problem becomes even harder as nowadays data is massive and is not given to us as a whole, 
%that is, real data is either
%or collected from sensors~\cite{AFSAR2014198} that are widely spread in an area or various locations, or 
%revealed in a streaming~\cite{DBLP:conf/stoc/Indyk04} fashion. 
In this paper, we study the $k$-center problem with outliers for data sets that are
too large to be stored and handled on a single machine. We consider two models.

%------------------------------------------------------------------------------------------
% \mypara{The MPC model.}
%------------------------------------------------------------------------------------------
In the first model, the data are distributed over many parallel machines~\cite{DBLP:conf/kdd/EneIM11}.
We will present algorithms for the \emph{Massively Parallel Computing (MPC) model},
introduced by Karloff, Suri, and Vassilvitskii~\cite{DBLP:conf/soda/KarloffSV10} and
later extended by Beame~\etal~\cite{DBLP:journals/jacm/BeameKS17} and Goodrich~\etal~\cite{DBLP:conf/isaac/GoodrichSZ11}. 
In this model, the input of size $n$ is initially distributed among $m$ machines, each
with a local storage of size~$s$. Computation proceeds in synchronous rounds, in which each machine 
performs an arbitrary local computation on its data and sends messages to the other machines.
% where the total messages sent or received by each machine in every round should not exceed its memory.
%The total amount of storage should be linear (that is, we require $m\cdot s = O(n)$).
%\lb{The storage of our MPC algorithms is not linear! Maybe we should delete the previous statement.}
%\mdb{Or write near-linear and omit the bound? But perhaps omitting the statement is better. That also saves space :)}
The local storage should be sublinear space (that is, we require $s = o(n)$); 
 the goal is to minimize the space per machine by employing a large number
of machines in an effective manner. 
We also want to use only a few rounds of communication.
We assume that there is one designated
machine, the \emph{coordinator}, which at the end of the computation must contain the 
answer to the problem being solved; the other machines are called \emph{worker machines}.
We make this distinction because MPC clusters often have 
% \footnote{Customers can borrow machines from big providers such as Google, Amazon, or Microsoft Azure to perform their parallel tasks. 
% The price per CPU and GPU machines is significantly different 
% and their storage sizes as well. See here: https://aws.amazon.com/ec2/instance-types/g4/ and https://cloud.google.com/compute/gpus-pricing 
% and https://azure.microsoft.com/en-us/pricing/details/virtual-machines/series/} 
different CPUs and GPUs running in parallel, some of which are more powerful and have
more storage than others.
% We often use these expensive GPU units to perform highly demanding tasks such as computing 
% final solutions (e.g. a coordinator receives coresets from worker (e.g. CPU) machines 
% for which computes a final solution) for which we expect to have a bigger storage.  

%------------------------------------------------------------------------------------------
% \mypara{The streaming model.}
%------------------------------------------------------------------------------------------
In the second model, the data arrives in a streaming fashion. 
We consider the classical, insertion-only streaming setting~\cite{DBLP:journals/jcss/AlonMS99}, and  
the dynamic streaming setting~\cite{DBLP:conf/stoc/Indyk04} where the stream contains insertions
as well as deletions of points. The goal of the streaming model is to approximately
solve a problem at hand---in our case, the $k$-center problem with outliers---using 
space that is sublinear in $n$. 
Note that there are no assumptions on the order of arrivals, that is, the stream can be
adversarially ordered.

%--------------------------------------------------------------------------------------
% \paragraph*{Our results and relation to previous work}
%--------------------------------------------------------------------------------------
We study the $k$-center problem with outliers in the MPC model and in the streaming model.
We will not solve the problems directly, but rather compute
a coreset for it: a small subset of the points that can be used to approximate
the solution on the actual point set. Since their introduction by Agarwal, Her-Peled
and Varadarajan~\cite{DBLP:journals/jacm/AgarwalHV04} coresets have been instrumental
in designing efficient approximation algorithms, in streaming and other models. 
For the (weighted) $k$-center problem with outliers, coresets can be defined as follows.
%------------------------------------------------------------------------------------------
\begin{definition}[$(\eps,k,z)$-coreset]
\label{def:coreset}
Let $0 < \eps \le 1$ be a parameter. 
Let $P$ be a weighted point set with positive integer weights in  a metric space $(X, \dist)$ and
let $P^*$ be a subset
% \footnote{Note that the weight of points in $P^*$ can be different than the weight of points in $P$.}
of $P$, where the points in $P^*$ may have different weights than the corresponding points in $P$.
Then, $P^*$ is an $(\eps,k,z)$-coreset of $P$ if the following hold:
\begin{enumerate}
    \item $(1-\eps)\cdot \optkz(P) \leq \optkz(P^*) \leq (1+\eps)\cdot \optkz(P)$.
    \item Let $B = \{\ball(c_1,r),\cdots,\ball(c_k,r)\}$ be any set of congruent balls 
    in the space $(X,\dist)$ such that the total weight of the points in $P^*$ that are not covered by $B$ is at most $z$.  
    Let $r' := r + \eps \cdot \optkz(P)$. 
    Then, the total weight of the points in $P$ that are not covered by the $k$ expanded congruent balls 
    $B' = \{ \ball(c_1,r'),\cdots ,\ball (c_k,r')\}$ is at most $z$.
\end{enumerate}
\end{definition}
%------------------------------------------------------------------------------------------
Table~\ref{table:summary} lists our algorithmic results for computing small coresets,
in the models discussed above, as well as existing results.
Before we discuss our results in more detail, we make two remarks about two of the 
quality measures in the table.
%------------------------------------------------------------------------------------------
\begin{table*}[t]
\begin{center}
\begin{tabular}{c|ccccc}
model & setting  &  approx. & storage  & deterministic & ref. \\
\hline
MPC        & $1$-round &  $1+\eps$ & $\sqrt{nk/\eps^{2d}} + \sqrt{n\log{n}}/\eps^{d} + z/\eps^{d}$ \Tstrut &  no & \cite{DBLP:journals/pvldb/CeccarelloPP19} \\ 
 & $1$-round &   $1+\eps$ & $\sqrt{nk/\eps^d}+ \sqrt{n}\cdot \min(\log{n}, z)+z$  &  no & here \\
\cline{2-6}
        & $1$-round &  $1+\eps$ & $\sqrt{nk/\eps^{2d}} + \sqrt{nz/\eps^{2d}}$ \Tstrut &  yes & \cite{DBLP:journals/pvldb/CeccarelloPP19} \\ 
%        & $1$-round &  $1+\eps$ & $\sqrt{n(k+z)/\eps^{2d}}$ &  yes & \cite{DBLP:journals/pvldb/CeccarelloPP19} \\ 
        & $2$-round &    $1+\eps$ & $\sqrt{nk/\eps^d}
+\sqrt{n}\cdot \log (z+1) + z$  &  yes & here \\     
\cline{2-6}
        & $\rrounds$-round &   $(1+\eps)^\rrounds$ & $n^{1/(\rrounds+1)}(k/\eps^{d} +z)^{\rrounds/(\rrounds+1)}$  &  yes & here \\    
\hline        
streaming 
        & insertion-only &   $1+\eps$ & $k/\eps^d + z/\eps^{d}$  &  yes & \cite{DBLP:journals/pvldb/CeccarelloPP19} \\
        & insertion-only &   $1+\eps$ & $k/\eps^d +z$  &  yes & here \\
        & insertion-only &   $1+\eps$ & $\boldsymbol{\Omega(k/\eps^d +z)}$  &  yes & here \\
\cline{2-6}
        & sliding-window &   $1+\eps$ & $(kz/\eps^d)\log \sigma$  &  yes & \cite{DBLP:conf/esa/BergMZ21} \\
        & sliding-window &   $1+\eps$ & $\boldsymbol{\Omega((kz/\eps)\log \sigma)}$  &  yes & \cite{DBLP:conf/esa/BergMZ21} \\
        & sliding-window &   $1+\eps$ & $\boldsymbol{\Omega((kz/\eps^d)\log \sigma)}$  &  yes & here \\
\cline{2-6}
        & fully dynamic &   $1+\eps$ & $(k/\eps^d+z)\log^4 (k\Delta/\eps \delta)$  &  no & here \\
        & fully dynamic &   $1+\eps$ & $\boldsymbol{\Omega((k/\eps^d)\log \Delta + z)}$  &  yes & here \\

\hline
\end{tabular}
\end{center}
\caption{Comparison of our results to previous work.
% Observe that if $z=0$, then we have the classical $k$-center problem without outliers. 
Storage bounds are asymptotic. 
The lower bounds are shown with $\Omega(\cdot)$ notations in bold. 
All results are for a metric space of doubling dimension~$d$,
except for the dynamic streaming algorithm which is for a discrete (Euclidean) space $[\Delta]^d$;
in both cases, the dimension~$d$ is considered to be a constant.
The size of the coreset computed by our algorithms is always~$O(k/\eps^d+z)$.
% In the deterministic column, ``no'' means that the proposed MPC algorithm is randomized. 
}
% \vspace{-0.4cm} % WARNING!!!!!
\label{table:summary}
\end{table*}
%--------------------------------------------------------------------------------------
%
%--------------------------------------------------------------------------------------
\mypara{{\rm \emph{About the approximation factor.}}} 
%--------------------------------------------------------------------------------------
Recall that our algorithms compute an $(\eps,k,z)$-coreset. 
To obtain an actual solution, one can run an offline algorithm for $k$-center
with outliers on the coreset. The final approximation factor then depends on the 
approximation ratio of the algorithm: running an optimal but slow 
algorithm on the coreset, gives a $(1+\eps)$-approximation; and 
running a fast $3$-approximation algorithm, for instance, gives
a $3(1+\eps)$-approximation. To make a fair comparison with the result of
Ceccarello~\etal~\cite{DBLP:journals/pvldb/CeccarelloPP19}, we list the
approximation ratio of their coreset in Table \ref{table:summary}, rather than their
final approximation ratio.
%
%--------------------------------------------------------------------------------------
\mypara{{\rm \emph{About the number of rounds of the MPC algorithms.}}}
%--------------------------------------------------------------------------------------
The original MPC model~\cite{DBLP:conf/soda/KarloffSV10} 
% and later extended by~\cite{DBLP:journals/jacm/BeameKS17,DBLP:conf/isaac/GoodrichSZ11} 
considers the number of communication rounds between the machines as a performance measure.
A few follow-up works count the number of computation rounds instead.
% As an example, if we have a MPC algorithm $\mathcal{A}$ in which workers compute 
% their sketches and send them to a coordinator where we compute the final solution, 
% these works consider $\mathcal{A}$ as a two-round of computation algorithm, but according to 
% the formal MPC model~\cite{DBLP:conf/soda/KarloffSV10}, $\mathcal{A}$ has  
% one-round communication complexity. 
As each communication round happens in between two computation rounds, we need
to subtract~1 from the bounds reported in~\cite{DBLP:journals/pvldb/CeccarelloPP19}.
In Table~\ref{table:summary} we made this adjustment, in order to get a fair comparison. 
%

%--------------------------------------------------------------------------------------
%\mypara{Previous work and our results for the MPC model. }
\mypara{Our results for the MPC model and relation to previous work.}
%--------------------------------------------------------------------------------------
Before our work, the best-known deterministic MPC algorithm was due to 
Ceccarello, Pietracaprina and Pucci~\cite{DBLP:journals/pvldb/CeccarelloPP19}, 
who showed how to compute an $(\eps,k,z)$-coreset in one round of communication
using $O(\sqrt{nk/\eps^{2d}} + \sqrt{nz}/\eps^{d})$ 
% using $O(\sqrt{n(k+z)/\eps^{2d}})$ 
local memory per machine, and $O(\sqrt{n/(k+z)})$ machines.
%      \mdb{How many machines? Size of coreset? Do they also have only one machine
%     that needs this much space, and workers that need a bit less?}
%    \lb{$O(\sqrt{n/(k+z)})$ machines. The size of coreset is $O(m(k+z)/\eps^d)$ 
%    because their goal is not to compute a coreset, but an actual solution. 
%   However, we can observe that the corset size will be $O((k+z)/\eps^d)$ 
%   if we compute the final coreset of their coreset}
%\mdb{Updated, based on Leyla's comment. I also updated the table. Leyla, you also put a comment in the first line of the table, can you put the correct bound there?}
(They also gave a slightly more efficient algorithm for the problem without outliers.)
% without outliers:  $O(\sqrt{nk} \cdot (1/\eps)^d)$ space per machine. 
Our main result for the MPC model is a 2-round algorithm for computing
a coreset of size $O(k/\eps^d+z)$, using $O(\sqrt{n\eps^d/k})$ machines having
$O(\sqrt{nk/\eps^d} +\sqrt{n}\cdot\log(z+1) + z)$ local memory. This is a significant
improvement for a large range of values of~$z$ and $k$. For
example, if $z=\sqrt{n}$ and $k=\log n$ then Ceccarello~\etal use $O(n^{0.75}/\eps^d)$
local memory, while we use $O(\sqrt{(n/\eps^d)\log n})$ local memory.
In fact, the local storage stated above for our solution is only needed
for the coordinator machine; the worker machines just need 
$O(\sqrt{nk/\eps^d} + \sqrt{n}\cdot\log(z+1))$ local storage, which 
is interesting as it avoids the $+z$ term.

% To make sure that the space usage does not have the term $O(\sqrt{nz}/\eps^d)$ like the previous work, we need to control the number of outliers sent to the coordinator machine. However, it seems hard to determine for a worker machine how many outliers it has. Previous approaches (as well as our randomized algorithm) deal with this by assuming that the points are distributed randomly over the machines, so that each machine has only ``few'' outliers in expectation. But in an adversarial setting, the outliers can be distributed very unevenly over the machines. We therefore develop a mechanism that allows each machine, in one round of communication, to obtain a good estimate of the number of outliers it has.

To avoid the term $O(\sqrt{nz}/\eps^d)$ in the storage of the previous work, we must control the number of outliers sent to the coordinator. 
However, it seems hard to determine for a worker machine how many outliers it has. Ceccarello~\etal~\cite{DBLP:journals/pvldb/CeccarelloPP19} % (as well as our randomized algorithm) 
assume that the points are distributed randomly over the machines, so each machine has only ``few'' outliers in expectation. But in an adversarial setting, the outliers can be distributed very unevenly over the machines. We develop a mechanism that allows each machine, in one round of communication, to obtain a good estimate of the number of outliers it has. 
Guha, Li and Zhang \cite{DBLP:journals/topc/GuhaLZ19} present a similar method to 
determine the number of outliers in each machine, but using their method 
the storage of each  worker machine will be $O(\sqrt{nk/\eps^d} +\sqrt{n}\cdot z)$.
Our refined mechanism reduces the dependency on $z$ from linear to logarithmic, 
which is a significant improvement.

We also present a 1-round randomized algorithm, and a deterministic $R$-round algorithm
giving a trade-off between the number of communication
rounds and the local storage; see Table~\ref{table:summary} for the bounds,
and a comparison with a 1-round algorithm of Ceccarello~\etal

%--------------------------------------------------------------------------------------
\vspace*{-4mm}
%\mypara{Previous work and our results for the streaming model.}
\mypara{Our results for the streaming model and relation to previous work.}
%--------------------------------------------------------------------------------------
% On the other hand, for small processing units (such as  Raspberry Pi) 
% we have a special  (more restricted) MPC model known as the one-round \emph{coordinator} 
% model~\cite{DBLP:journals/topc/GuhaLZ19,NIPS2017_7486cef2,NIPS2013_7f975a56} 
% (also known as \emph{simultaneous communication model}). 
% In this model, we have $m$ machines (Raspberry Pis or other small processing units) 
% and a coordinator (which is often a big processing unit or a server). 
% Each machine can communicate with the coordinator but they cannot communicate with each other. 
% The input data is partitioned among $m$ machines and these machines perform 
% their (fairly simple such as computing a sketch) computation 
% and then send their sketches to the coordinator. 
% The coordinator later performs the final computation using the union of received sketches and reports the solution.  
% As for the partitioning, the input data is partitioned either \emph{adversarially} or \emph{randomly}.  
% The former case is common in scenarios where small processing units 
% (such as sensors) collect data~\cite{PLAGERAS2018349}, and the latter one is used in scenarios 
% where there exists a dispatcher or distributor that receives a stream of data 
% and randomly assigns it to machines~\cite{DBLP:conf/wdag/GorenVM20,DBLP:conf/bigdataconf/Sanders20}. 
Early work focused on the problem without outliers in the insertion-only
model~\cite{DBLP:journals/siamcomp/CharikarCFM04,DBLP:conf/approx/McCutchenK08}.
McCutchen and Khuller~\cite{DBLP:conf/approx/McCutchenK08} also studied 
the problem with outliers, in general metric spaces, obtaining 
$(4 + \eps)$-approximation using $O(kz/\eps)$ space. 
More recently, Ceccarello~\etal~\cite{DBLP:journals/pvldb/CeccarelloPP19}  presented an algorithm for the $k$-center problem with outliers in spaces
of bounded doubling dimension,
which computes an $(\eps,k,z)$-coreset using $O(k/\eps^d+z/\eps^d)$ storage. 
We improve the result of Ceccarello~\etal by presenting a deterministic
algorithm that uses only $O(k/\eps^d+z)$ space. 
Interestingly, we will give a lower bound showing our algorithm is optimal. 

We next study the problem in the fully dynamic case, where the stream
may contain insertions as well as deletions. The $k$-center problem
with outliers has, to the best of our knowledge, not been studied in
this model. Our results are for the setting where the points in the
stream come from a $d$-dimensional discrete Euclidean space~$[\Delta]^d$. 
% (Note that then the maximum number of points in $P$ is at most $\Delta^d$.) 
We present a randomized dynamic streaming algorithm for this setting
that constructs an $(\eps,k,z)$-coreset using $O((k/\eps^d + z)\log^4(k\Delta/(\eps\delta)) )$ space. 
The idea of our algorithm is as follows. We construct a number 
of grids on the underlying space~$[\Delta]^d$, of exponentially increasing granularity.
For each of these grids, we maintain a coreset on the non-empty cells using 
an $s$-sample recovery sketch~\cite{DBLP:journals/tcs/BarkayPS15}, for a suitable
parameter $s=\Theta(k/\eps^d+z)$. We then use an $\zeronorm{F}$-estimator~\cite{DBLP:conf/pods/KaneNW10} 
to determine the finest grid that has at most $O(s)$ non-zero cells, and we prove that  
its corresponding coreset is an $(\eps,k,z)$-coreset with high probability.

Note that our dynamic streaming algorithm is randomized only because the subroutines 
providing an $\zeronorm{F}$-estimator and an $s$-sample recovery sketch are randomized. 
If both of these subroutines can be made deterministic, 
then our algorithm would also be deterministic, with bounds that 
are optimal up to polylogarithmic factors (See the lower bound 
that we obtain for the dynamic model in this paper). 
Interestingly, we can make the $s$-sample recovery 
sketch deterministic by using the Vandermonde matrix~\cite{DBLP:journals/tit/CandesRT06,DBLP:conf/icip/CandesR06,DBLP:conf/focs/PriceW11,DBLP:conf/approx/NelsonNW12}. 
%to maintain coresets for different grids. 
Such a deterministic recovery scheme can be used to return all non-zero cells of a grid 
with the exact number of points in each cell if the number of 
non-empty cells of that grid is at most $O(s)$. 
%If we knew which grid has at most $O(s)$ non-zero entries, 
To this end, we can use linear programming 
techniques % (say  interior-point methods~\cite{DBLP:books/cu/BV2014}) 
to retrieve the non-empty cells of that grid with their exact number of points. 
However, we do not know how to check deterministically if a grid has at most $O(s)$ non-zero cells at the moment.  
% The problem is there is a lower bound~\cite{DBLP:journals/jcss/AlonMS99} of $\Omega(n)$ for estimating the $\zeronorm{F}$-frequency 
% deterministically.  

Note that our dynamic streaming algorithm immediately gives a fully dynamic algorithm for the $k$-center problem with outliers that has a fast update
time in the standard (non-streaming) model. Indeed, after each update we can simply run a greedy algorithm, say the one in~\cite{DBLP:conf/esa/DingYW19}, 
on our coreset. This gives a dynamic  $(3+\epsilon)$-approximation algorithm  
with $O((k/\eps^d + z)\log^4(k\Delta/(\eps\delta)) )$ update time. 
Interestingly, to the best of our knowledge a dynamic algorithm with fast update time was not known so far
for the $k$-center problem with outliers, even in the standard setting where we
can store all the points. For the problem without outliers, there are some recent results in the fully dynamic model~\cite{DBLP:journals/tkde/ChanGHS22,DBLP:conf/alenex/GoranciHLSS21,DBLP:journals/corr/abs-2112-07050}. 
In particular, Goranci~\etal~\cite{DBLP:conf/alenex/GoranciHLSS21} 
developed a $(2+\epsilon)$-approximate dynamic algorithm 
for metric spaces of a bounded doubling dimension $d$. 
The update time of their algorithm is $O((\frac{2}{\epsilon})^{O(d)}\cdot \log\rho\log\log\rho)$ 
where $\rho$ is the spread ratio of the underlying space. 
Furthermore, Bateni~\etal~\cite{DBLP:journals/corr/abs-2112-07050} 
gave a $(2+\epsilon)$-approximate dynamic algorithm for any metric space
using an amortized update time of $O(k\,\textrm{polylog}\,(n,\rho))$. 
Both dynamic algorithms need $\Omega(n)$ space.
Our streaming algorithm needs much less space (independent of $n$), and can 
even deal with outliers. On the other hand, our algorithm works for discrete Euclidean spaces. 

The fully dynamic version of the problem is related to the 
\emph{sliding-window model}, where we are given window length~$W$, 
and we are interested in maintaining an $(\eps,k,z)$-coreset for
the last $W$ points in the stream. On the one hand, the fully-dynamic setting is
more difficult than the sliding-window setting, since any of the current points can be deleted. On the other hand, it is easier
since we are explicitly notified when a point is deleted, while 
in the sliding-window setting the expiration of a point may go unnoticed.
In fact, it is a long-standing open problem to see how 
different streaming models relate to each other.\footnote{\url{https://sublinear.info/index.php?title=Open_Problems:20}}

The sliding-window version of the $k$-center problem (without outliers) 
was studied by Cohen{-}Addad, Schwiegelshohn, and Sohler~\cite{DBLP:conf/icalp/Cohen-AddadSS16}
for general metric spaces. 
Recently, De~Berg, Monemizadeh, and Zhong~\cite{DBLP:conf/esa/BergMZ21} 
studied the $k$-center problem with outliers 
for spaces of bounded doubling dimension. The space usage of the latter
algorithm is $O((kz/\eps^d)\log \sigma)$, where $\sigma$ is the ratio of the largest
and the smallest distance between any two points in the stream.

% Note that the space bound has a factor $kz$, while our algorithm
% for the fully dynamic case has separate terms for $k$ and $z$.

%--------------------------------------------------------------------------------------
\vspace*{-4mm}
\mypara{Lower bounds for the streaming model.}
%--------------------------------------------------------------------------------------
The only lower bound that we are aware of for the $k$-center problem 
with $z$ outliers in different streaming settings is the one that 
De Berg, Monemizadeh, and Zhong~\cite{DBLP:conf/esa/BergMZ21,DBLP:journals/corr/abs-2109-11853} 
proved for the sliding-window model. 
In particular, they proved that any deterministic sliding-window algorithm that guarantees 
a $(1+\eps)$-approximation for the $k$-center problem
with outliers in~$\Reals^1$ must use $\Omega((kz/\eps)\log \sigma)$ space. 
However, this lower bound works for one-dimensional Euclidean space and in particular, 
it shows a gap between the space complexity of their algorithm 
which is $O((kz/\eps^d)\log \sigma)$ and their lower bound.
De~Berg~\etal~\cite{DBLP:conf/esa/BergMZ21,DBLP:journals/corr/abs-2109-11853} 
raised the following open question: 
\emph{``It would be interesting to investigate the dependency on the parameter $\eps$ in more detail
and close the gap between our upper and lower bounds. The main question here is whether it
is possible to develop a sketch whose storage is only polynomially dependent on the doubling
dimension $d$.'' 
}
We give a (negative) answer to this question, by proving
an $\Omega((kz/\eps^d)\log \sigma)$ lower bound for the 
sliding-window setting in $\Reals^d$ under the $L_{\infty}$-metric, thus improving the lower bound of
De~Berg, Monemizadeh, and Zhong and showing the optimality of their algorithm.
Our lower bound for the sliding-window model 
works in the same general setting as the lower bound of De~Berg~\etal~\cite{DBLP:conf/esa/BergMZ21,DBLP:journals/corr/abs-2109-11853}. 
%Indeed, their setting is very general: 
Essentially, the only restriction on the algorithm is that it can only update the solution when 
a new point arrives or at an explicitly stored expiration time of an input point. 
%In this model, they give an $\Omega((kz)\log \sigma)$ lower bound on the number of expiration times that need to be stored. 
This is a very powerful model since it does not make any assumptions about how the algorithm maintains a solution. 
It can maintain a coreset, but it can also maintain something completely different. 
% However, as we said before, their model inherently only works in the sliding-window model. 
% Indeed, in the insertion-only model, points never expire. Even in the fully dynamic model, 
% a deletion always happens through an explicit update, and so no expiration times need to be stored. 

The lower bound of De~Berg~\etal~\cite{DBLP:conf/esa/BergMZ21,DBLP:journals/corr/abs-2109-11853} 
inherently only works in the sliding-window model. 
Indeed, their lower bound is based on the expiration times of input points. 
However, in the insertion-only model, points never expire. Even in the fully dynamic model, 
a deletion always happens through an explicit update, and so no expiration times need to be stored. 
We are not aware of any lower bounds on the space usage of insertion-only
streaming algorithms or fully dynamic streaming model for the problem.  
We give the first lower bound for the insertion-only model, and show that any
deterministic algorithm that maintains an $(\eps,k,z)$-coreset in $\Reals^d$
must use $\Omega(k/\eps^d+z)$ space, thus proving the space complexity of our algorithm 
which is $O(k/\eps^d+z)$ is in fact, optimal.

Finally, we prove an $\Omega((kz/\eps^d)\log \Delta+z)$  lower bound for the fully dynamic streaming model, 
thus showing that the logarithmic dependency on $\Delta$ in the space bound of our algorithm is unavoidable.
Our lower bounds for the insertion-only and fully dynamic streaming models work for algorithms 
that maintain a coreset.
%  with a natural restriction that the total weight of the point set is at least the total weight of its coreset.

% \lb{I am not sure where we should put this paragraph}
% The lower bound of De Berg, Monemizadeh and Zhong~\cite{DBLP:conf/esa/BergMZ21,DBLP:journals/corr/abs-2109-11853} works in a very general model: essentially, the only restriction on the algorithm is that it can only update the solution when a new point arrives or at an explicitly stored expiration time of an input point. They then give an $\Omega((kz)\log \sigma)$ lower bound on the number of expiration times that need to be stored. This is a very powerful model since it does not make any assumptions about how the algorithm maintains a solution. It can maintain a coreset, but it can also maintain something completely different. However, it inherently only works in the sliding-window model. Indeed, in the insertion-only model, points never expire. Even in the fully dynamic model, a deletion always happens through an explicit update, and so no expiration times need to be stored. Our lower bound for the insertion-only and fully dynamic setting  therefore work for algorithms that maintain a coreset. Our lower bound for the streaming setting, however, works in the same model as the lower-bound of De Berg, Monemizadeh and Zhong.

%------------------------------------------------------------------------------------------
\section{Mini-ball coverings provide $(\eps,k,z)$-coresets}
\label{sec:mini-balls}
%------------------------------------------------------------------------------------------
The algorithms that we will develop are based on so-called mini-ball coverings.
A similar concept has been used implicitly before, see for example~\cite{DBLP:conf/esa/BergMZ21}. 
Here, we formalize the concept and prove several useful properties.
The idea behind the mini-balls covering is simple but powerful:
using mini-ball coverings we are able to improve the existing results
on the $k$-center problem with outliers both in the MPC model and in
the streaming models.
%
% Next, we state two useful properties for mini-ball coverings. 
% Namely, we first show that an $(\eps,k,z)$-mini-ball covering of a point set $P$ is an $(\eps,k,z)$-coreset of $P$. 
% Second, the union of mini-ball coverings is also a mini-ball covering.
% Finally, we develop an algorithm to compute an $(\eps,k,z)$-mini-ball covering of a point set $P$ in Section \ref{sec:mini-balls}. 
%------------------------------------------------------------------------------------------
\begin{definition}[$(\eps,k,z)$-mini-ball covering]
\label{def:miniball:covering}
Let $P$ be a weighted point set in  a metric space $(X, \dist)$
and let $k,z \in \mathbb{N}$, and $\eps\geq 0$. 
A weighted point set\footnote{Note that the weights of a point in $P^*$ can be different from its weight in $P$.}
$P^*=\{q_1,\ldots,q_{f}\} \subseteq P$ is an \emph{$(\eps,k,z)$-mini-ball covering} of $P$
if $P$ can be partitioned into pairwise disjoint subsets $Q_1,\ldots,Q_f$
with the following properties:
\begin{enumerate}
    \item \textbf{Weight property:} $\weight(q_i) =\sum_{p\in Q_i}\weight(p)$.
     and, hence, $\sum_{q\in P^*} \weight(q) = \sum_{p\in P}\weight(p)$.
    \item \textbf{Covering property:} $\dist(p,q_i)\leq \eps \cdot \optkz(P)$ for all $p\in Q_i$.
    In other words, $Q_i$ is contained in a ball of radius $\eps \cdot \optkz(P)$ around~$q_i$.
%    \mdb{``bounding the error'' does not really capture the condition.}
\end{enumerate}
For each $p\in Q_i$, we refer to $q_i$ as the \emph{representative} point of $p_i$.
\end{definition}
%------------------------------------------------------------------------------------------
\begin{figure}[t]
\begin{center}
\includegraphics[scale=0.30]{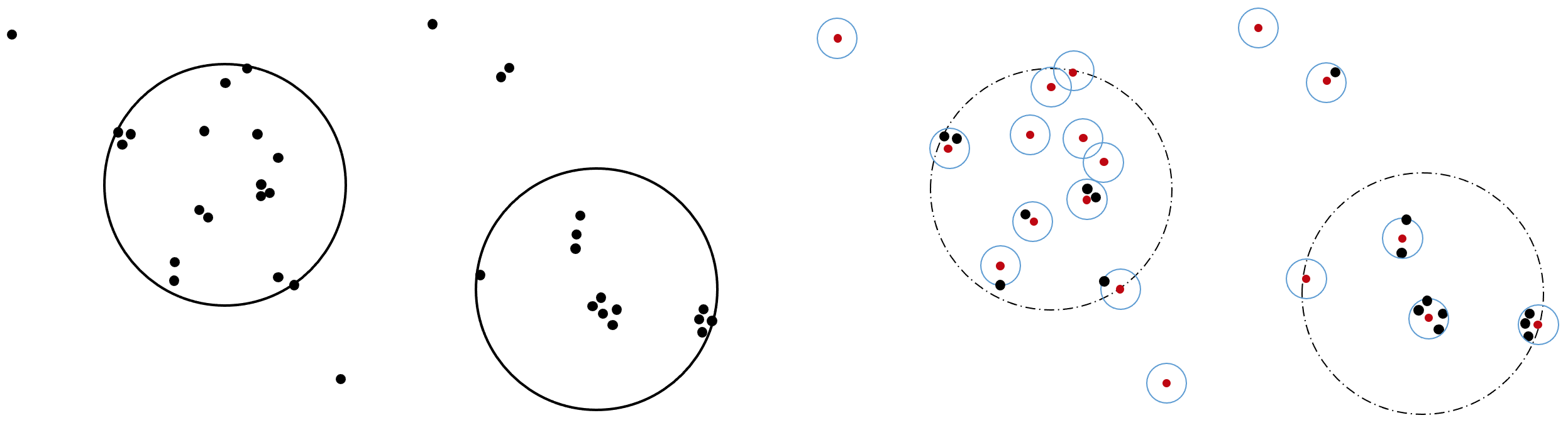}
\end{center}
\vspace{-0.3cm}
\caption{
Left: A set of points that are covered by $k=2$ balls with $z=5$ outliers.
Right: A mini-ball covering of the same point set.
The red points are the representative points. The weight of each mini-ball is the total weight of the points inside it.
}
\label{fig:miniball:covering}
\end{figure}
%------------------------------------------------------------------------------------------
See Figure \ref{fig:miniball:covering} for an example of mini-ball covering.
%\mdb{If we are really out of space, I think the figure can be omitted.}
Next, we show that an $(\eps,k,z)$-mini-ball covering of a point set $P$ is an $(\eps,k,z)$-coreset of $P$, 
and therefore $(1\pm\eps)$-approximates the optimal solution for the $k$-center problem with $z$ outliers.
The proof of this lemma, as well as missing proofs of other lemmas, can be found in the appendix.
%------------------------------------------------------------------------------------------
\begin{lemma}
\label{lem:MBC:is:coreset}
Let $P$ be a weighted point set in  a metric space $(X, \dist)$ and
let $P^*$ be an $(\eps, k, z)$-mini-ball covering of $P$. Then, $P^*$ is an $(\eps, k, z)$-coreset of $P$.
\end{lemma}
%------------------------------------------------------------------------------------------
The next lemma shows how to combine mini-ball coverings of subsets of $P$ into
a mini-ball covering for~$P$. Then Lemma \ref{lem:merge:coreset} proves that a mini-ball covering of a mini-ball covering is also a mini-ball covering,
albeit with adjusting the error parameters.
% The complete proofs can be found in the appendix.
%------------------------------------------------------------------------------------------
\begin{lemma}[Union Property]
\label{lem:union:coreset}
Let $P$ be a set of points in a metric space $(X,\dist)$. 
Let $k, z\in \Nats$ and $\eps \geq 0$ be parameters. 
Let $P$ be partitioned into disjoint subsets $P_1,\cdots,P_s$, and
let $Z = \{z_1,\cdots, z_s\}$ be a set of numbers such that 
$\opt_{k,z_i}(P_i) \leq \optkz(P)$ for each~$P_i$. 
If $P_i^*$ is an $(\eps, k, z_i)$-mini-ball covering of $P_i$ for each $1 \leq i \leq s$,
then $\cup_{i=1}^s P_i^*$ is an $(\eps,k,z)$-mini-ball covering of $P$. 
\end{lemma}
%------------------------------------------------------------------------------------------

%------------------------------------------------------------------------------------------
\begin{lemma}[Transitive Property]
\label{lem:merge:coreset}
Let $P$ be a set of $n$ points in a metric space $(X,\dist)$. 
Let $k, z \in \Nats$ and $\eps, \peps \geq 0$ be four parameters. 
Let $P^*$ be a $(\peps,k,z)$-mini-ball covering of $P$,
and let $Q^*$ be an $(\eps, k, z)$-mini-ball covering of $P^*$. 
Then, $Q^*$ is an $(\eps+\peps+\eps\peps, k,z)$-mini-ball covering of $P$.
\end{lemma}

%------------------------------------------------------------------------------------------
\mypara{An offline construction of mini-ball coverings.}
%\label{subsec:MBC:construction}
%------------------------------------------------------------------------------------------
In this section, we develop our mini-ball covering construction for a set $P$ of $n$ points 
in a metric space $(X,\dist)$ of doubling dimension~$d$.   
To this end, we first invoke the $3$-approximation algorithm
\greedy~by Charikar~\etal~\cite{DBLP:conf/soda/CharikarKMN01}.
%Indeed, let $P$ be a set of $n$ points in a general metric space $(X,\dist)$ (not necessarily of bounded doubling dimension). 
Their algorithm works for the $k$-center problem with outliers 
in general metric spaces (not necessarily of bounded doubling dimension) and returns $k$ balls 
of radius at most $3\cdot \optkz(P)$ which together cover all but at most $z$ points of the given points.
The running time of this algorithm, which we denote by  \greedy$(P,k,z)$, is $O(n^2k\log{n})$.  
Note that \greedy provides us with a bound on~$\optkz(P)$.
We use this to compute a mini-ball covering of $P$ in a greedy
manner, as shown in Algorithm~\computeCoreset.
% Next, for each non-empty mini-ball of radius $\eps\cdot\optkz(P)$, we choose a representative point whose weight is the total weight of points in that mini-ball, and add it to the mini-ball covering. Then, we remove all points in that mini-ball from $P$. 
%------------------------------------------------------------------------------------------
\begin{algorithm}[h] 
\caption{\computeCoreset$(P,k,z,\eps)$} 
\label{alg:compress}
\begin{algorithmic}[1]
\State Let $r$ be the radius of the $k$ congruent balls reported by~\greedy$(P,k,z)$. \label{step:greedy}
\State $P^* \gets \emptyset$.
\While{$|P| > 0 $} 
    \State Let $q$ be an arbitrary point in $P$ and let $R_q := \ball(q,\eps \cdot \frac{r}{3}) \cap P$. 
    \State Add $q$ to $P^*$ with weight $w(q) := \weight(R_q)$
    \label{line:MBCC:weight}
    \State $P \gets P \setminus R_q $. 
\EndWhile
\State Return $P^*$.
\end{algorithmic}
\end{algorithm}
%------------------------------------------------------------------------------------------

We will show that for metric spaces of doubling dimension~$d$, 
the number of mini-balls is at most $k(\frac{12}{\eps})^d+z$.
To this end, we first need to bound the size of any subset 
of $P$ whose pairwise distances are at least $\delta$.
% as follows. #TSS
%------------------------------------------------------------------------------------------
\begin{lemma}
\label{lem:coreset:size}
Let $P$ be a finite set of points in a metric space $(X,\dist)$ of doubling dimension $d$.
Let $0 < \delta \leq \optkz(P)$, and let $Q \subseteq P$ be a subset of $P$ such that for 
any two distinct points $q_1, q_2 \in Q$, $\dist(q_1, q_2) > \delta$. 
Then $|Q| \leq k\left( \frac{4\cdot\optkz(P)}{\delta} \right)^d+z$.
\end{lemma}
%------------------------------------------------------------------------------------------
Next we show that \computeCoreset computes an $(\eps, k, z)$-mini-ball covering.
%------------------------------------------------------------------------------------------
\begin{lemma}\label{lem:compress}
Let $P$ be a set of $n$ weighted points with positive integer weights in a metric 
space $(X,\dist)$ of doubling dimension $d$.  Let $k, z \in \Nats$ and $0 < \eps \le 1$.
Then \computeCoreset$(P,k,z,\eps)$ returns an \emph{$(\eps,k,z)$-mini-ball covering} of $P$ 
whose size is at most $k(\frac{12}{\eps})^d+z$. 
\end{lemma} 
%------------------------------------------------------------------------------------------
\begin{proof}
Let $r$ be the radius computed in Step~\ref{step:greedy} of \computeCoreset.
As \greedy~is a $3$-approximation algorithm, $\optkz(P) \leq r \leq 3\cdot \optkz(P)$.
We first prove that the reported set $P^*$  is an $(\eps,k,z)$-mini-ball covering of $P$,
and then we bound the size of $P^*$. 

By construction, the sets $R_q$ for $q\in P^*$ together form a partition
of~$P$. Since $q$ is added to $R_q$ with weight $w(R_q)$, the 
weight-preservation property holds. Moreover, for any $p\in R_q$ we have
$\dist(p,q) \leq \eps\cdot \frac{r}{3} \leq \eps\cdot \optkz(P)$.
Hence, $P^*$ is an $(\eps, k, z)$-mini-ball covering of $P$.

Next we bound the size of~$P^*$. Note that the distance between any two points
in $P^*$ is more than $\delta$, where $\delta = \eps \cdot \frac{r}{3}$.
Since $(X,\dist)$ has doubling dimension~$d$,
Lemma \ref{lem:coreset:size} thus implies that $|P^*| \leq k\cdot (4\cdot\frac{\optkz(P)}{\delta})^d+z$.
Furthermore, $\optkz(P) \leq r$. Hence,
\[
|P^*| \leq k\left(4\cdot\frac{\optkz(P)}{\delta}\right)^d+z 
= k\left(4\cdot\frac{\optkz(P)}{\eps r/3}\right)^d+z 
\leq k\left(4\cdot\frac{r}{\eps r/3}\right)^d+z 
= k\left(\frac{12}{\eps}\right)^d+z 
 . \qedhere
\]
% In conclusion, the size of mini-ball covering $P^*$ is at most $k(\frac{12}{\eps})^d+z$.
\end{proof}
%------------------------------------------------------------------------------------------

%------------------------------------------------------------------------------------------
\section{Algorithms for the MPC model}
%------------------------------------------------------------------------------------------
Let $M_1,\cdots,M_m$ be a set of $m$ machines. Machine~$M_1$ is labeled
as the \emph{coordinator}, and the others are \emph{workers}. 
Let  $(X,\dist)$ be a metric space of doubling dimension~$d$. 
Let $P \subseteq X$ be the input point set of size~$n$, which is stored in
a distributed manner over the~$m$ machines.
Thus, if $P_i$ denotes the point set of machine~$M_i$, then 
$P_i \cap P_j = \emptyset$ for $i\neq j$, and $\cup_{i=1}^m P_i = P$. 

We present three algorithms in the MPC model for the $k$-center problem with  outliers: 
a $2$-round deterministic algorithm and an $\rrounds$-round deterministic algorithm
in which $P$ can be distributed arbitrarily among the machines, and a
1-round randomized algorithm that assumes $P$ is distributed randomly. 
Our main result is the $2$-round algorithm explained next; 
other algorithms in the MPC model are given in Section~\ref{sec:more:MPC}.

%-----------------------------------------------------------------------------------------------------
\mypara{A deterministic $2$-round algorithm.}
% \label{subsec:2:round}
%-----------------------------------------------------------------------------------------------------
Our 2-round algorithm assumes that $P$ is distributed arbitrarily (but evenly) over the machines. 
Since the distribution is arbitrary, we do not have an upper bound on the number of outliers 
present at each machine. Hence, it seems hard to avoid sending $\Omega(z)$ points per machine to the coordinator. 
% as we had in Section~\ref{subsec:1:round}. 
% By an extra round of communication, we significantly improve the space of the deterministic algorithm in \cite{DBLP:journals/pvldb/CeccarelloPP19} that is $?$ to $O\left(\sqrt{n\gparameter}
% +\sqrt{\frac{n}{\gparameter}}\cdot\log(z+1) + z\right)$.
Next we present an elegant mechanism to guess the number of outliers present at each machine,
such that the total number of outlier candidates sent to the coordinator, over all machines, is~$O(z)$.
Our mechanism refines the method of Guha, Li and Zhang \cite{DBLP:journals/topc/GuhaLZ19}, and gives a significantly better dependency
on $z$ in the storage of the worker machines.
\medskip

In the first round of Algorithm \ref{alg:2:round}, each machine $M_i$ finds a $3$-approximation
of the optimal radius, for various numbers of outliers, and stores these radii in a vector~$V_i$.
The $3$-approximation is obtained by calling the algorithm \greedy of
Charikar~\etal~\cite{DBLP:conf/soda/CharikarKMN01}.
More precisely, $M_i$ calls \greedy$(P_i, k, 2^j-1)$ and stores the reported radius
(which is a $3$-approximation of the optimal radius for the $k$-center problem
with $2^j-1$ outliers on $P_i$) in~$V_i[j]$.
Then, each machine $M_i$ sends its vector $V_i$ to all other machines.
In the second round, all machines use the shared vectors to compute $\feasibleRadius$, 
which is is an approximate lower bound on the ``global'' optimal radius. Using $\feasibleRadius$,
each machine then computes a local mini-ball covering so that the total number of outliers 
over all machines is at most $2z$.

%------------------------------------------------------------------------------------------
\begin{algorithm}[h] 
\caption{A deterministic $2$-round algorithm to compute an $(\eps,k,z)$-coreset} 
\label{alg:2:round}
\textbf{Round 1, executed by each machine $M_i$:} %\\
%\emph{Computation:}
\begin{algorithmic}[1]
\State Let $V_i[0,1,\ldots, \ceil{\log(z+1)}]$ be a vector of size $\ceil{\log(z+1)}+1$.
\For{$j \gets 0$ \textbf{to} $\ceil{\log(z+1)}$}
    \State $V_i[j] \gets$ the radius of balls returned by  \greedy$(P_i, k , 2^j-1)$.
\EndFor
\State \emph{Communication round:} Send $V_i$ to all other machines. 
\end{algorithmic}
% \emph{Communication:}
% \begin{algorithmic}[1]
% \State Send $V_i$ to all other machines. 
% \end{algorithmic}
\vspace*{2mm}
\textbf{Round 2, executed by each machine $M_i$:} %\\
%\emph{Computation:}
\begin{algorithmic}[1]
\State Let $R \gets \{V_{\ell}[j] : 1\leq \ell \leq m \mbox{ and } 0\leq j \leq \ceil{\log(z+1)} \}$ 
      % be the set of all radii stored in the received vectors.
\State $\feasibleRadius \gets \min \left\{r \in R : \sum_{\ell=1}^{m} \left( 2^{\min\{j: V_{\ell}[j] \leq r\}} -1 \right)\leq 2z  \right\}$.
\State $\hat{j}_i \gets \min\{j : V_i[j] \leq \hat{r} \}$.
\State $P_i^* \gets$ \computeCoreset$(P_i, k, 2^{\hat{j}_i} - 1, \eps)$.
\State \emph{Communication round:}  Send $P_i^*$ to the coordinator.
\end{algorithmic}
%\emph{Communication:}
% \begin{algorithmic}[1]
%\State Send $P_i^*$ to the coordinator.
%\end{algorithmic}
%\vspace*{2mm}
\textbf{At the coordinator:}  Collect all mini-ball coverings~$P_i^*$ and report \computeCoreset$(\bigcup_i P_i^*,k,z,\eps)$ as the final mini-ball covering. 
\end{algorithm}
%------------------------------------------------------------------------------------------

% \mdb{Comments on the algorithm: Do we want a round 3,
% only executed at the coordinator, which reports the union of the received coresets?
% Otherwise the algorithm does not compute a coreset.}

%\mypara{Analysis.} 
First, we show that the parameter $\feasibleRadius$ that we computed
in the second round, can be used to obtain a lower bound on $\optkz(P)$.

%------------------------------------------------------------------------------------------
\begin{lemma} \label{lem:feasibleRadius}
Let $\feasibleRadius$ be the value computed in Round~2
of Algorithm~\ref{alg:2:round}.
Then, $\optkz(P) \geq \feasibleRadius/3$.
\end{lemma}
%------------------------------------------------------------------------------------------
\begin{proof}
Consider a fixed optimal solution for the $k$-center problem with $z$ outliers
on~$P$, and let $Z^*$ be the set of outliers in this optimal solution.
Let $z_i^* := |Z^* \cap P_i|$ be the number of outliers in~$P_i$. 
For each $i\in[m]$ we define $j_i^* := \ceil{\log(z_i^* + 1)}$, so that $2^{j_i^*-1}-1 < z_i^* \leq 2^{j_i^*}-1$.

First, we show that $\max_{i\in[m]} V_i[j_{i}^*] \leq 3\cdot \optkz(P)$. 
Let $i \in [m]$ be an arbitrary number.  
Since $z_i^* \leq 2^{j_i^*}-1$, we have $\opt_{k, 2^{j_i^*}-1}(P_i) \leq \opt_{k,z_i^*}(P_i)$.
Moreover, since $P_i \subseteq P$ and $z_i^*:= |Z^* \cap P_i|$, 
we have $\opt_{k,z_i^*} (P_i)\leq \opt_{k, z}(P)$.
Therefore,
$\opt_{k,2^{j_i^*}-1}(P_i) \leq \opt_{k,z_i^*} (P_i) \leq \optkz(P).$

% By definition of the vector~$V_i$ we also have $V_i[j_i^*] \leq 3\cdot \opt_{k, {2^{j_i^*}}-1}(P_i)$. 
Besides, $V_i[j_i^*]$ a $3$-approximation of the optimal radius for the $k$-center problem
with $2^{j_i^*}-1$ outliers on $P_i$.
Hence, $V_i[j_i^*] \leq 3\cdot\opt_{k, 2^{j_i^*}-1}(P_i) \leq 3\cdot \optkz(P) \enspace .$
The above inequality holds for any $i\in [m]$, so we have $\max_{i\in [m]} V_i[j_{i}^*] \leq 3\cdot \optkz(P)$.

Next, we show that $\feasibleRadius \leq \max_{i\in[m]} V_i[j_{i}^*]$. 
Let $\ell\in[m]$ be an arbitrary number. 
Since $V_\ell[j_\ell^*] \leq \max_{i\in[m]} V_{i}[j_{i}^*]$, 
we have $\min\{j:V_\ell[j]\leq \max_{i\in[m]} V_{i}[j_{i}^*] \} \leq j_\ell^*$.
Therefore,
\[
\sum_{\ell=1}^{m} \left( 2^{\min\{j: V_i[j] \leq \max_{{i}\in[m]} V_{i}[j_{i}^*] \}} -1 \right)
\leq \sum_{\ell=1}^m \left( 2^{j_\ell^*} -1 \right)
\leq \sum_{\ell=1}^m 2z_\ell^*
\leq 2z \enspace .
\]

Moreover, $\max_{i\in[m]} V_{i}[j_{i}^*] \in R$. 
So, we conclude $\feasibleRadius\leq \max_{i\in[m]} V_{i}[j_{i}^*]$. 
Putting everything together we have
$\feasibleRadius \leq \max_{i\in[m]} V_{i}[j_{i}^*] \leq 3\cdot \optkz(P)$,
which finishes the proof.
\end{proof}
%------------------------------------------------------------------------------------------

In the second round of Algorithm \ref{alg:2:round}, each machine $M_i$ sends an $(\eps, k, 2^{\hat{j}_i} - 1)$-mini-ball covering of $P_i$ to the coordinator.
As $\hat{j}_i$ may be less than $j_i^*$,
we cannot guarantee that $\opt_{k, 2^{\hat{j}_i} - 1}(P_i) \leq \optkz(P)$,
so we cannot immediately apply Lemma \ref{lem:union:coreset} to show that the union of mini-ball coverings that 
the coordinator receives is an $(\eps, k, z)$-mini-ball covering of~$P$. Therefore, we need a more careful analysis, 
which is presented in Lemma \ref{lem:2:round:alg}.

\begin{lemma}
\label{lem:2:round:alg}
Let $P_i^*$ be the weighted set that machine $M_i$ sends to the coordinator in the second round of Algorithm \ref{alg:2:round}. 
Then, $\cup^{m}_{i=1} P_i^*$ is an $(\eps, k, z)$-mini-ball covering of $P$.
\end{lemma}

\begin{proof}
% As an $(\eps, k, z)$-mini-ball covering of $P$ is an $(\eps, k, z)$-coreset of $P$ by Lemma \ref{lem:MBC:is:coreset}, it is enough to  
%  show that $\cup^{m}_{i=1} P_i^*$ is an $(\eps, k, z)$-mini-ball covering of $P$.
To show $\cup^{m}_{i=1} P_i^*$ is an $(\eps, k, z)$-mini-ball covering of $P$, 
 we prove that for each point $p \in P$ its representative point $q \in \cup^{m}_{i=1} P_i^*$ is such that $\dist(p, q) \leq \eps \cdot \optkz(P)$. 
Let $p$ be an arbitrary point in $P_i$, and let $q \in P_i^*$ be the representative point of $p$.
Observe that $P_i^*$ is a mini-ball covering returned by \computeCoreset$(P_i, k, 2^{\hat{j_i}}-1, \eps)$. Let $r_i$ be the radius of ball that \greedy$(P_i, k, 2^{\hat{j_i}}-1)$ returns, i.e. $r_i=V_i[\hat{j_i}]$. Note that \greedy~is a deterministic algorithm, and $\hat{j_i}$ is defined such that $r_i = V_i[\hat{j_i}] \leq \feasibleRadius$. When we invoke \computeCoreset$(P_i, k, 2^{\hat{j_i}}-1, \eps)$, first it invokes \greedy$(P_i, k, 2^{\hat{j_i}}-1)$, which returns balls of radius $r_i$, and next, assigns the points in each non-empty mini-ball of radius $\frac{\eps}{3}\cdot r_i$ to the center of that mini-ball. 
So, each point is assigned to a representative point of distance at most $\frac{\eps}{3}\cdot r_i$. Thus, $\dist(p, q) \leq \frac{\eps}{3}\cdot r_i$. According to Lemma \ref{lem:feasibleRadius}, $\feasibleRadius \leq 3\cdot \optkz(P)$, also $r_i \leq \feasibleRadius$. Putting everything together we have,
$\dist(p, q) \leq \frac{\eps}{3}\cdot r_i
\leq \frac{\eps}{3}\cdot \feasibleRadius
\leq \eps \cdot \optkz(P) \enspace . %\qedhere
$
\end{proof}

%------------------------------------------------------------------------------------------
We obtain the following result.
% ; a formal proof is in the appendix.
Note that the second term in the space bound, $\sqrt{n\eps^d/k}\cdot\log(z+1)$, 
can be simplified to $\sqrt{n}\cdot\log(z+1)$ since~$\eps^d/k<1$.
%------------------------------------------------------------------------------------------
\begin{theorem}[Deterministic $2$-round Algorithm]
\label{thm:2:round}
Let $P \subseteq X$ be a point set of size $n$ in a metric space $(X,\dist)$ of
doubling dimension~$d$. Let $k,z \in \mathbb{N}$ be two natural numbers, and
let $0 < \eps \le 1$ be an error parameter. 
Then, there exists a deterministic algorithm that computes an $(\eps, k, z)$-coreset of $P$
in the MPC model in two rounds of communication,
using  $m=O(\sqrt{n\eps^d/k})$ worker machines with 
$O(\sqrt{nk/\eps^d} + \sqrt{n\eps^d/k}\cdot\log(z+1))$ local memory,
and a coordinator with $O(\sqrt{nk/\eps^d} + \sqrt{n\eps^d/k}\cdot\log(z+1) + z)$ local memory. 
\end{theorem}
%------------------------------------------------------------------------------------------

\begin{proof}
Invoking Algorithm \ref{alg:2:round}, the coordinator receives
$\cup^{m}_{i=1} P_i^*$ after the second round, which is 
an $(\eps, k, z)$-mini-ball covering of $P$ by Lemma \ref{lem:2:round:alg}.
Then to reduce the size of the final coreset, the coordinator computes an $(\eps, k, z)$-mini-ball covering of $\cup^{m}_{i=1} P_i^*$, which is an $(\eps', k, z)$-mini-ball covering of $P$ by Lemma \ref{lem:merge:coreset}, and therefore an $(\eps', k, z)$-coreset of $P$ by Lemma \ref{lem:MBC:is:coreset}, where $\eps'=3\eps$.
Now, we discuss storage usage. 
In the first round, each worker machine needs $O(\frac{n}{m}) = O(\sqrt{nk/\eps^d})$ space to store the points and compute a mini-ball covering.
In the second round, each worker machine receives $m$ vectors of length $\ceil{\log(z+1)}+1$, and needs $O(m \cdot\log(z+1))$ to store them.
Therefore, the local space of each worker machine is of size 
$O(\sqrt{nk/\eps^d} +
\sqrt{n\eps^d/k}\cdot\log{(z+1)})$.

After the second round, the coordinator receives $\cup^{m}_{i=1} P_i^*$. As $P_i^*$ is returned by \computeCoreset$(P_i, k, 2^{\hat{j_i}}-1, \eps)$, Lemma \ref{lem:compress} shows that the size of $P_i^*$ is at most $k(\frac{12}{\eps})^d + (2^{\hat{j_i}}-1)$.
Besides, $\hat{j_i}$ is define such that $\sum_{i=1}^m (2^{\hat{j_i}}-1) \leq 2\cdot z$. 
Also, note that we can assume the doubling dimension $d$ is a constant.
Consequently, the required memory for the final mini-ball covering is
\[
\sum_{i=1}^m k\left(\frac{12}{\eps}\right)^d + (2^{\hat{j_i}}-1) 
= O\left(m \cdot k\left(\frac{1}{\eps}\right)^d + \sum_{i=1}^m (2^{\hat{j_i}}-1)\right)
% = O\left(\sqrt{\frac{n}{\gparameter}} \cdot \gparameter + z\right) 
= O\left(\sqrt{\frac{nk}{\eps^d}}+z\right) \enspace .
\]
Thus, the local memory of the coordinator is of size 
$O(\sqrt{nk/\eps^d} + \sqrt{n\eps^d/k}\cdot\log{(z+1)} + z)$.
\end{proof}
%------------------------------------------------------------------------------------------
%

%------------------------------------------------------------------------------------------
\section{A tight lower bound for insertion-only streaming algorithms}
\label{sec:streaming:lb}
%------------------------------------------------------------------------------------------
%\mypara{Lower bounds on the coreset size.}

In this section, we first show that any deterministic algorithm requires $\Omega(k/\eps^d+z)$ space to compute an $(\eps, k, z)$-coreset.
Then interestingly, we present a
deterministic streaming algorithm that uses $O(k/\eps^d+ z)$ space in section \ref{sec:tight:streaming:algorithm}, which is optimal.

To prove our lower bounds, we need to put a natural restriction on the
total weight of the coreset, as follows.
\mypara{Lower-bound setting.} 
Let $P(t)$ be the subset of points that are present at time~$t$, that is, $P(t)$ contains the points
that have been inserted.
% and, for the sliding-window model, have not yet expired.
Let $P^*(t) \subseteq P(t)$ be an $(\eps,k,z)$-coreset for~$P(t)$. Then we say that
$P^*(t)$ is a \emph{weight-restricted coreset} if
$w(P^*(t)) \leq w(P(t))$, that is, if the total weight of the
points in~$P^*(t)$ is upper bounded by the total weight of the points in~$P(t)$.
%\nw{We also have the assumption that if $p\not\in P^*(t)$ for some point $p\in P(t)$, then $p\not\in P^*(t')$ for any $t'>t$ (unless $p$ is re-inserted).}
%------------------------------------------------------------------------------------------
\begin{theorem}[Lower bound for insertion-only algorithms]
\label{thm:lower:bound:insertion:only}
Let $0 < \eps \leq \frac{1}{8d}$ and $k\geq 2d$. 
Any deterministic insertion-only streaming algorithm that maintains a weight-restricted $(\eps,k,z)$-coreset 
for the $k$-center problem with $z$ outliers in $\Reals^d$  
must use $\Omega(k/\eps^d+ z)$ space.
% Then, there exists an instance for this problem that requires 
\end{theorem}
% %-----------------------------------------------------------------------------

To prove Theorem~\ref{thm:lower:bound:insertion:only}, 
we consider two cases: $z\leq k/\eps^d$ and $z> k/\eps^d$. 
For the former cases, we show an $\Omega(k/\eps^d)$ lower bound in section \ref{sec:lb:z:small}.
Then for the latter case, we prove an $\Omega(z)$ lower bound in section \ref{sec:lb:k+z}, which also applies to randomized streaming algorithms.

%-----------------------------------------------------------------------------

%-----------------------------------------------------------------------------
\subsection{An $\Omega(k/\eps^d)$ lower bound}
\label{sec:lb:z:small}
%-----------------------------------------------------------------------------

The following lemma provides a good lower bound for the case where $z\leq k/\eps^d$.
%-----------------------------------------------------------------------------
\begin{lemma}
\label{lem:lb:z:small}
Let $0 < \eps \leq \frac{1}{8d}$ and $k\geq 2d$. 
% Let $z\leq k/\eps$. 
Any deterministic insertion-only streaming algorithm that maintains 
% a weight-restricted
an $(\eps,k,z)$-coreset 
for the $k$-center problem with $z$ outliers in $\Reals^d$ 
needs to use $\Omega(k/\eps^d)$ space.
% \lb{We do not need weight-restricted property or $z\leq k/\eps$ for this lemma, I removed them.}
\end{lemma}
%-----------------------------------------------------------------------------

% \begin{proof}
To prove the lemma, we may assume without loss of generality that~$\lambda := 1/(4d\eps)$ is an integer. 
Let $h := d(\lambda+2)/2$ and $r :=\sqrt{h^2-2h+d}$.
% Recall that  $P(t)\subset \mathbb{R}^d$ denotes the point set present at at time~$t$. 
We next present a set~$P(t)$ requiring a coreset of size~$\Omega(k/\eps^d)$.
The set $P(t)$ is illustrated in Figure~\ref{fig:lb:streaming:config}. It contains
$z$ outlier points $o_1,\ldots,o_z$ and $k-2d+1$ clusters $C_1,\ldots,C_{k-2d+1}$, defined
as follows.
\begin{itemize}
\item For $i \in [z]$, the outlier $o_i$ is a point with the coordinates $(- 4(h+r)i,0,0,\ldots,0)$.  
\item Each cluster $C_i$ is a $d$-dimensional integer grid of side length $\lambda$ that 
      consists of $(\lambda+1)^d$ points. The distance between two consecutive clusters 
      is $4(h+r)$ as illustrated in Figure~\ref{fig:lb:streaming:config}.
      In particular, $C_1 := \{(x_1,\dots,x_d) \ | \ x_j \in \{0,1,\cdots,\lambda\} \}$.
      For each $1 < i \le k-2d+1$, the cluster $C_i$ is  
      $C_i := \{(\delta+x_1,x_2,\dots,x_d) \ | \ (x_1,x_2,\dots,x_d) \in C_{i-1} \}$, 
      where $\delta=\lambda+4(h+r)$.
\end{itemize}
Let $P^*(t)\subseteq P(t)$ be the coreset that the algorithm maintains at time~$t$. 
We claim that $P^*(t)$ must contain all points of any of the clusters $C_1,\ldots,C_{k-2d+1}$. 
Since $|C_i| = (\lambda+1)^d = \Omega(1/\eps^d)$, we must then have $|P^*(t)|=\Omega(k/\eps^d)$.

To prove the claim, assume for a contradiction that
there is a point~$p^*=(p_1^*,\ldots,p_d^*)$ that is not explicitly stored in~$P^*(t)$.
Let $i^* \in [k-2d+1]$ be such that $p^*\in C_{i^*}$. %Now suppose that at the 
Now suppose the next $2d$ points that arrive 
are the points from $P^+ := \{p_1^+,\ldots,p_d^+ \}$ and $P^- := \{p_1^-,\ldots,p_d^- \}$.
Here $p_j^+ = (p_{j,1}^+,\ldots,p_{j,d}^+)$, where
$p_{j,j}^+:= p_j^*+(h+r)$ and $p_{j,\ell}^+:=p_\ell^*$ for all $\ell\neq j$.
Similarly, $p_j^- = (p_{j,1}^-,\ldots,p_{j,d}^-)$ where $p_{j,j}^-:= p_j^*-(h+r)$ and $p_{j,\ell}^-:=p^*_\ell$ 
for all $\ell\neq j$; see Figure~\ref{fig:lb:streaming:config}. 
It will be convenient to assume that each point in $P^+\cup P^-$ has weight~2;
of course we could also insert two points at the same location (or,
almost at the same location). 
% The weight of each point in $P^+\cup P^-$ is $2$.
% We consider these points weighted. 
% In the proof of Lemma \ref{claim:OPT}, we discuss how to change them to unweighted points.

\begin{figure}[h]
\begin{center}
% \hspace*{-3mm}
\includegraphics{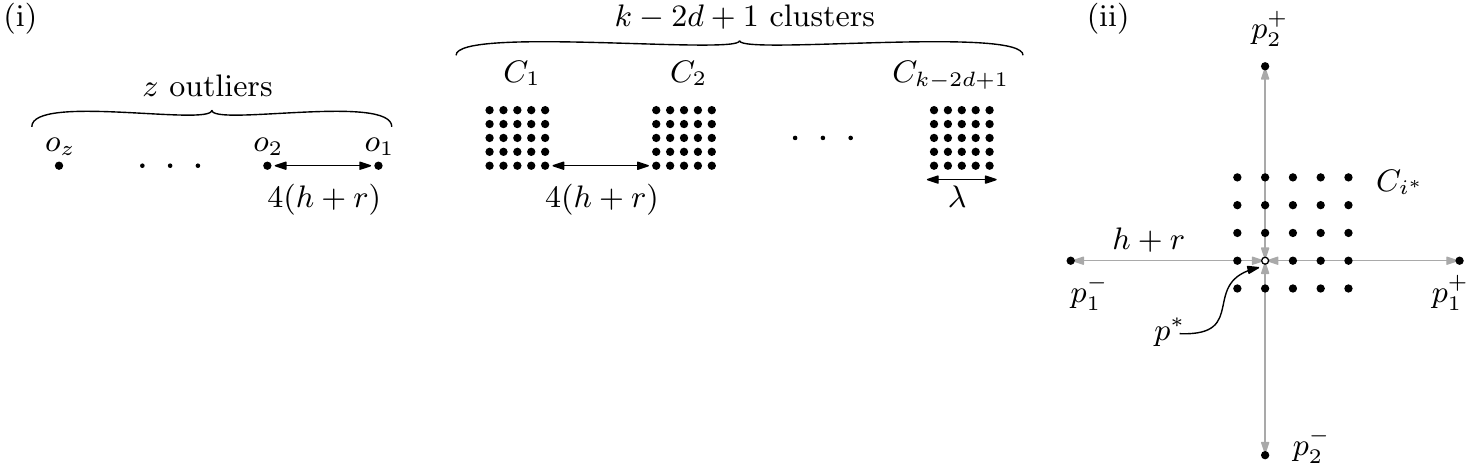}
\end{center}
\caption{Illustration of the lower bound in
Lemma \ref{lem:lb:z:small}.
% Theorem~\ref{thm:lower:bound:insertion:only}. 
We have $\lambda := 1/(4d\eps)$ is an integer, $h := d(\lambda+2)/2$ and $r :=\sqrt{h^2-2h+d}$.
Part~(i) shows the global construction, part~(ii) shows the points in $P^+$ and $P^-$. 
%Here, $P^*(t')$ underestimates $\optkz(P(t'))$ since $2d$ balls of radius $r$ can cover $C_{i^*}\cup P^+\cup P^-$ (dashed balls), 
%and then $\optkz(P^*(t')) \leq r$.
%However, $\optkz(P(t')) = (r+h)/2$ (the red ball).
} 
\label{fig:lb:streaming:config}
\end{figure}
%-----------------------------------------------------------------------------

Let $P(t') := P(t)\cup P^- \cup P^+$ and let $P^*(t')$ be the coreset of $P(t')$. 
Since $P^*(t)$ did not store $p^*$, we have $p^*\not\in P^*(t')$. We will show that this
implies that $P^*(t')$ underestimates the optimal radius by too much. We first give a
lower bound on $\optkz(P(t'))$.
\begin{claim}
\label{lem:OPT:lower:bound}
$\optkz(P(t')) \geq  (h+r)/2$.
\end{claim}
\begin{proof}
Recall that we have $k-2d+1$ clusters $C_1,\ldots,C_{k-2d+1}$ and that
$p^* \in C_{i^*}$. Pick an arbitrary point from each cluster $C_i \neq C_{i^*}$,
and let $Q$ be the resulting set of $k-2d$ points.
Define $X := Q \cup \{p^*\} \cup P^- \cup P^+ \cup \{o_1,\ldots,o_z\}$. 
Observe that $|X| = (k-2d) + 1 + 2d+z = k+z+1$, and that 
the pairwise distance between any two points in $X$ is at least~$h+r$.
Hence, $\optkz(P(t')) \geq \optkz(X) \geq (h+r)/2$.
\end{proof}
Next we show that, because $P^*(t')$ does not contain the point~$p^*$, it
must underestimate $\optkz(P^*(t'))$ by too much. To this end, we first show the following claim,
which is proved as Lemma~\ref{lem:OPT:upper:bound} in the appendix.
The idea of the proof is that an optimal solution for $P^*(t')$
can use $2d$ balls for $C_{i^*}\cup P^+\cup P^-$, and that because
$p^*\not\in P^*(t')$ this can be done with balls that are ``too small'' for an $(\eps,k,z)$-coreset; see Figure \ref{fig:new:lb:streaming:missing}.
%\lb{I don't like the last statement}
%\lb{``too small''? It is just $\eps$ fraction smaller}
%\mdb{Yes, but that is too small for a coreset. But if it's confusing then maybe we should reformulate.}
The formal proof is given in the appendix.
% \begin{claiminproof} (Lemma~\ref{lem:OPT:upper:bound} in Appendix~\ref{app:lb:streaming})
\begin{claim}[Lemma~\ref{lem:OPT:upper:bound} in Appendix~\ref{app:lb:streaming}]
\label{claim:lb:streaming}
$\optkz(P^*(t')) \leq r$.
\end{claim}
% \end{claiminproof}

\begin{figure}
\begin{center}
\hspace*{-3mm}
\includegraphics{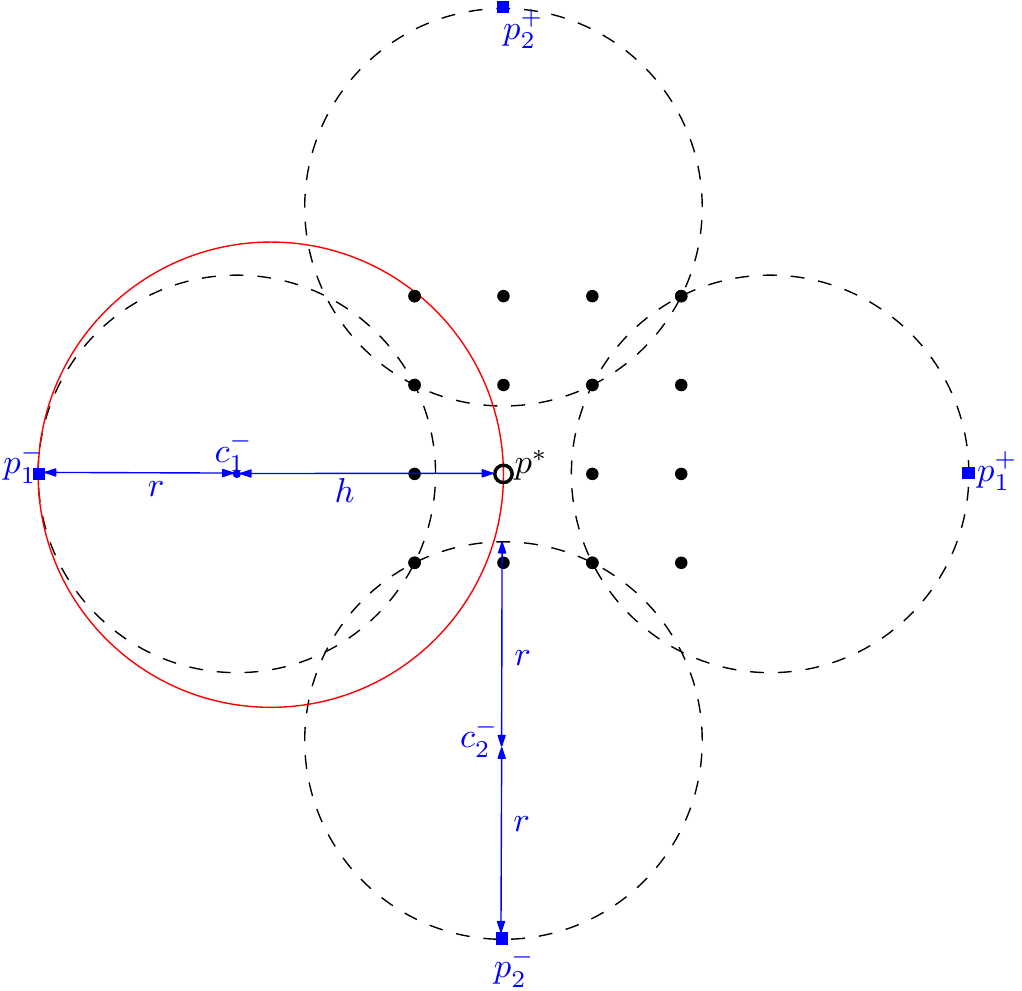}
\end{center}
\caption{Illustration of the lower bound for the streaming model. 
% We have $\lambda := 1/(4d\eps)$ is an integer, $h := d(\lambda+2)/2$ and $r :=\sqrt{h^2-2h+d}$.
Here, $P^*(t')$ underestimates $\optkz(P(t'))$ since $2d$ balls of radius $r$ can cover
$P^+\cup P^-\cup C_{i^*}\setminus\{p^*\}$ (dashed balls), 
and then $\optkz(P^*(t')) \leq r$.
However, $\optkz(P(t')) = (r+h)/2$ (the red ball).
} 
\label{fig:new:lb:streaming:missing}
\end{figure}

Lemma~\ref{lem:r:bound}, which can also be found in Appendix~\ref{app:lb:streaming},
gives us that $r<(1-\eps)(r+h)/2$. Putting everything together, we have
\[
(1-\eps)\cdot\optkz(P(t')) \ \ \geq \ \  (1-\eps)(r+h)/2 \ \ > \ \ r \ \ \geq \ \   \optkz(P^*(t')) \enspace.
\]
However, this is a contradiction to our assumption that $P^*(t')$ is an $(\eps,k,z)$-coreset of $P(t')$.
Hence, if $P^*(t)$ does not store all points from each of the clusters~$C_i$,
then it will not be able to maintain an $(\eps,k,z)$-coreset.
This finishes the proof of Lemma \ref{lem:lb:z:small}.
% \end{proof}
%------------------------------------------------------------------------------------------

%-----------------------------------------------------------------------------
\subsection{An $\Omega(z)$ lower bound}
\label{sec:lb:k+z}
%-----------------------------------------------------------------------------

% The next lemma deals with the former case, presenting an instance
% that requires a coreset of size~$\Omega(z)$. Its proof,
% which also applies to randomized streaming algorithms, is given in Appendix~\ref{app:lb:streaming}.

Now we provide an $\Omega(z)$ lower bound in Lemma \ref{lem:lb:k+z}. Note that the proof also applies to randomized streaming algorithms.

% We show an $\Omega (k+z)$ lower bound for the used in the following lemma.

\begin{lemma}
\label{lem:lb:k+z}
Let $0 < \eps < 1$ and $k\geq 1$.
Any streaming (deterministic or randomized) algorithm that maintains a  weight-restricted $(\eps,k,z)$-coreset 
for the $k$-center problem with $z$ outliers in $\Reals^1$  
must use $\Omega(k+ z)$ space.
\end{lemma}

%-----------------------------------------------------------------------------
\begin{figure}[h]
\begin{center}
\includegraphics[scale=0.8]{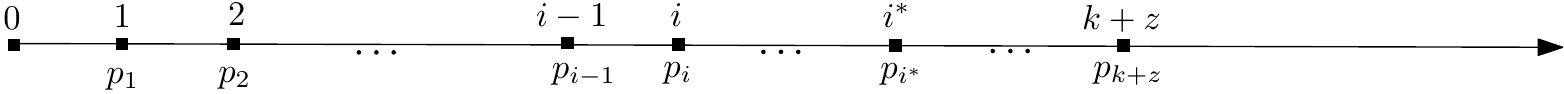}
\end{center}
\caption{Illustration of the lower-bound construction of Lemma~\ref{lem:lb:k+z} for $\mathbb{R}^1$.}
\label{fig:lb:one:dim}
\end{figure}
%-----------------------------------------------------------------------------
\begin{proof}
Let $t;=k+z$ of and let $P(t) = \{p_1,\ldots,p_{k+z}\}$ be the set of points that 
are inserted up to time $t$. 
Our lower-bound instance is a one-dimensional point set (i.e., points are on a line), 
where the value ($x$-coordinate) of the points in $P(t)$ is equal to their index. 
That is, for $i \in [k+z]$, we have $p_i = i$. See Figure~\ref{fig:lb:one:dim}. 

Consider a streaming algorithm that maintains an $(\eps,k,z)$-coreset and let 
$P^*(t)\subseteq P(t)$ be its coreset at time $t$. 
We claim that $P^*(t)$ must contain all points $p_1,\ldots,p_{k+z}$. 
% In other words, none of
% the points~$p_i$ can be charged to any other point $p_j \in P^*(t)$ 
% where $p_j$ is a weighted point in the coreset $P^*$ that contains one unit weight corresponding to 
% $p_i$. That is, we claim that $|P^*(t)|=k+z$.

To prove the claim, we assume for the sake of the contradiction that 
there is a point~$p_{i^*}$ that is not explicitly stored in~$P^*(t)$. 
Suppose at time $t+1 = k+z+1$, the next point $p_{k+z+1} = k+z+1$ arrives. 
Observe that $P(t+1)$ consists of $k+z+1$ points $P(t+1) = \{p_1,\ldots,p_{k+z}, p_{k+z+1}\}$ 
at unit distance from each other. 
Thus, one of the clusters in an optimal solution of $P(t+1)$ will contain two
points. Hence $\optkz(P(t+1)) = 1/2$.

Next, we prove that $\optkz(P^*(t+1)) = 0$. 
Suppose for the moment that this claim is correct. 
Then, $\optkz(P(t+1)) = 1/2$ and $\optkz(P^*(t+1)) = 0$, 
which contradicts that $P^*(t+1)$ is an $(\eps,k,z)$-coreset.
That is, all points $p_1,\ldots,p_{k+z}$ must be in $P^*(t)$. 
Therefore, any streaming algorithm that can $c$-approximate 
$\optkz(P(t+1))$ for $c>0$,  must maintain a coreset whose size is $\Omega(k+z)$ at time $t$.

It remains to prove that $\optkz(P^*(t+1)) = 0$. 
First of all, observe that since $p_{i^*}\notin P^*(t)$, the followup coresets do not know 
about the existence of $p_{i^*}$, therefore, $p_{i^*}$ will not be added to such coresets. 
Therefore, $|P^*(t+1)| \leq k+z$. We consider two cases. 

Case $1$ is if $|P^*(t+1)| \leq k$. In this case, we put a center on each of the points in $P^*(t+1)$ and so $\optkz(P^*(t+1)) = 0$.

Case $2$ occurs when $k < |P^*(t+1)| \leq k+z$. 
% Let $Q \subseteq P^*(t+1)$ be the set of $k$ points that have the maximum total weight. 
Let $Q\subseteq P^*(t+1)$ be the set of $k$ points of largest weight, with ties broken arbitrarily. 
That is, $Q = \arg\max_{Q' \subset P^*(t+1): |Q'| = k} w(Q')$, where $w(Q') = \sum_{q \in Q'}w(q)$.  

\begin{claiminproof}
\label{clm:set:q:zero}
The total weight of $P^*(t+1)\setminus Q$ is at most $z$. 
\end{claiminproof}

\begin{proofinproof}
Note that the weight of every point of a coreset is a positive integer. 
Since $Q$ contains the $k$ points of largest weight from $P^*(t+1)$ and $|P^*(t+1)|\leq k+z$, the total weight of $P^*(t+1)\setminus Q$ is at most a $z/(k+z)$ fraction of the total weight of  $P^*(t+1)$. Hence,
\[
w(P^*(t+1)\setminus Q) \leq \frac{z}{k+z} w(P^*(t+1)) \leq \frac{z}{k+z} w(P(t+1)) = \frac{z}{k+z} \cdot (k+z+1) < z + 1.
\]
Since all weights are integers, we can conclude that $w(P^*(t+1)\setminus Q) \leq z$.

% We consider two cases: $\weight(Q) \leq k$ and $\weight(Q) > k$. 
% Let us consider the first case. 
% The set $Q$ is a subset of $P^*(t+1)$ of size $k$ that has the maximum total weight of $Q$. 
% However, the first case says that $\weight(Q) \leq k$. 
% This essentially means that the weight of every point in $Q$ is $1$. 
% Thus, $\weight(P^*(t+1)\setminus Q) = |P^*(t+1)\setminus Q| \leq (k+z)-k = z$. 

% The second case says that the total weight of $Q$ is at least $k+1$. 
% Now, since we assume $P^*(t+1)$ is a weight-restricted coreset, we have $\weight(P^*(t+1)) \leq w(P(t+1)) = |P(t+1)|$, 
% and then the total weight of $P^*(t+1)\setminus Q$ is at most $(k+z+1) - (k+1) = z$.
\end{proofinproof}

Now, since $w(P^*(t+1)\setminus Q) \leq z$ and $|Q|=k$, putting a center on each point from $Q$ gives $\optkz(P^*(t+1)) = 0$. 
This finishes the proof of this lemma. 
\end{proof}
%-----------------------------------------------------------------------------
%------------------------------------------------------------------------------------------
\subsection{A space-optimal streaming algorithm}
\label{sec:tight:streaming:algorithm}
%------------------------------------------------------------------------------------------

% \mdb{Or do we just want to call it "A streaming algorithm"? Insertion-only is 
% the standard (I think), and adding it to the title sounds like we are doing a special case.}
Let $P \subseteq X$ be a point set of size $n$ in a metric space $(X,\dist)$ of doubling dimension~$d$. 
% Let $k,z \in \mathbb{N}$ be two natural numbers, and let $0 < \eps \le 1$ be an error parameter. 
% \mdb{Perhaps the previous sentence is not needed. This setting is the same throughout the paper.}
In the streaming model, the points of $P$ arrive sequentially. We denote the set of points
that have arrived up to and including time~$t$ by $P(t)$. In this section, we present a 
deterministic $1$-pass streaming algorithm to maintain an $(\eps, k, z)$-coreset for $P(t)$.
Interestingly, our algorithm use $O(\frac{k}{\eps^d}+z)$ space, which is optimal.

\medskip
In Algorithm \ref{alg:streaming}, we maintain a variable $r$ that is a lower bound 
for the radius of an optimal solution, and a weighted point set~$P^*$ that is 
an $(\eps, k, z)$-mini-ball covering of $P(t)$. 
% \mdb{Do we want to use $r(t)$ and $P^*(t)$, or is it better to omit the parameter~$t$?}
% \lb{$r$ and $P^*$ are in the algorithm and it is clear that they are a function of $t$, but $P$ does not appear in the algorithm and we may need to clarify that we mean the current $P$, not the final $P$.}
When a new point $p_t$ arrives at time $t$, we assign it to a representative point in $P^*$ 
within distance $(\eps/2)\cdot r$, or add $p_t$ to $P^*$ if there is no such nearby 
representative. 
% The initial values of $r$ is $0$, and the first time that the size of $P^*$ become $k+z+1$, we increase it to a non-zero value.
We have to be careful, however, that the size of $P^*$ does not increase too much.
To this end, we need to update~$r$ in an appropriate way, and then
update $P^*$ (so that it works with the new, larger value of~$r$)
whenever the size of $P^*$ reaches a threshold.
%$k(\frac{16}{\eps})^d+z$.
But this may lead to another problem: if a point $p$ is first assigned to 
some representative point~$q$, and later $q$ (and, hence, $p$)
is assigned to another point, then the distance between $p$ and its
representative may increase. (In other words, the ``errors'' that we incur
because we work with representatives may accumulate.)
We overcome this problem by doubling the value of $r$ whenever we update it.
% \mdb{Perhaps I would replace the next two sentences by:}
Lemmas~\ref{lem:rho} and~\ref{lem:streaming:coreset} show that this
keeps the error as well as the size of the mini-ball covering under control.
Our algorithm is similar to the streaming algorithm by Ceccarello \etal~\cite{DBLP:journals/pvldb/CeccarelloPP19}, however, by using a more clever threshold for the size of $P^*$ we improve the space significantly.
% Then, the distance between each point and its representative point is at most
% $O(\eps\cdot\sum_{i=0}^{\infty} \frac{r}{2^i})$
% = $O(\eps\cdot 2 r)$
% = $O(\eps)\cdot\optkz(P(t))$.
% The detailed proofs are presented in Lemma \ref{lem:rho} and \ref{lem:streaming:coreset}.

%------------------------------------------------------------------------------------------
\begin{algorithm}[h] 
\caption{\streaming} 
\label{alg:streaming}
\textbf{Initialization:}
\begin{algorithmic}[1]
\State $r \gets 0$ and $P^* \gets \emptyset$.
\end{algorithmic}
\vspace*{1mm}
\textbf{HandleArrival$(p_t)$}
\begin{algorithmic}[1]
    \If{there is $q \in P^*$ such that $\dist(p_t,q) \leq \frac{\eps}{2}\cdot r$}
        \State $\weight(q) \gets \weight(q)+1$. \label{line:streaming:rep:q} \Comment{$q$ is the representative of $p_t$ now}
        % \mdb{I wonder if we should ``track'' the representatives explicitly, in a comment in the algorithm. So here we would add the comment: (We now have $\myrep(p) = q$)  We would then have to write something similar elsewhere also.}
    \Else
        \State Add $p_t$ to $P^*$. \label{line:streaming:rep:itself}
    \EndIf
    
    \If{$r = 0$ and $|P^*| \geq k+z+1$}
        \State Let $\Delta$ be the minimum distance between any two (distinct) points in $P^*$.
        \State $r\gets \Delta/2$.
    \EndIf

    \While{$|P^*| \geq k(\frac{16}{\eps})^d+z$} 
    \label{line:streaming:threshold}
    % \mdb{$\geq$ or $>$?} \lb{Both work}
        \State $r \gets 2\cdot r$. \label{line:streaming:doubling:r}
        \State $P^* \gets $ \UpdateCoreset$(P^*, \frac{\eps}{2}\cdot r)$.
    \EndWhile
\end{algorithmic}
\vspace*{1mm}
\textbf{Report coreset:}
\begin{algorithmic}[1]
\State \Return $P^*$.
\end{algorithmic}
\end{algorithm}
%------------------------------------------------------------------------------------------

%------------------------------------------------------------------------------------------
\begin{algorithm}[h] 
\caption{\UpdateCoreset$(Q, \delta)$} 
\label{alg:update:coreset}
\begin{algorithmic}[1]
\State Let $Q^* = \emptyset$.
\While{$|Q| > 0 $} 
    \State Take an arbitrary point $q \in  Q$ and let $R_q = B(q,\delta) \cap Q$. 
    \State Add $q$ to $Q^*$ with weight $\weight(q) := \weight(R_q)$. 
    \State $Q \gets Q \setminus R_q $. 
\EndWhile
\State \Return $Q^*$.
\end{algorithmic}
\end{algorithm}
%------------------------------------------------------------------------------------------
\medskip
We need the following lemma to prove the correctness of our algorithm.
Its proof is in the appendix.
%------------------------------------------------------------------------------------------
\begin{lemma}
\label{lem:rho}
After the point $p_t$ arriving at time~$t$ has been handled, we have:
for each point $p \in P(t)$
there is a representative point $q\in P^*$ such that $\dist(p, q) \leq \eps \cdot r$. 
\end{lemma}

%------------------------------------------------------------------------------------------

%------------------------------------------------------------------------------------------
Now we can prove that after handling~$p_t$ at time~$t$, the set $P^*$ is an $(\eps, k, z)$-coreset of 
the points that have arrived until time $t$.
%------------------------------------------------------------------------------------------
\begin{lemma}
\label{lem:streaming:coreset}
The set $P^*$ maintained by Algorithm \ref{alg:streaming} is an $(\eps, k, z)$-coreset of $P(t)$ and its size is at most
$k(\frac{16}{\eps})^{d} + z$.
\end{lemma}
%------------------------------------------------------------------------------------------
\begin{proof}
Recall that the algorithm maintains a value $r$ that serves as an estimate of the radius
of an optimal solution for the current point set~$P(t)$. To prove $P^*$ is an $(\eps, k, z)$-coreset 
of $P(t)$, we first show that $r \leq \optkz(P(t))$.

We trivially have $r \leq \optkz(P(t))$ after the initialization, since then~$r=0$.
The value of $r$ remains zero until $|P^*| \geq k+z+1$. At this time, we increase $r$ to $\Delta/2$,
where $\Delta$ is the minimum distance between any two points in $P^*$.
Since no two points in $P^*$ will coincide by construction, we have $\Delta > 0$.
% Note that all points in $P^*$ are disjoint and then $l>0$. 
Consider an optimal solution for~$P(t)$.
As $|P^*| \geq k+z+1$ and $P^* \subseteq P(t)$, and we allow at most~$z$ outliers,
there are at least two points in $P^*$ that are covered by the same ball in the optimal
solution. This ball has radius~$\optkz(P(t))$. Thus, $\Delta/2 \leq \optkz(P(t))$ and 
we have $r \leq \optkz(P(t))$.

Now suppose we update the value of $r$ to $2\cdot r$. This happens when
$|P^*| \geq k(\frac{16}{\eps})^d+z$. The distance between any 
two points in $P^*$ is more than $\delta = \frac{\eps}{2}\cdot r$, 
because we only add point to $P^*$ when its distance to all
existing points in $P^*$ is more than $\frac{\eps}{2}\cdot r$.
Note that Lemma \ref{lem:coreset:size} implies that 
$|P^*| \leq k \cdot ( 4\cdot\optkz(P(t))/\delta)^d+z$. 
Putting everything together we have
$$
k \left( \frac{16}{\eps} \right)^d+z
\leq |P^*| 
\leq k \left( 4\cdot\frac{\optkz(P(t))}{\delta} \right)^d+z
= k \left( 4\cdot\frac{\optkz(P(t))}{(\eps/2)\cdot r} \right)^d+z
\enspace ,
$$
which implies $\frac{16}{\eps} \leq 4\cdot \frac{\optkz(P(t))}{(\eps/2)\cdot r}$.
Hence, $2\cdot r \leq \optkz(P(t))$ holds before we update the value of $r$ to $2\cdot r$.
We conclude that $r \leq \optkz(P(t))$ always holds, as claimed.
\medskip

Lemma \ref{lem:rho} states that for any point $p\in P(t)$,
there is a representative point $q\in P^*$ such that 
$\dist(p,q) \leq \eps \cdot r$. 
Therefore,
$
\dist(p,q) 
\leq \eps \cdot r
\leq \eps \cdot\optkz(P(t)) \enspace .
$
Thus, for each point $p\in P(t)$, there is a representative point $q\in P^*$
such that $\dist(p,q)\leq \eps \cdot\optkz(P(t))$. This means that $P^*$ is an 
$(\eps, k, z)$-mini-ball covering of $P(t)$, which is an $(\eps, k, z)$-coreset of $P(t)$
by Lemma \ref{lem:MBC:is:coreset}. It remains to observe
that the size of $P^*$ is at most $k(\frac{16}{\eps})^d+z$ by the while-loop in lines \ref{line:streaming:threshold} of the algorithm.
\end{proof}
%------------------------------------------------------------------------------------------
Since we consider the doubling dimension $d$ to be a constant, Algorithm \ref{alg:streaming} requires 
$O\left(\frac{k}{\eps^d} + z\right)$ memory to  maintain an $(\eps, k, z)$-coreset.
We summarize our result in the following theorem.

%------------------------------------------------------------------------------------------
\begin{theorem}[Streaming Algorithm]
\label{thm:streaming}
Let $P$ be a stream of points from a metric space $(X,\dist)$ of
doubling dimension~$d$. Let $k,z \in \mathbb{N}$ be two natural numbers, and
let $0 < \eps \le 1$ be an error parameter. 
Then, there exists a deterministic $1$-pass streaming algorithm
that maintains an $(\eps, k, z)$-coreset  of $P$  for the $k$-center problem with $z$ outliers using $O\left(k/\eps^d+ z\right)$ storage.
\end{theorem}
%-----------------------------------------------------------------------------

%-----------------------------------------------------------------------------
\section{A fully dynamic streaming algorithm}
%-----------------------------------------------------------------------------
In this section, we develop a fully dynamic streaming algorithm that maintains 
an $(\eps,k,z)$-coreset for the $k$-center problem with $z$ outliers. Our algorithm
works when the stream consists of inserts and deletes of points
from a discrete Euclidean space~$[\Delta]^d$.
% \lb{Mention to the lower bound}

\subsection{The algorithm}

Our algorithm uses known sparse-recovery 
techniques~\cite{DBLP:journals/tcs/BarkayPS15,DBLP:conf/soda/MonemizadehW10}, which we explain first.

%-----------------------------------------------------------------------------
\mypara{Estimating 0-norms and sparse recovery.}
%-----------------------------------------------------------------------------
Consider a stream of pairs $(a_j,\xi_j)$, where $a_j\in [U]$ (for some
universe size~$U$) and $\xi_j\in \Integers$. 
If $\xi_j>0$ then the arrival of a pair $(a_j,\xi_j)$ can be interpreted as increasing the
frequency of the element $a_j$ by $\xi_j$, and if
$\xi_j<0$ then it can be interpreted as decreasing the frequency of~$a_j$ by $|\xi_j|$.
Thus the arrival of $(a_j,\xi_j)$ amounts to updating the frequency vector~$F[0..U-1]$
of the elements in the universe, by setting $F[a_j] \gets F[a_j]+\xi_j$.
We are interested in the case where $F[j]\geq 0$ at all times---this is called the strict turnstile model---and
where either $\xi_j=+1$ (corresponding to an insertion) or 
$\xi_j=-1$ (corresponding to a deletion). For convenience, we will limit
our discussion of the tools that we use to this setting.

Let $\zeronorm{F} := \sum_{j\in[U]} |F[j]|^0$ denote ``0-norm'' of $F$,
that is, $\zeronorm{F}$ is the number of elements with non-zero frequency.
We need the following result on estimating~$\zeronorm{F}$ in data streams.
%-----------------------------------------------------------------------------
\begin{lemma}[$\zeronorm{F}$-estimator \cite{DBLP:conf/pods/KaneNW10}]
\label{lem:f0}
For any given error parameter $0< \eps<1$ and failure parameter $0<\delta <1$,
we can maintain a data structure that uses $O((1/\eps^{2}+\log U)\log(1/\delta))$ space
and that, with probability at least $1-\delta$, reports a $(1\pm\eps)$-approximation 
of~$\zeronorm{F}$.
\end{lemma}
%-----------------------------------------------------------------------------
Define $J^* := \{ (j,F[j]): j\in [U] \mbox{ and } F[j]\neq 0 \}$ to be
the set of elements with non-zero frequency.% with, for each such element, its exact frequency. %%%% Commented by Leyla
The next lemma allows us to sample a subset of the elements from~$J^*$.
Recall that a sample $S\subseteq J^*$ is called \emph{$t$-wise independent} if 
any subset of $t$ distinct elements from $J^*$ has the same probability to be in~$S$.
%-----------------------------------------------------------------------------
\begin{lemma}[$s$-sample recovery \cite{DBLP:journals/tcs/BarkayPS15}]
\label{lem:s-recovery}
Let $s$ be a given parameter indicating the desired sample size, and 
let $0< \delta <1$ be a given error parameter, where $s=\Omega(1/\delta)$. 
Then we can generate a $\Theta(\log (1/\delta))$-wise independent sample~$S\subseteq J^*$,
where $\min(s,|J^*|) \leq |S|\leq s'$ for some $s'=\Theta(s)$, with a randomized streaming algorithm 
that uses $O(s\log(s/\delta)\log^2 U)$ space. The success probability of the algorithm
is at least $1-\delta$ and the algorithm fails with probability at most $\delta$.
% in which it returns anything. 
\end{lemma}
%-----------------------------------------------------------------------------
Note that the sample $S$ not only provides us with a sample from the set
of elements with non-zero frequency, but for each element in the sample
we also get its exact frequency.

%-----------------------------------------------------------------------------
\mypara{Our algorithm.} 
%-----------------------------------------------------------------------------
Let $G_0,G_1, \cdots,G_{\lceil \log\Delta \rceil}$ be a collection of $\lceil \log\Delta\rceil$ grids 
imposed on the space $[\Delta]^d$, where cells in the grid $G_i$ have side length~$2^i$ (i.e., 
they are hypercube of size $2^i\times \cdots \times 2^i$). 
Note that $G_i$ has $\lceil \Delta^d/2^{2i} \rceil$ cells. 
In particular, the finest grid $G_0$ has $\Delta^d$ cells of side length one. 
Since our points come from the discrete space~$[\Delta]^d$,
which is common in the dynamic geometric streaming model~\cite{DBLP:conf/stoc/Indyk04,DBLP:conf/stoc/FrahlingS05}, 
each cell $c \in G_0$ contains at most one point.
Thus, the maximum number of distinct points that can be placed in $[\Delta]^d$ is $\Delta^d$.
%\mdb{Is this important? It seems nicer to me to still use $n$ in the bounds.}

Let $S$ be a stream of inserts and deletes of points to an underlying point set $P \subseteq [\Delta]^d$. 
Let  $i \in [\lceil\log\Delta \rceil ]$. 
For the grid $G_i$, we maintain  two sketches in parallel:
\begin{itemize}
    \item A $\ell_0$-frequency moment sketch $\F(G_i)$  (based on Lemma~\ref{lem:f0}) that approximates the number of non-empty cells 
    of the grid $G_i$.
    \item A $s$-sparse recovery sketch $\mathcal{S}(G_i)$ (based on Lemma~\ref{lem:s-recovery}) 
    that supports query and update operations.  In particular, a query of $\mathcal{S}(G_i)$, 
    returns a sample of $s$ non-empty cells of $G_i$ (if there are that many non-empty cells)
    with their exact number of points. 
    Upon the insertion or deletion of a point $q$, for every grid $G_i$ 
    we update the sketches $\mathcal{S}(G_i)$ by updating the cell of $G_i$ 
    that contains the point $q$. For our dynamic streaming algorithm, we let $s = k(4\sqrt{d}/\eps)^d + z)$. 
%    \mdb{Exact number?}
\end{itemize}

Let $P(t) \subseteq P$ be the set of points that are present at time~$t$,
that is, that have been inserted more often than they have been deleted.
%\mdb{Do we allow frequencies larger than 1?}
Using the sketches $\mathcal{S}(G_i)$ and $\F(G_i)$, we can obtain an $(\eps,k,z)$-coreset 
of $P(t)$. 
To this end, we first query the sketches $\F(G_i)$ 
for all $i \in [\lceil \log\Delta \rceil ]$ to compute the approximate number of non-empty cells in each grid.
We then find the grid $G_j$ of the smallest cell side length 
that has at most $s$ non-empty cells and query the sketch $\mathcal{S}(G_j)$ to 
extract the set $Q_j$ of non-empty cells of $G_j$. 
For every cell $c \in Q_j$, 
we choose the center of $c$ as the representative of $c$ and assign the number of points 
in $c$ as the weight of this representative. We claim that the set of weighted 
representatives of non-empty cells in $Q_j$ is an $(\eps,k,z)$-coreset, except that
the points in the coreset are not a subset of the original point set~$P$
(which is required by  Definition~\ref{def:coreset}) but centers of certain grid cells. 
We therefore call the reported coreset a \emph{relaxed coreset}.
This results in the following theorem, which is proven in more detail in the remainder of this section.
% in Appendix~\ref{appendix:analysis:Dynamic}. 

\begin{theorem}
\label{thm:dynamic:streams}
Let $S$ be a dynamic stream of polynomially bounded by $\Delta^{O(d)}$ of updates (inserts and deletes) to 
a point set $P \subseteq [\Delta]^d$. 
Let $k,z \in \Nats$ be two parameters. 
Let $0 < \eps,\delta \le 1$ be the error and failure parameters.
Then,  there exists a dynamic streaming algorithm that with probability at least $1-\delta$,
maintains a relaxed $(\eps,k,z)$-coreset at any time $t$ of the stream
for the $k$-center cost with $z$ outliers of the subset $P(t) \subseteq P$ of points that are inserted up to time $t$ of the stream $S$ 
but not deleted.  
The space complexity of this algorithm is $O((k/\eps^d+z)\log^4 (k\Delta/\eps \delta))$. 
\end{theorem}

%------------------------------------------------------------------------------------------
Next, we describe our algorithm in more detail.
Recall that for every grid $G_i$, we maintain a $s$-sparse recovery sketch $\mathcal{S}(G_i)$ 
where $s = k(4\sqrt{d}/\eps)^d + z)$.
The sketch $\mathcal{S}(G_i)$ supports the following operations: 

\begin{itemize}
    \item \textsc{Query}$(\mathcal{S}(G_i))$: This operation returns up to $s$ (almost uniformly chosen) 
         non-empty cells of the grid $G_i$ with their exact number of points. 
    \item \textsc{Update}$(\mathcal{S}(G_i), (c,\xi))$ where $\xi\in\{+1,-1\}$:
This operation updates the sketch of the grid $G_i$.
In particular, the operation \textsc{Update}$(\mathcal{S}(G_i), (c, +1))$
means that we add a point to a cell $c \in G_i$.
The operation \textsc{Update}$(\mathcal{S}(G_i), (c, -1))$
means that we delete a point from a cell $c \in G_i$.
\end{itemize}

The pseudocode of our dynamic streaming algorithm is given below. 
We break the analysis of this algorithm and the proof of Theorem~\ref{thm:dynamic:streams} into a few steps. 
We first analyze the performance of the $s$-sparse-recovery sketch
from \cite{DBLP:journals/tcs/BarkayPS15} in our setting.
% We first prove that if the number of non-empty cells of a grid $G_i$ is at most $s$ 
% at any arbitrary time $t$ of the stream, the $s$-sparse recovery sketch $\mathcal{S}(G_i)$ 
% reports all of these cells and the exact number of points inside each one with high probability. 
% \mdb{I'm confused. Isn't this simply what \cite{DBLP:journals/tcs/BarkayPS15} gives us?
% The only thing we do in Lemma~\ref{lem:correctness:prob:sketch} is analyze the amount
% of storage. So I would write: }
We next prove that there exists a grid $G_j$ that has a set $Q_j$ of at most $s$ non-empty cells 
    such that the weighted set of centers of cells of $Q_j$ is a relaxed $(\eps,k,z)$-coreset.
We then combine these two steps and 
prove that at any time $t$ of the stream, there exists a grid whose set of 
non-empty cells is a relaxed $(\eps,k,z)$-coreset of size at most $s$.
The final step is to prove the space complexity of Algorithm~\ref{alg:dyn:stream}. 

%------------------------------------------------------------------------------------------

\begin{algorithm}[ht] 
\caption{A dynamic streaming algorithm to compute $(\eps,k,z)$-coreset} 
\label{alg:dyn:stream}
\begin{algorithmic}[1]
\State Let $G_i$, for $i \in [\lceil \log\Delta \rceil ]$, be a partition of $[\Delta]^d$ 
       into a grid with cells of size $2^i\times\cdots\times 2^i$. 
\State Let $\mathcal{S}(G_i)$ be a $s$-sample recovery sketch for the grid $G_i$, 
where $s = k(4\sqrt{d}/\eps)^d + z$, as provided by Lemma~\ref{lem:s-recovery}.
\State Let $\F(G_i)$ be an $\zeronorm{F}$-estimator for the number of non-empty cells of $G_i$, as provided by Lemma~\ref{lem:f0}.
\While{not end of the stream}
    \State \begin{minipage}[t]{13cm}
           Let $(q,\xi)$ be the next element in the stream, where
           $q\in [1..\Delta]^d$ and $\xi\in\{+1,-1\}$ indicates
           whether $q$ is inserted or deleted. \\[-3mm]
           \end{minipage}
    \For{$i = 0$ to $\lceil \log\Delta \rceil$}
        \State Let $c(q)$ be the cell in $G_i$ that contains the point $q$. 
        \State \textsc{Update}$(\mathcal{S}(G_i), (c(q),\xi))$ 
                \hfill $\rhd$ update the $s$-sample recovery sketch for $G_i$
        \State  \textsc{Update}$(\F(G_i), (c(q),\xi))$ 
                \hfill $\rhd$ update the $\zeronorm{F}$ estimator for $G_i$.
    \EndFor
    \State  Let $G_j$ be the grid with the smallest cell side length for which \textsc{Query}$(\F(G_j))\leq s$.
    \State  $Q_j \gets$ \textsc{Query}$(\mathcal{S}(G_j))$ 
            \hfill $\rhd$  extract the non-empty cells with their number of points
    \For{each cell $c \in Q_j$}
    \State Choose the center of $c$ as the representative of $c$ and assign the number of points 
    in $c$ as the weight of this representative. 
    \State \textbf{report} the weighted representatives of non-empty cells in $Q_j$
    as a coreset of $P(t)$. 
\EndFor
\EndWhile

\end{algorithmic}
\end{algorithm}

%------------------------------------------------------------------------------------------
% \subsection{Analysis}
\medskip

\begin{lemma}
\label{lem:correctness:prob:sketch}
For any time $t$, the following holds:
If the number of non-empty cells of a grid $G_i$ at time $t$ is at most $s$, 
then querying the $s$-sample recovery sketch $\mathcal{S}(G_i)$ returns 
all of them with probability $1-\delta$. 
The space usage of the sketch $\mathcal{S}(G_i)$ is 
$O((k/\eps^d + z)\log^3(\frac{k\Delta }{\eps\delta} ))$. 
% \mdb{We are mixing up 'bits of storage" and "space" here. This
% should be fixed.}
\end{lemma}

\begin{proof}
The $s$-sample recovery sketch $\mathcal{S}(G_i)$ of \cite{DBLP:journals/tcs/BarkayPS15}
reports, with probability of at least $1-\delta$, 
all elements with non-zero frequency together with their exact frequency,
if the number of such elements is at most~$s$. 
(If it is more, we will get a sample of size~$s$; see Lemma~\ref{lem:s-recovery}) for the exact statement.)
% Let $Q_i$ be the set of non-empty cells of the grid $G_i$ at time $t$.
% Since $|Q_i| \le s$, we can use the sketch $\mathcal{S}(G_i)$ to return 
% all cells of the set $Q_i$ with their exact number of points. 
This proves the first part of the lemma.

The structures uses $O(s\log(s/\delta)\log^2 U)$ space, 
where  $U$ is the size of the universe. 
For the grid $G_i$, the universe size $U$ is the number of cells in~$G_i$,
which is $\lceil \Delta^d/2^{2i} \rceil$. 
The parameter $U$ is maximized for the grid $G_0$, which has $\Delta^d$ cells. 
Therefore, for $s = \Theta(k(\sqrt{d}/\eps)^d + z)$, 
the space usage of the sketch $\mathcal{S}(G_i)$ is 
\[
%    O(s\log(s/\delta)\log^2 U) 
%    = O((k(\sqrt{d}/\eps)^d + z)\log(\frac{k(\sqrt{d}/\eps)^d + z}{\delta})\cdot \log^2(\Delta^d) ) 
    O\left(\left(  k \left(\frac{\sqrt{d}}{\eps}\right)^d + z \right)\log \left(\frac{(k\sqrt{d}/\eps) + z}{\delta}\right)\cdot \log^2 (\Delta^d) \right)  =
    O \left( \left(k/\eps^d + z\right)\log^3\left(\frac{k\Delta }{\eps\delta}\right) \right)  \enspace . 
\]
where we use that $z\leq \Delta^d$ and that $d$ is assumed to be a constant.
\end{proof}

Lemma~\ref{lem:correctness:prob:sketch} provides the sketch $\mathcal{S}(G_i)$ 
if we query only once (say, at the end of the stream) and only for one fixed grid $G_i$.
Next, we assume that the length of the stream $S$ is polynomially bounded by $\Delta^{O(d)}$ 
and apply the union bound to show that the statement of Lemma~\ref{lem:correctness:prob:sketch}  
is correct for every grid $G_i$ at any time $t$ of the stream $S$. 
% \mdb{Actually, our algorithm will only query one specific $G_i$ (the one 
% at the ``correct resolution'') so we do not need that it is correct for all $G_i$
% simultaneously. Also, we do not really need to claim that with a certain probability
% all answers (at every point in time) are correct. We can just say that, at any time $t$,
% we get a correct answer with probability $1-\delta$, can't we?}

\begin{lemma}
\label{lem:union:sprase:recovery}
Suppose the length of the stream $S$ is polynomially bounded by $\Delta^{O(d)}$. 
Then, at any time $t$, 
we can return  all non-empty cells (with their exact number of points) of any grid $G_i$ 
that has at most $s$ non-empty cells with probability at least $1-\delta$. 
The space that we use to provide this task is 
$O\left((k/\eps^d + z)\log^4\left(\frac{kd\Delta }{\eps\delta} \right)\right)$. 
\end{lemma}

\begin{proof}
Lemma~\ref{lem:correctness:prob:sketch} with probability at least $1-\delta$, 
guarantees that we can return all non-empty cells of a fixed grid $G_i$ 
if for a fixed time $t$, $G_i$ has at most $s$ non-empty cells.
We have $\lceil \log\Delta \rceil$ grids and we assume that $|S| = \Delta^{O(d)}$. 
Thus, we can replace the failure probability $\delta$ by 
$\delta' = \frac{\delta}{\log(\Delta)\cdot \Delta^{O(d)}} = \frac{\delta}{\Delta^{\tilde{O}(d)}}$ 
to provide such a guarantee for any grid $G_i$ at any time $t$. 
By that, assuming $d$ is constant, 
the space usage of all sketches $\mathcal{S}(G_i)$ for $i \in [\lceil \log\Delta \rceil]$  will be 
$$
O\left(\log(\Delta)(k/\eps^d + z)\log^3\left(\frac{kd\Delta}{\eps\delta'} \right)\right) = 
O\left((k/\eps^d + z)\log^4\left(\frac{kd\Delta }{\eps\delta} \right)\right) \enspace . 
$$
\end{proof}

In Algorithm~\ref{alg:dyn:stream}, 
the sketch $\F(G_i)$ is an $\zeronorm{F}$-estimator for the number of non-empty cells of $G_i$. 
This sketch is provided by Lemma~\ref{lem:f0} 
for which we use $O((1/\eps^{2}+\log U)\log(1/\delta))$ space  
to obtain a success probability of at least $1-\delta$. 
Similar to Lemma~\ref{lem:union:sprase:recovery} we have the following lemma. 

\begin{lemma}
\label{lem:union:f:zero}
Suppose the length of the stream $S$ is polynomially bounded by $\Delta^{O(d)}$. 
Then, at any time $t$, 
we can return  approximate the number of non-empty cells of any grid $G_i$ 
within $(1\pm\epsilon)$-factor with the success probability of at least $1-\delta$. 
The space that we use to provide this guarantee is 
$O(\frac{1}{\eps^2}\cdot \log^2(\Delta/\delta))$. 
\end{lemma}

\begin{proof}
The space usage of the sketch that Lemma~\ref{lem:f0} provides is $O((1/\eps^{2}+\log U)\log(1/\delta))$. 
Recall that the parameter $U$ is maximized for the grid $G_0$, which has $\Delta^d$ cells. 
Thus, by applying the union bound for any grid $G_i$ at any time $t$, 
we provide the $(1\pm\epsilon)$-approximation of the number of non-empty cells of $G_i$ 
with probability $1-\delta$ and the space usage of 
$
    O((1/\eps^{2}+\log U)\log(1/\delta)) = O(\frac{1}{\eps^2}\cdot \log^2(\Delta/\delta)) \enspace .
$
\end{proof}

Next, we prove that there exists a grid $G_j$ that has a set $Q_j$ of at most $s$ non-empty cells 
    such that the weighted set of centers of cells of $Q_j$ is a relaxed $(\eps,k,z)$-coreset.

\begin{lemma}
\label{lem:num:cell:opt}
Let $P \subseteq [\Delta]^d$ be a point set and $0 < \eps \le 1$ be 
the error parameter. 
% Let $\optkz(P)$ be the optimal $k$-center radius of $P$ with $z$ outliers. 
Suppose that $2^j \le \frac{\eps}{\sqrt{d}} \cdot \optkz(P) < 2^{j+1}$. 
Then, 
\begin{itemize}
    \item at most $k(4\sqrt{d}/\eps)^d + z$ cells of the grid $G_j$ 
    % that we impose on the discrete space $[\Delta]^d$
are non-empty, and 
    \item the set of representative points of non-empty cells $Q_j$  of $G_j$ is a relaxed $(\eps,k,z)$-coreset 
for the $k$-center cost of $P$ with $z$ outliers.
\end{itemize}
\end{lemma}

\begin{proof}
Let $C^*=\{c^*_1,\cdots,c^*_k\}$ be an optimal set of $k$ centers. 
% whose radii are the optimal $k$-center (with $z$ outliers) radius $\optkz(P)$. 
Since, $2^j \le \frac{\eps}{\sqrt{d}} \cdot \optkz(P) < 2^{j+1}$, 
the balls centered at centers $C^*=\{c^*_1,\cdots,c^*_k\}$ of radius $\optkz(P)$ 
are covered by hypercubes of side length $\frac{2\sqrt{d}}{\eps}\cdot2^{j+1}$. 
Thus, these balls can cover or intersect at most 
$k\cdot (\frac{\frac{2\sqrt{d}}{\eps}\cdot2^{j+1}}{2^j})^d = k(4\sqrt{d}/\eps)^d$
cells of the grid $G_j$. The number of cells of the grid $G_j$ 
that can contain at least one outlier is at most $z$. 
Thus, the total number of non-empty cells of the grid $G_j$ is at most $k(4\sqrt{d}/\eps)^d + z$ 
what proves the first claim of this lemma.

The proof that the set of representative points of non-empty cells $Q_j$ 
is a relaxed $(\eps,k,z)$-coreset of $P$ is similar to the proof of Lemma~\ref{lem:MBC:is:coreset} 
and so we omit it here. 
The only difference is that 
the centers are now centers of non-empty grid cells, so we get a relaxed corset instead of a ``normal'' coreset (whose points are required to be a subset of the input points).
\end{proof}

Now, we prove that at any time $t$ of the stream, 
there exists a grid whose set of 
non-empty cells provides a relaxed $(\eps,k,z)$-coreset 
of size at most $s$.

\begin{lemma}
\label{lem:good:grid:exists}
Suppose the length of the stream $S$ is polynomially bounded by $\Delta^{O(d)}$. 
Let $t$ be any arbitrary time of the stream $S$. 
Let $P(t)$ be the subset of points that are inserted up to time $t$ of the stream $S$ 
but not deleted.  
Let $\optkz(P(t))$ be the optimal $k$-center radius with $z$ outliers at time $t$. 
Then, with probability $1-\delta$, 
Algorithm \ref{alg:dyn:stream} returns
% there exists a grid $G_{j}$ whose extracted set $Q_j$ of non-empty cells 
% (queried by \textsc{Query}$(\mathcal{S}(G_j))$) provides 
a relaxed $(\eps,k,z)$-coreset 
for the $k$-center cost with outliers of the set $P(t)$.
\end{lemma}

\begin{proof}
Based on Lemma~\ref{lem:union:sprase:recovery}, 
at any time $t$, we can return  all non-empty cells (with their exact number of points) of any grid $G_i$ 
that has at most $s$ non-empty cells with probability at least $1-\delta$. 
Moreover, according to Lemma~\ref{lem:union:f:zero}, 
at any time $t$, we can return  approximate the number of non-empty cells of any grid $G_i$ 
within $(1\pm\epsilon)$-factor with the success probability of at least $1-\delta$. 

Now, assume that  at time $t$, 
the optimal $k$-center radius with $z$ outliers of the point set $P(t)$ is 
$\optkz(P(t))$. Assume that $2^i \le \frac{\eps}{\sqrt{d}} \cdot \optkz(P(t)) < 2^{i+1}$. 
Then, Lemma~\ref{lem:num:cell:opt} shows that 
the number of non-empty cells of the grid $G_i$ is at most $s$. 
Moreover, the set of representative points of these non-empty cells is a relaxed $(\eps,k,z)$-coreset 
for the $k$-center cost of $P(t)$ with $z$ outliers. 
In Algorithm~\ref{alg:dyn:stream}, we consider the grid $G_{j}$ for $j \le i$ of smallest side length 
that has at most $s$ non-empty cells. 
Let $Q_j$ be the set of non-empty cells of $G_j$. 
Then, the set of centers of the cells in $Q_j$ is a relaxed $(\eps,k,z)$-coreset 
for the $k$-center cost with outliers of the set $P(t)$ 
which proves the lemma. 
\end{proof}

\begin{lemma}
\label{lem:space:sketch}
The total space used by Algorithm~\ref{alg:dyn:stream} is $O\left((k/\eps^d + z)\log^4\left(\frac{kd\Delta }{\eps\delta} \right)\right)$.
\end{lemma}

\begin{proof}
The space of Algorithm~\ref{alg:dyn:stream} is dominated 
by the space usage of the $s$-sparse recovery sketches for grids $G_i$ and 
the space usage of $\zeronorm{F}$-estimators for the number of non-empty cells of $G_i$ 
for $i \in [\lceil \log\Delta\rceil]$. 
Using Lemma~\ref{lem:union:sprase:recovery}, the space of the former one is 
$O\left((k/\eps^d + z)\log^4\left(\frac{kd\Delta }{\eps\delta} \right)\right)$. 
The space of the latter one based on Lemma~\ref{lem:union:f:zero} is  
$O(\frac{1}{\eps^2}\cdot \log^2(\Delta/\delta))$.  
The second space complexity is dominated by the first one. 
Thus, the total space complexity of Algorithm~\ref{alg:dyn:stream} 
is 
$O\left((k/\eps^d + z)\log^4\left(\frac{kd\Delta }{\eps\delta} \right)\right)$. 
\end{proof}

%------------------------------------------------------------------------------------------

\subsection{A lower bound for the fully dynamic streaming model}
In this section, we provide a lower bound that shows the dependency on the universe size~$\Delta$ is unavoidable in the dynamic streaming model.
The restriction that we put to prove Theorem \ref{thm:lower:bound:dynamic} is the same as the setting in section \ref{sec:streaming:lb} for the insertion-only lower bound.
 
\mypara{Overview.}
For the fully dynamic streaming model, where it is also possible to delete the points, we show an $\Omega((k/\eps^d)\log{\Delta})$ lower bound for the points in
a $d$-dimensional discrete Euclidean space $[\Delta]^d=\{1,2,3,\cdots,\Delta\}^d$.
Adding it to the $\Omega(z)$ lower bound of the insertion-only streaming model leads to an $\Omega((k/\eps^d)\log\Delta + z)$ lower bound for the fully dynamic streaming setting. 

In the insertion-only construction the $k-2d-1$ ``clusters'' where just single points, but here each cluster $C_i$ 
% Similar to the streaming model, we consider an instance of points $P(t)$ that consists of $z$ outlier points and $k-2d+1$ clusters. However, each cluster
consists of $\Theta(\log{\Delta})$ groups $G_i^1,G_i^2,\ldots$ that are scaled copies of (a part of) a grid of size $\Theta(1/\eps^d)$, where the $j$-th copy is scaled by $2^j$; see Figure \ref{fig:lb:dynamic:config}. We claim that all the non-outlier points in $P(t)$ must be in any $(\eps, k, z)$-coreset of $P(t)$.
To prove the claim by contradiction, we will assume that the coreset does not contain a non-outlier point $p^* \in G_{i^*}^{m^*}$, 
and then delete all groups $G_i^m$ for all $i$ and for all $m > m^*$. 
Next, we insert a carefully chosen set of $2^d$ new points to the stream such that the coreset underestimates the optimal radius, which is a contradiction.
This will lead to the following theorem.

\begin{theorem}[Lower bound for dynamic streaming algorithms]
\label{thm:lower:bound:dynamic}
Let $0 < \eps \leq \frac{1}{8d}$, $k\geq 2d$ and $\Delta \geq ((2k+z)(\frac{1}{4\eps}+d))^2$. 
Any deterministic fully dynamic streaming algorithm that maintains a weight-restricted $(\eps,k,z)$-coreset 
for the $k$-center problem with $z$ outliers in
a $d$-dimensional discrete Euclidean space $[\Delta]^d=\{1,2,3,\cdots,\Delta\}^d$  
must use $\Omega((k/\eps^d)\log{\Delta}+ z)$ space.
% where $\Delta\in \Nats$ indicates the size of the universe from which the coordinates are taken.
\end{theorem}

The remainder of this section is dedicated to the proof of Theorem \ref{thm:lower:bound:dynamic}. To prove the theorem, we will present a scenario of insertions and deletions that forces the size of the coreset to be $\Omega((k/\eps^d)\log{\Delta})$. 
Recall that by Lemma \ref{lem:lb:k+z}, the size of coreset is $\Omega(z)$ even in the insertion-only model.
Therefore, the coreset size must be~$\Omega((k/\eps^d)\log{\Delta}+z)$ in the fully dynamic streaming model.
\\

Let $\lambda := 1/(4d\eps)$, and assume without loss of generality $\lambda/2$ is an integer. 
Let $h := d(\lambda+2)/2$ and $r :=\sqrt{h^2-2h+d}$, 
and let $g:=\frac{1}{2}\log{\Delta} - 2$. 
Instance $P(t)$ consists of $k-2d+1$ clusters $C_1,\ldots,C_{k-2d+1}$ at distance $2^{g+2}(h+r)$ from each other, 
and also $z$ outlier points $o_1,\ldots,o_z$ at distance $2^{g+2}(h+r)$ from each other ; see Figure \ref{fig:lb:dynamic:config}.
Each cluster $C_i$ consists of $g$ groups $G^1_i,\ldots,G_i^g$. Each group $G_i^m$ is is constructed by placing $(\lambda+1)^d$ points in a grid whose cells have side length $2^m$, 
and the omitting the lexicographically smallest ``octant''.
The omitted octant is used to place the groups $G_i^{1}\cup\ldots \cup G_i^{m-1}$ 
as illustrated in Figure \ref{fig:lb:dynamic:config}.
Therefore, each group consists of $(\lambda+1)^d-(\lambda/2+1)^d = \Omega(1/\eps^{d})$ points.

%-----------------------------------------------------------------------------

\begin{figure}
\begin{center}
\includegraphics{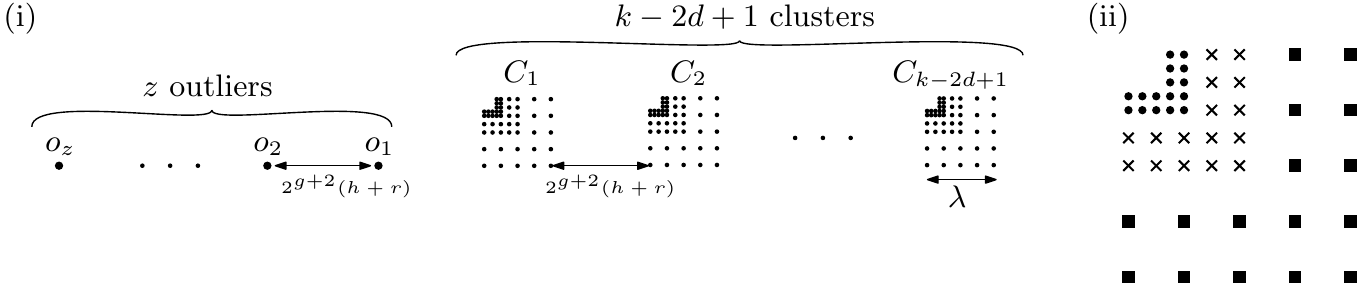}
\end{center}
\caption{Illustration of the lower bound in
Theorem \ref{thm:lower:bound:dynamic}.
% We have $\lambda := 1/(4d\eps)$ is an integer, $h := d(\lambda+2)/2$ and $r :=\sqrt{h^2-2h+d}$.
Part~(i) shows the global construction, part~(ii) shows an example of a cluster $C_i$, where $g=3$. The points in groups $G_i^{1}$, $G_i^{2}$ and $G_i^{3}$ are showed by  disks, crosses and squares respectively.}
\label{fig:lb:dynamic:config}
\end{figure}
%----------------------------------------------------------------------------- 

Suppose that all points in $P(t)$ are inserted into the stream by time $t$, and let $P^*(t)$ be the maintained $(\eps, k, z)$-coreset at time $t$. We claim that $P^*(t)$ must contain all non-outlier points, which means the size of $P^*(t)$ must be $\Omega(kg/\eps^d) = \Omega((k/\eps^d) \log{\Delta})$.

%------------------------------------------------------------------------------------------

\begin{claim}
Let $p$ be an arbitrary non-outlier point in $P(t)$, that is, a point from one of the cluster $C_i$, and let $P^*(t)$ be an $(\eps, k, z)$-coreset of $P(t)$.
Then, $p$ must be in $P^*(t)$.

\end{claim}
% \begin{proofof}{Claim \ref{claim:sliding}}
\begin{proof}
To prove the claim, assume for a contradiction that
there is a point~$p^* \in G_{i^*}^{m^*}$ that is not explicitly stored in~$P^*(t)$, where $p^*=(p_1^*,\ldots,p_d^*)$. 
First we delete all points of $G_i^m$ for all $m \geq m^*$ and all $i$.
Then the next $2d$ points that we insert 
are the points from $P^+ := \{p_1^+,\ldots,p_d^+ \}$ and $P^- := \{p_1^-,\ldots,p_d^- \}$.
Here $p_j^+ = (p_{j,1}^+,\ldots,p_{j,d}^+)$, where
$p_{j,j}^+:= p_j^*+2^{m^*}(h+r)$ and $p_{j,\ell}^+:=p_\ell^*$ for all $\ell\neq j$.
Similarly, $p_j^- = (p_{j,1}^-,\ldots,p_{j,d}^-)$ where $p_{j,j}^-:= p_j^*-2^{m^*}(h+r)$ and $p_{j,\ell}^-:=p^*_\ell$ 
for all $\ell\neq j$.
% ; see Figure~\ref{fig:???} \lb{Will we have a figure for this part?}. 
It will be convenient to assume that each point in $P^+\cup P^-$ has weight~2;
of course we could also insert two points at the same location (or,
almost at the same location). 
Note that this is similar to the construction used in the insertion-only lower bound, which was illustrated in Figure~\ref{fig:new:lb:streaming:missing}.

Let $P(t') := P(t)\cup P^- \cup P^+ \setminus \left(\bigcup_{m > m^*} G_{i}^{m}\right)$ and let $P^*(t')$ be the coreset of $P(t')$. 
Since $P^*(t)$ did not store $p^*$, we have $p^*\not\in P^*(t')$. We will show that this
implies  $P^*(t')$ underestimates the optimal radius by too much. We first give a
lower bound on $\optkz(P(t'))$. 
Using the same argument as in the proof of Claim \ref{lem:OPT:lower:bound} we can conclude
\begin{claiminproof}
$\optkz(P(t')) \geq  2^{m^*}\cdot (h+r)/2$.
\end{claiminproof}

Next we show that, because $P^*(t')$ does not contain the point~$p^*$, it
must underestimate $\optkz(P^*(t'))$ by too much. To this end, we first have the following claim,
which can be proved in the same way as Claim~\ref{claim:lb:streaming}.
% The idea of the proof is that an optimal solution for $P^*(t')$
% can use $2d$ balls of radius $2^{m^*}\cdot r$ to cover $C_{i^*}\cup P^+\cup P^- \setminus \left(\bigcup_{m > m^*} G_{i^*}^{m}\right)$.
\begin{claiminproof} 
$\optkz(P^*(t')) \leq 2^{m^*}\cdot r$.
\end{claiminproof}
Lemma~\ref{lem:r:bound} in the appendix
gives us that $r<(1-\eps)(r+h)/2$. Putting everything together, we have
\[
(1-\eps)\cdot\optkz(P(t')) \ \ \geq \ \  (1-\eps)\cdot 2^{m^*}(r+h)/2 \ \ > \ \ 2^{m^*}\cdot r \ \ \geq \ \   \optkz(P^*(t')) \enspace.
\]
However, this is a contradiction to our assumption that $P^*(t')$ is an $(\eps,k,z)$-coreset of $P(t')$.
Hence, if $P^*(t)$ does not store all points from each of the clusters~$C_i$,
then it will not be able to maintain an $(\eps,k,z)$-coreset.

\end{proof}

It remains to verify that the points of our construction are from
a $d$-dimensional discrete Euclidean space $[\Delta]^d=\{1,2,3,\cdots,\Delta\}^d$.
Note that all points in our construction can have integer coordinates. 
Thus, it is enough to show that $\Delta' \leq \Delta$,
where $\Delta'$ is the maximum of the value $\max_{1\leq i\leq d} |p_i-q_i|$ over all pairs of points $p,q$ used in the construction.
$P(t)$ consists of  $z$ outlier points and $k-2d+1$ clusters of side length $2^g\cdot \lambda$, where the distance between any two consecutive outliers or clusters is $2^{g+2}(h+r)$.
In the construction, we then also add sets $P^+$ and $P^-$, whose points are at distance at most $2^g(h+r)$ from some point $p^*$ in one of the clusters.
Therefore, $\Delta' \leq (k+z)\cdot 2^{g+2}(h+r) + k\cdot 2^g  \lambda$.
Recall that $\lambda = 1/(4d\eps)$ and $h = d(\lambda+2)/2$.
Thus, $\lambda/2 \leq h$ and then $\Delta' \leq (2k+z)\cdot 2^{g+2}(h+r)$.
Besides, $r \leq h$ since $r =\sqrt{h^2-2h+d}$. Therefore,
$$
\Delta' \leq
(2k+z)\cdot 2^{g+2}(h+r) \leq
(2k+z)\cdot 2^{g+2}(2h) = (2k+z)\cdot 2^{g+2} d(\lambda+2) =
$$
$$
(2k+z)\cdot 2^{g+2}\cdot d\left(\frac{1}{4d\eps}+2\right) =
(2k+z)\cdot 2^{g+2}\left(\frac{1}{4\eps}+2d\right) \enspace .
$$
Hence, $\log{\Delta'}  \leq  2 + g + \log{((2k+z)(\frac{1}{4\eps}+d))}$.
Recall that $g= \frac{1}{2}\log{\Delta}-2$ and we assume $\Delta \geq ((2k+z)(\frac{1}{4\eps}+d))^2$, therefore, $\log{((2k+z)(\frac{1}{4\eps}+d))}\leq \frac{1}{2}\log{\Delta}$.
Thus $\log{\Delta'} \leq \log{\Delta}$, which means $\Delta' \leq \Delta$. 
This finishes the proof of Theorem 27.

%-----------------------------------------------------------------------------
\section{A lower bound for the sliding-window model}
\label{app:lb:sliding}
%-----------------------------------------------------------------------------
In this section, we show that
any deterministic algorithm in the sliding-window model that guarantees a $(1\pm\eps)$-approximation 
for the $k$-center problem with outliers in~$\Reals^d$ must use $\Omega((kz/\eps^d)\log \sigma)$ space, 
where $\sigma$ is the ratio of the largest and smallest distance between any two points in the stream. 
Recently De~Berg, Monemizadeh, and Zhong~\cite{DBLP:conf/esa/BergMZ21} developed 
a sliding-window algorithm that uses $O((kz/\eps^d)\log \sigma)$ space. 
Our lower bound shows the optimality of their algorithm 
and gives a (negative) answer to a question posed by De~Berg~\etal~\cite{DBLP:conf/esa/BergMZ21}, 
who asked whether there is a sketch for 
this problem whose storage is polynomial in~$d$. 

%Our lower bound for the insertion-only and fully dynamic setting  therefore work for algorithms that maintain a coreset. Our lower bound for the streaming setting, however, works in the same model as the lower-bound of De Berg, Monemizadeh and Zhong.

%-----------------------------------------------------------------------------

%-----------------------------------------------------------------------------
% \lb{TODO: Change the theorem, it should not be for coresets.}
% \lb{From the Appendix:}

\mypara{Lower-bound setting.} 
Let $P:= \langle p_1,p_2,\ldots \rangle$ be a possibly infinite stream of points from
a metric space~$X$ of doubling dimension~$d$ and spread ratio~$\sigma$, where~$d$ is considered to
be a fixed constant. We denote the arrival time
of a point~$p_i$ by $\at(p_i)$. We say that $p_i$ \emph{expires} at 
time~$\et(p_i) := \at(p_i) + W$, where $W$ is the given length of the time window.
To simplify the exposition, we consider the $L_\infty$-distance instead of the Euclidean distance,  
where the $L_{\infty}$-distance between two points $p,q\in \Reals^d$  is defined as $L_\infty(p,q) = \max_{i=1}^d |p_i-q_i|$. 
Note that the doubling dimension of $\Reals^d$ under the $L_{\infty}$-metric is~$d$.

The constructions we presented earlier for the insertion-only and the fully-dynamic streaming model, gave lower bounds on the size of an $(\eps,k,z)$-coreset maintained by the algorithm. For the sliding-window model, we will use the lower-bound model introduced by De~Berg, Monemizadeh, and Zhong~\cite{DBLP:conf/esa/BergMZ21}. This model gives lower bounds on \emph{any} algorithm that maintains a $(1\pm\eps)$-approximation of the radius of an optimal $k$-center clustering with $z$ outliers. Such an algorithm may do so by maintaining an $(\eps,k,z)$-coreset, but it may also do it in some other (unknown) way. The main restriction is that the algorithm can only change its answer when either a new point arrives or at some explicitly stored expiration time. More precisely, their lower-bound model is as follows \cite{DBLP:conf/esa/BergMZ21}.

Let $S(t)$ be the collection of objects being stored at time~$t$. These objects may be points, weighted points, balls, or anything else that the algorithm needs to store to be able to approximate the optimal radius. The only conditions on $S(t)$ are as follows.
\begin{itemize}
\item Each object in $S(t)$ is accompanied by an expiration time, which is equal to the expiration 
      time of some point $p_i\in P(t)$. 
\item Let $p_i\in P(t)$. If no object in $S(t)$ uses $\et(p_i)$
      as its expiration time, then no object in $S(t')$ with $t'>t$ can use $\et(p_i)$
      as its expiration time. (Once an expiration time has been discarded, it cannot be recovered.)
\item The solution reported by the algorithm is uniquely determined by $S(t)$, and
      the algorithm only modifies $S(t)$ when a new point arrives or when an object in~$S(t)$ expires.
\item The algorithm is deterministic and oblivious of future arrivals. In other words, the set
      $S(t)$ is uniquely determined by the sequence of arrivals up to time~$t$, and the
      solution reported for $P(t)$ is uniquely determined by~$S(t)$.
\end{itemize}
The storage used by the algorithm is defined as the number of objects in~$S(t)$.
The algorithm can decide which objects to keep in $S(t)$ in any way it wants; it may even 
keep an unbounded amount of extra information in order to make its decisions. The algorithm can 
also derive a solution for $P(t)$ in any way it wants, as long as the solution is valid and uniquely
determined by~$S(t)$. 

\begin{theorem}[Lower bound for sliding window]
\label{thm:lower:bound:sliding}
Let $k \geq 2d$, $0 < \eps \leq 1/24$ and $\sigma \geq {(kz/\eps)}^2$. 
Any deterministic $(1\pm\eps)$-approximation algorithm in the sliding-window model
that adheres to the model described above
and solves the $k$-center problem with $z$ outliers in the  metric space $(\Reals^d, L_\infty)$
must use $\Omega((kz/\eps^d)\log \sigma)$ space, where 
$\sigma$ is the ratio of the largest and smallest distance between any two points in the stream.
\end{theorem}

\begin{proof}
Consider a deterministic $(1\pm\eps)$-approximation algorithm for the $k$-center 
clustering with $z$ outliers in the sliding-window model. With a slight abuse of notation, we let $S(t)$ be the set of expiration times that the algorithm maintains at time $t$.
In the following, we present a set of points $P(t)$ 
such that the algorithm needs to store $\Omega((kz/\eps^d)\log{\sigma})$ expiration times.

Let $\lambda := 1/(8\eps)$, and assume without loss of generality that~$\lambda$ is an odd integer. 
Let $g:=\frac{1}{2}\log{\sigma} - 1$ and $s:=\lambda^d-(\frac{\lambda+1}{2})^d$. 
Let $\zeta := \floor{\sqrt[d]{z}}$, and observe that $\zeta^d < z+1 \leq (\zeta+1)^d$.
Instance $P(t)$ consists of $k-2d+1$ clusters $C_1,\ldots,C_{k-2d+1}$ at distance $3\cdot 2^g\zeta\cdot 2\lambda$ from each other.
Each cluster $C_i$ consists of $g$ groups $G^1_i,\ldots,G_i^g$, 
and each group $G_i^j$ consists of $s$ subgroups $G_i^{j,1},\ldots,G_i^{j,s}$. 
Finally, each subgroup consists of $z+1$ points. 
Figure ~\ref{fig:lb:sliding:config} shows an overview of the construction, which we describe in more detail next.
Consider a grid $\mathcal{G}^j$ whose cells have side length $2^j$ and which has $(\zeta+1)^d$ grid points.  
The points of each subgroup $G_i^{j,\ell}$ are the lexicographically smallest $z+1$ points of this grid $\mathcal{G}^j$.
(That is, the first $z+1$ points in the lexicographical order of the coordinates). 
Recall that we consider the $L_\infty$-distance instead of the Euclidean distance. 
Therefore, the diameter of the subgroup $G_i^{j,\ell}$ is $2^j\zeta$.

Now we describe the relative position of the subgroups in a group $G_i^{j}$.
Let $\Pi^j$ be a $d$-dimensional grid consisting of $(2\lambda-1)^d$ cells that have side length $2^j\zeta$.
We label the cells in $\Pi$ as $\pi=(\pi_1,\cdots,\pi_d)$, where $1 \leq \pi_i \leq 2\lambda-1$ for all $i\in[d]$. 
%\lb{Please check if the previous statement is fine}
For instance, for $d=2$ the bottom-left cell would be labeled $(1,1)$.
We say the cell $\pi=(\pi_1,\cdots,\pi_d)$ is an \textit{odd cell}, if $\pi_i$ is odd for all $i\in [d]$. Hence, there are $\lambda^d$ odd cells in $\Pi^j$. 
Let the set $\Gamma^j$ be equal to $\Pi^j$ except that the lexicographically smallest ``octant''. More formally, 
$\Gamma^j = \Pi^j \setminus \{(\pi_1,\cdots,\pi_d) \in \Pi^j\ \text{ : } \forall_{i\in[d]} \pi_i \leq \lambda  \}$. Then $\Gamma^j$ is of size $\lambda^d-(\frac{\lambda+1}{2})^d=s$. 
The subgroups $G_i^{j,1},\cdots,G_i^{j,s}$ are placed in the cells of $\Gamma^j$, and groups $G_i^{j-1},\cdots,G_i^{1}$ are recursively placed in the omitted octant. See Figure \ref{fig:lb:sliding:config}.
Therefore, the diameter of group $G_i^j$ is $2^j\zeta\cdot \lambda + 2^j\zeta\cdot(\lambda -1 ) = 2^j\zeta\cdot(2\lambda - 1)$.

%-----------------------------------------------------------------------------

\begin{figure}
\begin{center}
\includegraphics{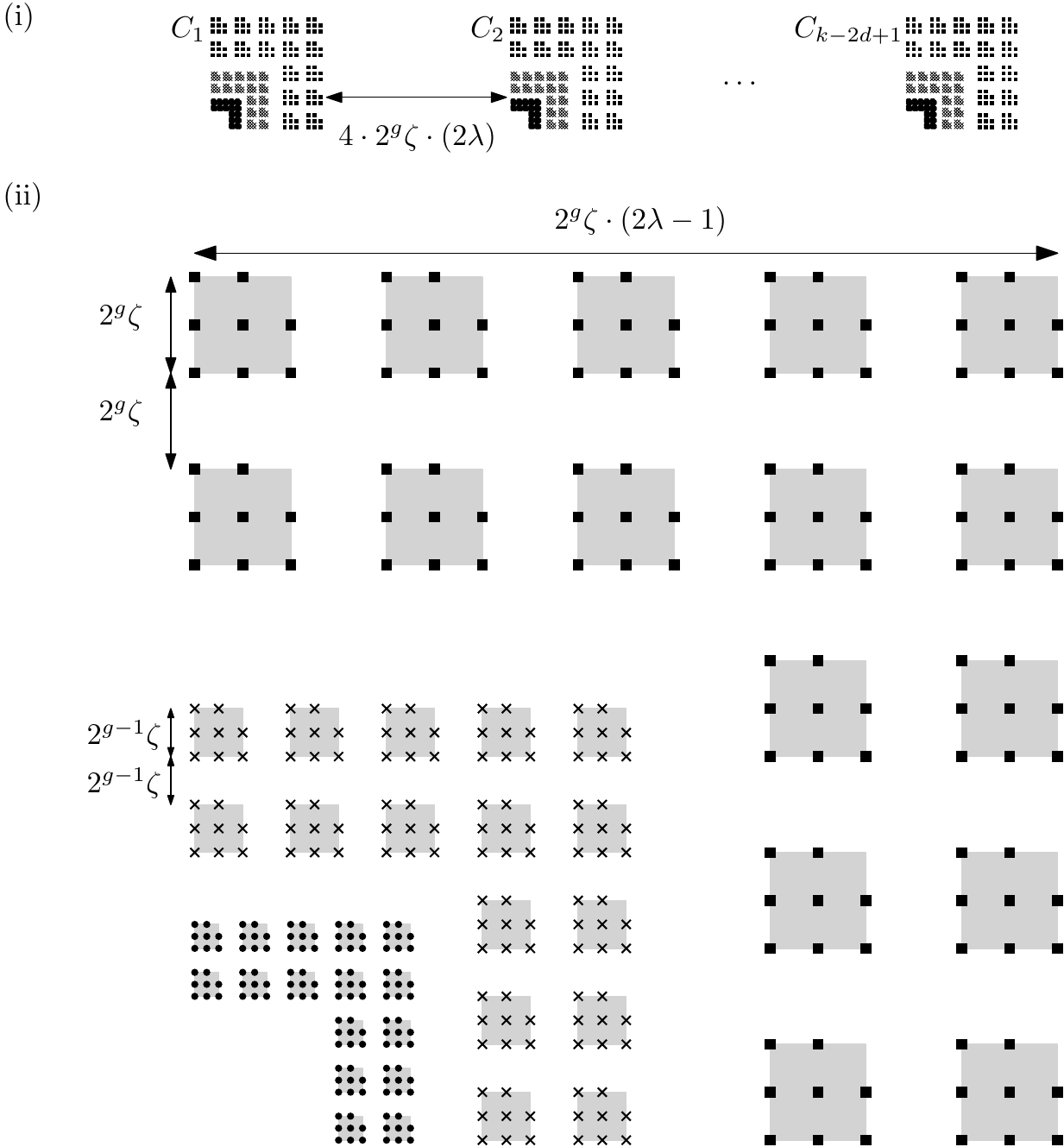}
\end{center}
\caption{Illustration of the lower bound in
Theorem \ref{thm:lower:bound:sliding}.
% We have $\lambda := 1/(4d\eps)$ is an integer, $h := d(\lambda+2)/2$ and $r :=\sqrt{h^2-2h+d}$.
Recall that $\lambda=\Theta(1/\eps)$ and $g=\Theta(\log{\sigma})$.
Part~(i) shows the global construction, part~(ii) shows an example of a cluster $C_i$, where $z=7$, $\lambda=3$ and $g=3$. The points in groups $G_i^{1}$, $G_i^{2}$ and $G_i^{3}$ are showed by  disks, crosses and squares respectively.}
\label{fig:lb:sliding:config}
\end{figure}
%----------------------------------------------------------------------------- 

Next we explain the order of arrivals.
First, the subgroups $G_{k-2d+1}^{g,s},\ldots,G_{1}^{g,s}$ arrive. %\lb{$s,g$ or $g,s$?}
Then the subgroups $G_{k-2d+1}^{g,s-1},\ldots,G_{1}^{g,s-1}$ arrive, 
and so on.
More formally, $G_{i}^{j,\ell}$ arrives before $G_{i'}^{j',\ell'}$ if and only if $j > j'$ 
or ($j=j'$ and $\ell > \ell'$) or ($j=j'$ and $\ell=\ell'$ and $i>i'$).

Now, we claim that the size of $S(t)$ must be $\Omega((kzg)/\eps^d) = \Omega((kz\cdot \log{\sigma})/\eps^d)$.

\begin{claim}
\label{claim:sliding}
Let $p \in G_{i}^{j,\ell}$ be an arbitrary point in $P(t)$ such that $j>1$ or $\ell>1$, and $\et(p) > t + (2d(z+1)+z)$.
Then, $\et(p)$ must be in $S(t)$.

\end{claim}
% \begin{proofof}{Claim \ref{claim:sliding}}
\begin{proofinproof}
For the sake of contradiction, assume there is a point $p^* \in G_{i^*}^{j^*, \ell^*}$, where $j^*>1$ or $\ell^*>1$, and $\et(p^*) > t + (2d(z+1)+z)$, while $\et(p^*)$ is not explicitly stored in $S(t)$.
Let $t^-_{p^*}$ and $t^+_{p^*}$ be the time just before and just after the expiration of the point $p^*$ respectively.
As $\et(p^*)\notin S(t)$, then the sketch that the deterministic algorithm maintains at time $t^-_{p^*}$ and $t^+_{p^*}$ is the same, 
and so, it reports the same clustering for both $P(t^-_{p^*})$ and $P(t^+_{p^*})$. 
However, we show it is possible to insert a point set after the points of $P(t)$ have been inserted such that ${\optkz(P(t^+_{p^*}))}/{\optkz(P(t^-_{p^*}))} > 1-3\eps$.
Thus either at time $t^+_{p^*}$ or at time $t^-_{p^*}$, the answer of the algorithm cannot be a $(1\pm\eps)$-approximation.

Recall that the group $G_{i^*}^{j^*}$ consists of $s$ subgroups of diameter $2^{j^*}\zeta$ 
in a $\lambda^d$ grid-like fashion, and the diameter of 
$G_{i^*}^{j^*}$ is $2^{j^*}\zeta\cdot(2\lambda-1)$. 
Observe that we consider the $L_\infty$-distance instead of the Euclidean distance. 
First, we define $x_\mathrm{min}^*(\alpha)$ and $x_\mathrm{max}^*(\alpha)$.
For $\alpha \in [d]$, we define
\[
x_\mathrm{min}^*(\alpha) := \min\{ x_\alpha \ | \ (x_1,\dots,x_\alpha,\dots,x_d) \in G_{i^*}^{j^*,\ell^*} \} \enspace ,
\]
\[
x_\mathrm{max}^*(\alpha) := \max\{ x_\alpha \ | \ (x_1,\dots,x_\alpha,\dots,x_d) \in G_{i^*}^{j^*,\ell^*}  \} \enspace .
\]

%-----------------------------------------------------------------------------
\begin{figure}
\begin{center}
\includegraphics{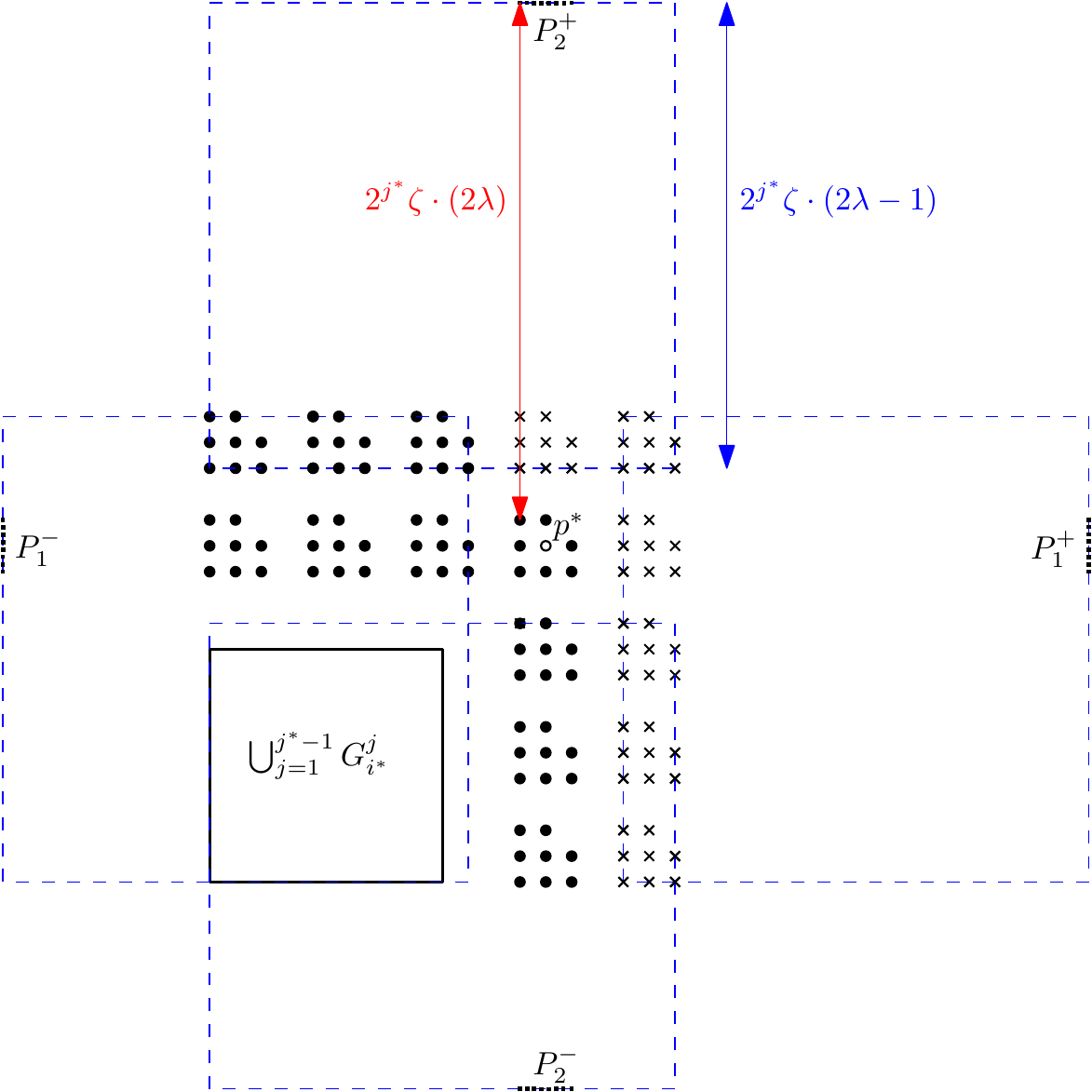}
\end{center}
\caption{
If the expiration time of the point $p^*\in G_{i^*}^{j^*,\ell^*}$ is not stored,
 we insert the $2d$ point sets $P^+_1,\ldots,P^+_d$ and $P^-_1,\ldots,P^-_d$.
The points that have expired (and are not re-inserted) before the expiration of $p^*$ are shown by crosses.
The optimal radius just before the expiration of $p^*$ is $2^{j^*}\zeta\cdot(2\lambda)$.
However, since we can consider all points in $G_{i^*}^{j^*,\ell^*}\setminus\{p^*\}$ 
as outliers after the expiration of $p^*$, the optimal radius 
just after the expiration of $p^*$ is $2^{j^*}\zeta(2\lambda-1)$, (dashed balls).}
\label{fig:lb:sliding:missing}
\end{figure}
%----------------------------------------------------------------------------- 

Now, we define the point sets $P^+_1,\ldots,P^-_d$ and $P^-_1,\ldots,P^-_d$ as follows (also see Figure \ref{fig:lb:sliding:missing}).
For every $\alpha\in [d]$, $P^+_\alpha=\{p^{+,0}_{\alpha},\dots,p^{+,z}_{\alpha}\}$ 
and for every $0\leq \iota \leq z$, we have $p^{+,\iota}_{\alpha} = (p^{+,\iota}_{\alpha,1},\dots,p^{+,\iota}_{\alpha,d})$ where
\[ p^{+,\iota}_{\alpha,\alpha} = x_\mathrm{max}^*(\alpha)+2^{j^*}\zeta\cdot(2\lambda),
\text{ and } p^{+,\iota}_{\alpha,\beta} = x_{min}^*(\beta)+ \frac{\iota(x_\mathrm{max}^*(\beta)-x_{min}^*(\beta))}{z} \text{ for all } \beta \neq \alpha \enspace .
\]

Similarly, for all $\alpha\in [d]$, $P^-_\alpha=\{p^{-,0}_{\alpha},\dots,p^{-,\iota}_{\alpha},\dots,p^{-,z}_{\alpha}\}$, where for each point $p^{-,\iota}_{\alpha} = (p^{-,\iota}_{\alpha,1},\dots,p^{-,\iota}_{\alpha,d})$ we have
\[ p^{-,\iota}_{\alpha,\alpha} = x_\mathrm{min}^*(\alpha)-2^{j^*}\zeta\cdot(2\lambda),
\text{ and } p^{-,\iota}_{\alpha,\beta} = x_\mathrm{min}^*(\beta)+ \frac{\iota(x_\mathrm{max}^*(\beta)-x_\mathrm{min}^*(\beta))}{z} \text{ for all } \beta \neq \alpha \enspace .
\]
Hence, $P^+_\alpha$ (and $P^-_\alpha$) consists of $z+1$ points
% of a parallel line to $G_{i^*}^{j^*,\ell^*}$ 
at distance $2^{j^*}\zeta\cdot(2\lambda)$ of $G_{i^*}^{j^*,\ell^*}$. 
Moreover, $x_\mathrm{min}^*(\beta) \leq p^{+,\iota}_{\alpha,\beta},\ p^{-,\iota}_{\alpha,\beta} \leq x_\mathrm{max}^*(\beta)$ 
if $\beta \neq \alpha$.
% while its diameter is same as the diameter of $G_{i^*}^{j^*,\ell^*}$.
We insert all points of the sets $P^+_1,\ldots,P^-_d$ and $P^-_1,\ldots,P^-_d$.
Moreover, for each point in $G_{i^*}^{j^*,\ell^*}\setminus \{p^*\}$, we re-insert it after its expiration.
Note that as we assume $\et(p^*) > t + (2d(z+1)+z)$, we have enough time from $t$ to $\et(p^*)$ to insert all these points.

As we assume $j^*>1$ or $\ell^*>1$ then each cluster $C_i$ contains at least $z+1$ points at time $t^-_{p^*}$ that are not expired.
In addition, each point set $G_{i^*}^{j^*,\ell^*}$,
$P^+_1,\ldots,P^+_d$, and
$P^-_1,\ldots,P^-_d$ consists of $z+1$ points at time $t^-_{p^*}$ that are not expired. 
As any pairwise distance between these $2d+1$ point sets is at least $2^{j^*}\zeta\cdot (2\lambda)$, then $\optkz(P(t^-_{p^*})) \geq 2^{j^*} \zeta \cdot\lambda$.
On the other hand, since $p^*$ is expired at time $t^+_{p^*}$,  
we consider the points of the set $G_{i^*}^{j^*,\ell^*}$ 
that are not expired (note that there are $z$ such points) 
as the outliers at time $t^+_{p^*}$ (see Figure \ref{fig:lb:sliding:missing}), thus
$\optkz(P(t^+_{p^*})) \leq 2^{j^*}\zeta\cdot(2\lambda-1)/2$.
Putting everything together we have
\[
\frac{\optkz(P(t^+_{p^*}))}{\optkz(P(t^-_{p^*}))}
\ \leq \ \frac{2^{j^*}\zeta\cdot(2\lambda-1)/2}{2^{j^*}\zeta\cdot\lambda}
\ = \ \frac{2\lambda - 1}{2\lambda}
\ = \ 1 -4\eps
\ < \ 1 - 3\eps \enspace .
\]
Which is a contradiction.
% \end{proofof}
\end{proofinproof}

It remains to show that the spread ratio of our construction is not more than $\sigma$.
Let $\sigma'$ be the ratio of the largest and smallest distance between any two points in our construction.
We show $\sigma' \leq \sigma$.
The diameter of each cluster $C_1,\ldots,C_{k-2d+1}$ is $2^g\zeta\cdot(2\lambda-1)$, and 
every two consecutive clusters are at distance $3\cdot2^g\zeta\cdot(2\lambda)$ from each other. 
Hence, the largest distance between any two points in the stream is less than $k\cdot 4\cdot 2^g\zeta\cdot(2\lambda)$. 
Besides, the points in sets $P^+_1,\ldots,P^-_d$ and $P^-_1,\ldots,P^-_d$ that we defined in Claim \ref{claim:sliding} are at distance at least $(x_\mathrm{max}^*(\alpha)-x_\mathrm{min}^*(\alpha)))/z = 2^{j^*}\zeta/z$ from each other. 
Therefore, the smallest distance between any two points in $P^+_1\cup\ldots\cup P^-_d \cup P^-_1\cup\ldots\cup P^-_d$ is at least $2\zeta/z$.
Moreover,  the smallest distance between any two points in  $C_1\cup\ldots\cup C_{k-2d+1}$ is $2^1$, and $2 
\geq 2\zeta/z$ since $\zeta=\sqrt[d]{z}$. 
% Therefore, the largest distance between any two points in the stream is at least $2\zeta/z$. \lb{What?}
Then we have $\sigma'  \leq   \frac{k\cdot 4\cdot2^g\zeta\cdot(2\lambda)}{2\zeta/z}  =  4\cdot 2^g kz\cdot \delta = 2\cdot 2^g kz/\eps$.
Hence, $\log{\sigma'}  \leq  1 + g + \log{(kz/\eps)}$.
Recall $g= \frac{1}{2}\log{\sigma}-1$.
Since we assume $\sigma \geq {(kz/\eps)}^2$, we therefore have $\log{(kz/\eps)}\leq \frac{1}{2}\log{\sigma} $.
Thus $\log{\sigma'} \leq \log{\sigma}$, which means $\sigma' \leq \sigma$. This finishes the proof of the theorem.
% Thus $\Omega(kzg/\eps^d) = \Omega((kz\log{\sigma})/\eps^d)$.
\end{proof}

\section{More MPC algorithms}
\label{sec:more:MPC}
In this section, we present two more MPC algorithms: a randomized $1$-round algorithm and a multi-round deterministic algorithm. The former algorithm is quite similar to an algorithm of Ceccarello~\etal~\cite{DBLP:journals/pvldb/CeccarelloPP19}, but by \nw{a more clever coreset construction,} we obtain an improved bound. The latter algorithm provides a trade-off between the number of rounds and the space usage.
%--------------------------------------------------------------------------------------------------
\subsection{A randomized $1$-round MPC algorithm}
\label{appendix:subsec:1:round}
%--------------------------------------------------------------------------------------------------
In this section, we present our $1$-round randomized algorithm. The algorithm itself
does not make any random choices; the randomization is only in the assumption that the distribution
of the set~$P$ over the machines~$M_i$ is random. More precisely, we assume
each point $p\in P$ is initially assigned 
uniformly at random to one of the $m$ machines~$M_i$.
The main observation is Lemma~\ref{lem:bound:outliers:per:machine} that with high 
probability\footnote{We say an event occurs with high probability
if it occurs with a probability of at least $1-1/n^2$.},
the number of outliers assigned to an arbitrary worker machine $M_i$ is at most 
$z' =\min(\outliersBound, z)$. 
As shown in Algorithm~\ref{alg:1:round},
each machine~$M_i$ therefore computes an $(\eps, k, z')$-mini-ball covering of $P_i$, and sends it to the coordinator.
By Lemma~\ref{lem:union:coreset} the union of the received mini-ball coverings will be an $(\eps, k, z)$-mini-ball covering of $P$,
with high probability.
The coordinator then reports an $(\eps, k, z)$-mini-ball covering of this union as the final coreset.

%------------------------------------------------------------------------------------------
\begin{algorithm}[h] 
\caption{{\sc{RandomizedMPC}}: A randomized $1$-round algorithm to compute an $(\eps,k,z)$-coreset} 
\label{alg:1:round}
\textbf{Round 1, executed by each machine $M_i$:} \\
\emph{Computation:}
\begin{algorithmic}[1]
\State $z' \gets \min(\outliersBound, z)$.
    \State $P_i^* \gets $ \computeCoreset$(P_i, k, z', \eps)$.
\end{algorithmic}

\emph{Communication:}
\begin{algorithmic}[1]
\State Send $P_i^*$ to the coordinator.
\end{algorithmic}

\vspace*{2mm}
\textbf{At the coordinator:}  Collect all mini-ball coverings~$P_i^*$ and report \computeCoreset$(\bigcup_i P_i^*,k,z,\eps)$ as the final mini-ball covering.
% \mdb{Not sure if this is the right way to write it, but somewhere we should say what the final coreset is.}
\end{algorithm}
%------------------------------------------------------------------------------------------
Consider an optimal solution for the $k$-center problem with $z$ outliers on~$P$.
Let $\Bopt$ be the set of $k$~balls in this optimal solution and let $\Pout\subset P$ 
be the outliers, that is, the points not covered by the balls in~$\Bopt$.
Lemma \ref{lem:bound:outliers:per:machine} states the number of outliers that are assigned to 
% an arbitrary machine $M_i$
each machine
is concentrated around its expectation.
The lemma was already observed by Ceccarello~\etal~\cite[Lemma~7]{DBLP:journals/pvldb/CeccarelloPP19}, 
but we present the proof for completeness in the appendix.
%------------------------------------------------------------------------------------------
\begin{lemma}[\cite{DBLP:journals/pvldb/CeccarelloPP19}]
\label{lem:bound:outliers:per:machine}
% For each machine $M_i$ we have 
$\Pr{\forall_{1\leq i \leq m} |P_i \cap \Pout|  \leq \outliersBound} \geq 1-1/n^2$.
\end{lemma}
%--------------------------------------------------------------------------------------
Now we can prove that Algorithm \ref{alg:1:round} computes a coreset for $k$-center with $z$ outliers in a single round. 
\begin{theorem}[Randomized $1$-Round  Algorithm]
\label{thm:1:round}
Let $P \subseteq X$ be a point set of size $n$ in a metric space $(X,\dist)$ of doubling dimension~$d$. 
Let $k,z \in \mathbb{N}$ be two natural numbers, and let $0 < \eps \le 1$ be an error parameter. 
Assuming~$P$ is initially distributed randomly over the machines,
there exists a randomized algorithm that computes an $(\eps,k,z)$-coreset of $P$ in the MPC model in one round of communication,
using $m= O(\sqrt{n\eps^d/k})$ worker machines with $O(\sqrt{nk/\eps^d})$ local memory,
and a coordinator with 
$O(\sqrt{nk/\eps^d} + \sqrt{n\eps^d/k}\cdot\min(\log{n}, z)+z)$ local memory.
\end{theorem}

% \begin{recall:thm:1:round}[Randomized $1$-Round  Algorithm]
% Let $P \subseteq X$ be a point set of size $n$ in a metric space $(X,\dist)$ of doubling dimension~$d$. 
% Let $k,z \in \mathbb{N}$ be two natural numbers, and let $0 < \eps \le 1$ be an error parameter. 
% Assuming~$P$ is initially distributed randomly over the machines,
% there exists a randomized algorithm that computes an $(\eps,k,z)$-coreset of $P$ in the MPC model in one round of communication,
% using $m= O(\sqrt{n\eps^d/k})$ worker machines with $O(\sqrt{nk/\eps^d})$ local memory,
% and a coordinator with 
% $O(\sqrt{nk/\eps^d} + \sqrt{n\eps^d/k}\cdot\min(\log{n}, z)+z)$ local memory.
% \end{recall:thm:1:round}

\begin{proof}
In Algorithm \ref{alg:1:round}, each machine $M_i$ sends the coordinator a weighted point 
set~$P_i^*$, which is an $(\eps, k, z')$-mini-ball covering of $P_i$.
Recall that $z'=\min(\outliersBound, z)$.
Lemma \ref{lem:bound:outliers:per:machine} shows that with high probability, 
at most $\outliersBound$ outliers are assigned to each machine. Trivially, 
at most $z$ outliers can be assigned to a single machine, so with high probability at most
$z'$ outliers are assigned to each machine.
Hence, with high probability, $\opt_{k,z'}(P_i) \leq \optkz(P)$ for each $i \in [m]$. 
Lemma~\ref{lem:union:coreset} then implies that $\cup^m_{i=1} P_i^*$ is an $(\eps, k, z)$-mini-ball covering of $P$.
To report the final coreset, the coordinator computes an $(\eps, k, z)$-mini-ball covering of $\cup^{m}_{i=1} P_i^*$, which is an $(\eps', k, z)$-mini-ball covering of $P$ by Lemma \ref{lem:merge:coreset},
and therefore an $(\eps', k, z)$-coreset of $P$ by Lemma \ref{lem:MBC:is:coreset} , where $\eps'=3\eps$.

Next we discuss storage usage. The points are distributed randomly among $m= O(\sqrt{n\eps^d/k})$
machines. Applying the Chernoff bound and the union bound in the same way as in the proof of
Lemma~\ref{lem:bound:outliers:per:machine}, it follows that at most $\frac{6n}{m}+3\log{n} = O(\frac{n}{m})$ 
points are allocated to each machine with high probability.
% with probability at least $1-1/n^2$.
Thus, each worker machine needs $O(\frac{n}{m}) = O(\sqrt{nk/\eps^d})$ local memory to store
the points and compute a mini-ball covering for them.
The coordinator receives $m$~mini-ball coverings, and
according to Lemma \ref{lem:compress}, each mini-ball covering is of size at most
$k(\frac{12}{\eps})^d + z' = O(k/\eps^d +z')$. 
(Recall that we consider $d$ to be a constant.)
Therefore, the storage required by the coordinator is
\[
\begin{array}{lll}
m \cdot O(k/\eps^d +z')
& = & O\left(\sqrt{n\eps^d/k} \cdot \frac{k}{\eps^d}\right) + m \cdot \min(\outliersBound, z)
\\
%  = O\left(\sqrt{n\gparameter}\right) + O(z) + m \cdot \min(3\log{n}, z)
& = & O\left(\sqrt{nk/\eps^d} + \sqrt{n\eps^d/k}\cdot\min(\log{n}, z) + z\right)
\enspace . \qedhere
\end{array}
\]
\end{proof}
%------------------------------------------------------------------------------------------------
\subsection{A deterministic $\rrounds$-round MPC algorithm}
\label{subsec:R:round}
%------------------------------------------------------------------------------------------------
We present a deterministic multi-round algorithm in the MPC model for the $k$-center problem
with $z$ outliers. 
It shows how to obtain a trade-off between the number of rounds
and the local storage.
Our algorithm is parameterized by $\rrounds$, the number of rounds of 
communication we are willing to use; the larger $\rrounds$, the smaller amount of storage per machine.
% Let $M_1,M_2,\ldots,M_m$ be the machines we are using.
Initially, the input point set~$P$ is distributed arbitrarily (but evenly) over the machines. 
% Let $P_i$ denote the point set of machine~$M_i$. 

All machines are active in the first round. In every subsequent round, the number 
of active machines reduces by a factor $\beta$, where $\beta = \ceil{m^{1/R}}$.
Note that this implies that after $R$ rounds, we are left with a single active machine $M_1$, which is the coordinator.
% \mdb{Is the last machine the coordinator? How does the model used here relate
% to the model we had before?}

As shown in Algorithm \ref{alg:R:round}, in each round, every active machine $M_i$ computes an
$(\eps, k, z)$-mini-ball covering on the union of sets that is sent to it in the previous round, 
and then sends it to machine $M_{\ceil{i/\beta}}$.

%------------------------------------------------------------------------------------------
\begin{algorithm}[h] 
\caption{A deterministic multi-round algorithm to compute $((1+\eps)^{\rrounds}-1,k,z)$-coreset} 
\label{alg:R:round}

\textbf{Round $t$, executed by each active machine $M_i$ $(1\leq i \leq{\ceil{m/\beta^{t-1}}})$:} \\
\emph{Computation:}
\begin{algorithmic}[1]
\State Let $Q_i$ be the union of sets that $M_i$ received. %(from $M_{\beta(i-1)+1},\cdots,M_{\beta i}$)
\State $Q_i^* \gets$ \computeCoreset$(Q_i, k, z,\eps)$.
\end{algorithmic}

\emph{Communication:}
\begin{algorithmic}[1]
\State Send $Q_i^*$ to $M_{\ceil{i/\beta}}$.
\end{algorithmic}
\end{algorithm}
%------------------------------------------------------------------------------------------

\medskip
We first prove that machine $M_1$ receives a $((1+\eps)^{R}-1, k, z)$-coreset after $R$ rounds.

%------------------------------------------------------------------------------------------
\begin{lemma}
The union of sets that machine $M_1$ receives after executing algorithm \ref{alg:R:round}
is a $((1+\eps)^R-1, k, z)$-coreset of $P$.

\end{lemma}
%------------------------------------------------------------------------------------------
\begin{proof}
We prove by induction that
for each $0 \leq t \leq \rrounds$, and for each $i \in [\ceil{m/{\beta^t}}]$,
the union of sets that machine $M_i$ receives after round $t$ 
% \mdb{This is just $Q_i$, right? Then it's better to just write that so that it is
% more in line with the algorithm.} \lb{No, it is not!}
is a $((1+\eps)^t-1, k, z)$-mini-ball covering of $P_{\beta^{t}(i-1)+1}\cup\cdots\cup P_{\beta^{t} i}$.
Recall that $P_i\subset P$ denotes the set of points initially stored in machine~$M_i$.

We prove this lemma by induction. The base case is $t=0$. 
% Before the first round of the algorithm,
As $P_i$ is a $(0, k, z)$-mini-ball covering for $P_i$, the lemma trivially holds for $t=0$.
The induction hypothesis is that the lemma 
holds for $t-1$. We show that then it holds for $t$ too. 

Let $i \in [\ceil{m/{\beta^t}}]$, and $j$ be an arbitrary integer such that 
$\beta(i-1) + 1 \leq j \leq \min(\beta i, m)$, so, $\ceil{j/\beta} = i$.
Let $S_j$ be the union of sets that machine $M_j$ receives after round $t-1$. 
The induction hypothesis says that $S_j$ is a $((1+\eps)^{t-1}-1, k, z)$-mini-ball covering of 
$P_{\beta^{t-1}(j-1)+1}\cup \cdots \cup P_{\beta^{t-1} j}$. In round $t$, machine $M_j$ computes
an $(\eps, k, z)$-mini-ball covering of $S_j$ and send it to $M_i$. 
We refer to this mini-ball covering as $S^*_j$.
Using Lemma \ref{lem:merge:coreset} with set $\peps=(1+\eps)^{t-1}-1$ implies 
$S^*_j$ is an $(\eps+\peps+\eps\peps, k, z)$-mini-ball covering 
of $P_{\beta^{t-1}(j-1)+1}\cup \cdots \cup P_{\beta^{t-1} j}$.
We have
$$
\eps +\peps + \eps\peps = 
\eps + (1+\eps)^{t-1}-1 + \eps\cdot((1+\eps)^{t-1}-1)
=(1+\eps)^{t-1}(1+\eps)-1 
= (1+\eps)^{t}-1
\enspace .
$$
Therefore, $S^*_j$ is a $((1+\eps)^{t}-1, k, z)$-mini-ball covering
of $P_{\beta^{t-1}(j-1)+1}\cup \cdots \cup P_{\beta^{t-1} j}$. After round $t$, $M_i$ receives 
$\cup_{j=\beta(i-1)+1}^{\beta i} S^*_j$. As each $S_j^*$ is a $((1+\eps)^{t}-1, k, z)$-mini-ball covering, 
Lemma \ref{lem:union:coreset} implies that set $\cup_{j=\beta(i-1)+1}^{\beta i} S^*_j$
that $M_i$ receives after round $t$, is a $((1+\eps)^{t}-1, k, z)$-mini-ball covering 
for $P_{\beta^{t}(i-1)+1}\cup\cdots\cup P_{\beta^{t} i}$.

Thus, the union of sets that $M_1$ receives after $\rrounds$ rounds is a $((1+\eps)^{\rrounds}-1, k, z)$-mini-ball covering of $P$, and then a $((1+\eps)^{\rrounds}-1, k, z)$-coreset of $P$ by Lemma \ref{lem:MBC:is:coreset}.
\end{proof}
%------------------------------------------------------------------------------------------

Now, we state our result for $\rrounds$ rounds.
%------------------------------------------------------------------------------------------------
\begin{theorem}[Deterministic $\rrounds$-round Algorithm]
\label{thm:R:round}
Let $P \subseteq X$ be a point set of size $n$ in a metric space $(X,\dist)$ of doubling dimension $d$. 
Let $k,z \in \mathbb{N}$ be two natural numbers, and let $0 < \eps \leq 1$ be an error parameter. 
Then there exists a deterministic algorithm that computes a $((1+\eps)^{\rrounds}-1, k, z)$-coreset of $P$
% \mdb{I think it would be even better if we would write
% $(1\pm \eps)$. Then we have to run the algorithm with a suitable $\eps'$ 
% (would $\eps'=\Theta(\eps/R)$ work?).
% This may also show up in the storage bound. But having a final approximation bound of $(1\pm \eps)$ is cleaner. }
% \lb{We did not do that because then $R^d$ will show up in the space!}
in the MPC model in $\rrounds$ rounds of communication using
$m=O\left((\frac{n}{k/\eps^d +z})^{\rrounds/(\rrounds+1)}\right)$ machines
and $O(n^{1/(\rrounds+1)}(k/\eps^d +z)^{\rrounds/(\rrounds+1)})$
storage per machine. 
\end{theorem}
%------------------------------------------------------------------------------------------------

% \begin{theorem}[Deterministic $\rrounds$-round Algorithm]
% \label{thm:R:round}
% Let $P \subseteq X$ be a point set of size $n$ in a metric space $(X,\dist)$ of doubling dimension $d$. 
% Let $k,z \in \mathbb{N}$ be two natural numbers, and let $0 < \eps \leq 1$ be an error parameter. 
% Then there exists a deterministic algorithm that computes a $((1+\eps)^{\rrounds}-1, k, z)$-coreset of $P$
% % \mdb{I think it would be even better if we would write
% % $(1\pm \eps)$. Then we have to run the algorithm with a suitable $\eps'$ 
% % (would $\eps'=\Theta(\eps/R)$ work?).
% % This may also show up in the storage bound. But having a final approximation bound of $(1\pm \eps)$ is cleaner. }
% % \lb{We did not do that because then $R^d$ will show up in the space!}
% in the MPC model in $\rrounds$ rounds of communication using
% $m=O\left((\frac{n}{k/\eps^d +z})^{\rrounds/(\rrounds+1)}\right)$ machines
% and $O(n^{1/(\rrounds+1)}(k/\eps^d +z)^{\rrounds/(\rrounds+1)})$
% storage per machine. 
% \end{theorem}
%------------------------------------------------------------------------------------------

\begin{proof}
By Lemma \ref{lem:2:round:alg}, invoking Algorithm \ref{alg:R:round} results in a 
$((1+\eps)^{\rrounds}-1, k, z)$-coreset of $P$ on $M_1$.
Next we discuss the required storage. In the first round, each machine~$M_i$ needs 
$O(\frac{n}{m}) = O(n^{1/(\rrounds+1)}(k/\eps^d +z)^{\rrounds/(\rrounds+1)})$ storage
for $P_i$ (and to compute a coreset for it).
In each subsequent round, every active machine receives $\beta$ coresets. 
By Lemma~\ref{lem:compress} each coreset is of size at most $k(\frac{12}{\eps})^d+z = O(k/\eps^d +z)$. 
Since $\beta= \ceil{m^{1/\rrounds}}$, the storage per machine is 
\[
\beta \cdot O\left(k/\eps^d+z\right) 
= m^{1/\rrounds} \cdot O\left( k/\eps^d+z\right)
= O\left( n^{1/(\rrounds+1)} (k/\eps^d +z)^{\rrounds/(\rrounds+1)} \right) \enspace .
\qedhere
\]
\end{proof}
%------------------------------------------------------------------------------------------
% \mdb{Question: Can we also do the R-round approach for the randomized version?}

%--------------------------------------------------------------------------------------------------
%\subsection{More MPC algorithms}
% \label{subsec:1:round}
%--------------------------------------------------------------------------------------------------
%\lb{TO DO: Add an overview}

%------------------------------------------------------------------------------------------------
% \section{Deterministic $\rrounds$-round MPC algorithm}
% \label{appendix:subsec:R:round}
%------------------------------------------------------------------------------------------------

\newcommand{\Proc}{Proceedings of the~}
\newcommand{\AMS}{Annals of Mathematical Statistics}

\bibliography{references}

\begin{thebibliography}{10}

\bibitem{DBLP:journals/jacm/AgarwalHV04}
Pankaj~K. Agarwal, Sariel Har{-}Peled, and Kasturi~R. Varadarajan.
\newblock Approximating extent measures of points.
\newblock {\em J. {ACM}}, 51(4):606--635, 2004.
\newblock \href {https://doi.org/10.1145/1008731.1008736}
  {\path{doi:10.1145/1008731.1008736}}.

\bibitem{DBLP:books/crc/aggarwal2013}
Charu~C. Aggarwal and Chandan~K. Reddy, editors.
\newblock {\em Data Clustering: Algorithms and Applications}.
\newblock {CRC} Press, 2014.
\newblock URL: \url{http://www.crcpress.com/product/isbn/9781466558212}.

\bibitem{DBLP:journals/jcss/AlonMS99}
Noga Alon, Yossi Matias, and Mario Szegedy.
\newblock The space complexity of approximating the frequency moments.
\newblock {\em J. Comput. Syst. Sci.}, 58(1):137--147, 1999.
\newblock \href {https://doi.org/10.1006/jcss.1997.1545}
  {\path{doi:10.1006/jcss.1997.1545}}.

\bibitem{DBLP:journals/tcs/BarkayPS15}
Neta Barkay, Ely Porat, and Bar Shalem.
\newblock Efficient sampling of non-strict turnstile data streams.
\newblock {\em Theor. Comput. Sci.}, 590:106--117, 2015.
\newblock \href {https://doi.org/10.1016/j.tcs.2015.01.026}
  {\path{doi:10.1016/j.tcs.2015.01.026}}.

\bibitem{10.1007/978-981-16-4641-6_21}
Ritika Bateja, Sanjay~Kumar Dubey, and Ashutosh Bhatt.
\newblock Evaluation and application of clustering algorithms in healthcare
  domain using cloud services.
\newblock In {\em Proc.~2nd International Conference on Sustainable
  Technologies for Computational Intelligence}, pages 249--261, 2022.

\bibitem{DBLP:journals/corr/abs-2112-07050}
MohammadHossein Bateni, Hossein Esfandiari, Rajesh Jayaram, and Vahab~S.
  Mirrokni.
\newblock Optimal fully dynamic k-centers clustering.
\newblock {\em CoRR}, abs/2112.07050, 2021.
\newblock URL: \url{https://arxiv.org/abs/2112.07050}, \href
  {http://arxiv.org/abs/2112.07050} {\path{arXiv:2112.07050}}.

\bibitem{DBLP:journals/jacm/BeameKS17}
Paul Beame, Paraschos Koutris, and Dan Suciu.
\newblock Communication steps for parallel query processing.
\newblock {\em J. {ACM}}, 64(6):40:1--40:58, 2017.
\newblock \href {https://doi.org/10.1145/3125644} {\path{doi:10.1145/3125644}}.

\bibitem{DBLP:books/lib/Bishop07}
Christopher~M. Bishop.
\newblock {\em Pattern recognition and machine learning, 5th Edition}.
\newblock Information science and statistics. Springer, 2007.
\newblock URL: \url{https://www.worldcat.org/oclc/71008143}.

\bibitem{DBLP:conf/icip/CandesR06}
Emmanuel~J. Cand{\`{e}}s and Justin~K. Romberg.
\newblock Robust signal recovery from incomplete observations.
\newblock In {\em Proceedings of the International Conference on Image
  Processing, {ICIP} 2006, October 8-11, Atlanta, Georgia, {USA}}, pages
  1281--1284. {IEEE}, 2006.
\newblock \href {https://doi.org/10.1109/ICIP.2006.312579}
  {\path{doi:10.1109/ICIP.2006.312579}}.

\bibitem{DBLP:journals/tit/CandesRT06}
Emmanuel~J. Cand{\`{e}}s, Justin~K. Romberg, and Terence Tao.
\newblock Robust uncertainty principles: exact signal reconstruction from
  highly incomplete frequency information.
\newblock {\em {IEEE} Trans. Inf. Theory}, 52(2):489--509, 2006.
\newblock \href {https://doi.org/10.1109/TIT.2005.862083}
  {\path{doi:10.1109/TIT.2005.862083}}.

\bibitem{DBLP:journals/pvldb/CeccarelloPP19}
Matteo Ceccarello, Andrea Pietracaprina, and Geppino Pucci.
\newblock Solving k-center clustering (with outliers) in mapreduce and
  streaming, almost as accurately as sequentially.
\newblock {\em Proc. {VLDB} Endow.}, 12(7):766--778, 2019.
\newblock \href {https://doi.org/10.14778/3317315.3317319}
  {\path{doi:10.14778/3317315.3317319}}.

\bibitem{DBLP:journals/tkde/ChanGHS22}
T.{-}H.~Hubert Chan, Arnaud Guerquin, Shuguang Hu, and Mauro Sozio.
\newblock Fully dynamic {\textdollar}k{\textdollar}k-center clustering with
  improved memory efficiency.
\newblock {\em {IEEE} Trans. Knowl. Data Eng.}, 34(7):3255--3266, 2022.
\newblock \href {https://doi.org/10.1109/TKDE.2020.3023020}
  {\path{doi:10.1109/TKDE.2020.3023020}}.

\bibitem{DBLP:journals/siamcomp/CharikarCFM04}
Moses Charikar, Chandra Chekuri, Tom{\'{a}}s Feder, and Rajeev Motwani.
\newblock Incremental clustering and dynamic information retrieval.
\newblock {\em {SIAM} J. Comput.}, 33(6):1417--1440, 2004.
\newblock \href {https://doi.org/10.1137/S0097539702418498}
  {\path{doi:10.1137/S0097539702418498}}.

\bibitem{DBLP:conf/soda/CharikarKMN01}
Moses Charikar, Samir Khuller, David~M. Mount, and Giri Narasimhan.
\newblock Algorithms for facility location problems with outliers.
\newblock In {\em Proc.~12th Annual Symposium on Discrete Algorithms (SODA)},
  pages 642--651, 2001.
\newblock URL: \url{http://dl.acm.org/citation.cfm?id=365411.365555}.

\bibitem{Cher}
H.~Chernoff.
\newblock A measure of asymptotic efficiency for tests of a hypothesis based on
  the sum of observations.
\newblock {\em \AMS}, 23(4):493 -- 507, 1952.

\bibitem{DBLP:conf/icalp/Cohen-AddadSS16}
Vincent Cohen{-}Addad, Chris Schwiegelshohn, and Christian Sohler.
\newblock Diameter and k-center in sliding windows.
\newblock In {\em Proc.~43rd International Colloquium on Automata, Languages,
  and Programming, (ICALP 2016)}, volume~55 of {\em LIPIcs}, pages 19:1--19:12,
  2016.
\newblock \href {https://doi.org/10.4230/LIPIcs.ICALP.2016.19}
  {\path{doi:10.4230/LIPIcs.ICALP.2016.19}}.

\bibitem{DBLP:journals/rsa/CzumajS07}
Artur Czumaj and Christian Sohler.
\newblock Sublinear-time approximation algorithms for clustering via random
  sampling.
\newblock {\em Random Struct. Algorithms}, 30(1-2):226--256, 2007.
\newblock \href {https://doi.org/10.1002/rsa.20157}
  {\path{doi:10.1002/rsa.20157}}.

\bibitem{DBLP:conf/esa/BergMZ21}
Mark de~Berg, Morteza Monemizadeh, and Yu~Zhong.
\newblock $k$-{C}enter clustering with outliers in the sliding-window model.
\newblock In {\em Proc.~29th Annual European Symposium on Algorithms (ESA
  2021)}, volume 204 of {\em LIPIcs}, pages 13:1--13:13, 2021.
\newblock \href {https://doi.org/10.4230/LIPIcs.ESA.2021.13}
  {\path{doi:10.4230/LIPIcs.ESA.2021.13}}.

\bibitem{DBLP:journals/corr/abs-2109-11853}
Mark de~Berg, Morteza Monemizadeh, and Yu~Zhong.
\newblock k-center clustering with outliers in the sliding-window model.
\newblock {\em CoRR}, abs/2109.11853, 2021.
\newblock URL: \url{https://arxiv.org/abs/2109.11853}, \href
  {http://arxiv.org/abs/2109.11853} {\path{arXiv:2109.11853}}.

\bibitem{DHANACHANDRA2015764}
Nameirakpam Dhanachandra, Khumanthem Manglem, and Yambem~Jina Chanu.
\newblock Image segmentation using k-means clustering algorithm and subtractive
  clustering algorithm.
\newblock {\em Procedia Computer Science}, 54:764--771, 2015.
\newblock \href {https://doi.org/https://doi.org/10.1016/j.procs.2015.06.090}
  {\path{doi:https://doi.org/10.1016/j.procs.2015.06.090}}.

\bibitem{DBLP:conf/esa/DingYW19}
Hu~Ding, Haikuo Yu, and Zixiu Wang.
\newblock Greedy strategy works for k-center clustering with outliers and
  coreset construction.
\newblock In {\em Prof.~27th Annual European Symposium on Algorithms (ESA
  2019)}, volume 144 of {\em LIPIcs}, pages 40:1--40:16, 2019.
\newblock \href {https://doi.org/10.4230/LIPIcs.ESA.2019.40}
  {\path{doi:10.4230/LIPIcs.ESA.2019.40}}.

\bibitem{DBLP:conf/kdd/EneIM11}
Alina Ene, Sungjin Im, and Benjamin Moseley.
\newblock Fast clustering using mapreduce.
\newblock In {\em Proc.~17th {ACM} {SIGKDD} International Conference on
  Knowledge Discovery and Data Mining}, pages 681--689, 2011.
\newblock \href {https://doi.org/10.1145/2020408.2020515}
  {\path{doi:10.1145/2020408.2020515}}.

\bibitem{everitt2011cluster}
B.S. Everitt, S.~Landau, M.~Leese, and D.~Stahl.
\newblock {\em Cluster Analysis}.
\newblock Wiley Series in Probability and Statistics. Wiley, 2011.
\newblock URL: \url{https://books.google.nl/books?id=w3bE1kqd-48C}.

\bibitem{DBLP:conf/stoc/FederG88}
Tom{\'{a}}s Feder and Daniel~H. Greene.
\newblock Optimal algorithms for approximate clustering.
\newblock In {\em Proc.~20th Annual {ACM} Symposium on Theory of Computing
  (STOC 1988)}, pages 434--444, 1988.
\newblock \href {https://doi.org/10.1145/62212.62255}
  {\path{doi:10.1145/62212.62255}}.

\bibitem{DBLP:conf/stoc/FrahlingS05}
Gereon Frahling and Christian Sohler.
\newblock Coresets in dynamic geometric data streams.
\newblock In {\em Proc.~37th Annual {ACM} Symposium on Theory of Computing
  (STOC 2005)}, pages 209--217, 2005.
\newblock \href {https://doi.org/10.1145/1060590.1060622}
  {\path{doi:10.1145/1060590.1060622}}.

\bibitem{DBLP:journals/tcs/Gonzalez85}
Teofilo~F. Gonzalez.
\newblock Clustering to minimize the maximum intercluster distance.
\newblock {\em Theor. Comput. Sci.}, 38:293--306, 1985.
\newblock \href {https://doi.org/10.1016/0304-3975(85)90224-5}
  {\path{doi:10.1016/0304-3975(85)90224-5}}.

\bibitem{DBLP:conf/isaac/GoodrichSZ11}
Michael~T. Goodrich, Nodari Sitchinava, and Qin Zhang.
\newblock Sorting, searching, and simulation in the mapreduce framework.
\newblock In {\em Proc.~22nd International Symposium on Algorithms and
  Computation (ISAAC 2011)}, volume 7074 of {\em Lecture Notes in Computer
  Science}, pages 374--383, 2011.
\newblock \href {https://doi.org/10.1007/978-3-642-25591-5\_39}
  {\path{doi:10.1007/978-3-642-25591-5\_39}}.

\bibitem{DBLP:conf/alenex/GoranciHLSS21}
Gramoz Goranci, Monika Henzinger, Dariusz Leniowski, Christian Schulz, and
  Alexander Svozil.
\newblock Fully dynamic \emph{k}-center clustering in low dimensional metrics.
\newblock In Martin Farach{-}Colton and Sabine Storandt, editors, {\em
  Proceedings of the Symposium on Algorithm Engineering and Experiments,
  {ALENEX} 2021, Virtual Conference, January 10-11, 2021}, pages 143--153.
  {SIAM}, 2021.
\newblock \href {https://doi.org/10.1137/1.9781611976472.11}
  {\path{doi:10.1137/1.9781611976472.11}}.

\bibitem{DBLP:journals/topc/GuhaLZ19}
Sudipto Guha, Yi~Li, and Qin Zhang.
\newblock Distributed partial clustering.
\newblock {\em {ACM} Trans. Parallel Comput.}, 6(3):11:1--11:20, 2019.
\newblock \href {https://doi.org/10.1145/3322808} {\path{doi:10.1145/3322808}}.

\bibitem{DBLP:conf/stoc/Har-PeledM04}
Sariel Har{-}Peled and Soham Mazumdar.
\newblock On coresets for k-means and k-median clustering.
\newblock In {\em Proc.~36th Annual {ACM} Symposium on Theory of Computing
  (STOC 2004)}, pages 291--300, 2004.
\newblock \href {https://doi.org/10.1145/1007352.1007400}
  {\path{doi:10.1145/1007352.1007400}}.

\bibitem{DBLP:conf/stoc/Indyk04}
Piotr Indyk.
\newblock Algorithms for dynamic geometric problems over data streams.
\newblock In {\em Proc.~36th Annual {ACM} Symposium on Theory of Computing
  (STOC 2004)}, pages 373--380, 2004.
\newblock \href {https://doi.org/10.1145/1007352.1007413}
  {\path{doi:10.1145/1007352.1007413}}.

\bibitem{DBLP:conf/pods/KaneNW10}
Daniel~M. Kane, Jelani Nelson, and David~P. Woodruff.
\newblock An optimal algorithm for the distinct elements problem.
\newblock In {\em Proc.~29 {ACM} {SIGMOD-SIGACT-SIGART} Symposium on Principles
  of Database Systems (PODS 2010)}, pages 41--52, 2010.
\newblock \href {https://doi.org/10.1145/1807085.1807094}
  {\path{doi:10.1145/1807085.1807094}}.

\bibitem{DBLP:conf/soda/KarloffSV10}
Howard~J. Karloff, Siddharth Suri, and Sergei Vassilvitskii.
\newblock A model of computation for mapreduce.
\newblock In {\em Proc.~21st Annual {ACM-SIAM} Symposium on Discrete Algorithms
  (SODA 2010)}, pages 938--948. {SIAM}, 2010.
\newblock \href {https://doi.org/10.1137/1.9781611973075.76}
  {\path{doi:10.1137/1.9781611973075.76}}.

\bibitem{DBLP:conf/approx/McCutchenK08}
Richard~Matthew McCutchen and Samir Khuller.
\newblock Streaming algorithms for k-center clustering with outliers and with
  anonymity.
\newblock In {\em Proc.~11th and 12th International Workshop on Approximation,
  Randomization and Combinatorial Optimization (APPROX and RANDOM)}, volume
  5171 of {\em Lecture Notes in Computer Science}, pages 165--178, 2008.
\newblock \href {https://doi.org/10.1007/978-3-540-85363-3\_14}
  {\path{doi:10.1007/978-3-540-85363-3\_14}}.

\bibitem{DBLP:conf/soda/MonemizadehW10}
Morteza Monemizadeh and David~P. Woodruff.
\newblock 1-pass relative-error l\({}_{\mbox{p}}\)-sampling with applications.
\newblock In {\em Proceedings of the Twenty-First Annual {ACM-SIAM} Symposium
  on Discrete Algorithms, {SODA} 2010, Austin, Texas, USA, January 17-19,
  2010}, pages 1143--1160, 2010.
\newblock \href {https://doi.org/10.1137/1.9781611973075.92}
  {\path{doi:10.1137/1.9781611973075.92}}.

\bibitem{DBLP:conf/approx/NelsonNW12}
Jelani Nelson, Huy~L. Nguy{\^{e}}n, and David~P. Woodruff.
\newblock On deterministic sketching and streaming for sparse recovery and norm
  estimation.
\newblock In Anupam Gupta, Klaus Jansen, Jos{\'{e}} D.~P. Rolim, and Rocco~A.
  Servedio, editors, {\em Approximation, Randomization, and Combinatorial
  Optimization. Algorithms and Techniques - 15th International Workshop,
  {APPROX} 2012, and 16th International Workshop, {RANDOM} 2012, Cambridge, MA,
  USA, August 15-17, 2012. Proceedings}, volume 7408 of {\em Lecture Notes in
  Computer Science}, pages 627--638. Springer, 2012.
\newblock \href {https://doi.org/10.1007/978-3-642-32512-0\_53}
  {\path{doi:10.1007/978-3-642-32512-0\_53}}.

\bibitem{DBLP:journals/sj/PaekK17}
Jeongyeup Paek and JeongGil Ko.
\newblock K-means clustering-based data compression scheme for wireless imaging
  sensor networks.
\newblock {\em {IEEE} Syst. J.}, 11(4):2652--2662, 2017.
\newblock \href {https://doi.org/10.1109/JSYST.2015.2491359}
  {\path{doi:10.1109/JSYST.2015.2491359}}.

\bibitem{DBLP:conf/focs/PriceW11}
Eric Price and David~P. Woodruff.
\newblock {(1} + eps)-approximate sparse recovery.
\newblock In Rafail Ostrovsky, editor, {\em {IEEE} 52nd Annual Symposium on
  Foundations of Computer Science, {FOCS} 2011, Palm Springs, CA, USA, October
  22-25, 2011}, pages 295--304. {IEEE} Computer Society, 2011.
\newblock \href {https://doi.org/10.1109/FOCS.2011.92}
  {\path{doi:10.1109/FOCS.2011.92}}.

\bibitem{Priebe5995}
Carey~E. Priebe, Youngser Park, Joshua~T. Vogelstein, John~M. Conroy, Vince
  Lyzinski, Minh Tang, Avanti Athreya, Joshua Cape, and Eric Bridgeford.
\newblock On a two-truths phenomenon in spectral graph clustering.
\newblock {\em Proceedings of the National Academy of Sciences},
  116(13):5995--6000, 2019.
\newblock URL: \url{https://www.pnas.org/content/116/13/5995}, \href
  {https://doi.org/10.1073/pnas.1814462116}
  {\path{doi:10.1073/pnas.1814462116}}.

\bibitem{DBLP:conf/ipcv/SchmiederCL09}
Anthony Schmieder, Howard Cheng, and Xiaobo Li.
\newblock A study of clustering algorithms and validity for lossy image set
  compression.
\newblock In {\em Proc.~2009 International Conference on Image Processing,
  Computer Vision, {\&} Pattern Recognition, (IPCV 2009)}, pages 501--506,
  2009.

\end{thebibliography}

%------------------------------------------------------------------------------------------

\appendix

%------------------------------------------------------------------------------------------
\section{Omitted proofs}
%------------------------------------------------------------------------------------------

\subsection*{Proof of Lemma \ref{lem:MBC:is:coreset}}

\newtheorem*{recall:lem:MBC:is:coreset}{Lemma~\ref{lem:MBC:is:coreset}}

\begin{recall:lem:MBC:is:coreset}
Let $P$ be a weighted point set in a metric space $(X, \dist)$ and
let $P^*$ be an $(\eps, k, z)$-mini-ball covering of $P$. Then, $P^*$ is an $(\eps, k, z)$-coreset of $P$.
\end{recall:lem:MBC:is:coreset}

\begin{proof}
Let $P^*$ be an $(\eps, k, z)$-mini-ball covering of $P$.
First, we prove the second condition of coreset holds for $P^*$.
Let $B = \{\ball(c_1,r),\cdots,\ball(c_k,r)\}$ be any set of congruent balls 
in the space $(X,\dist)$ such that the sum 
of weights of points in $P^*$ that are not covered by $B$ is at most $z$.  
Let $r' := r + \eps \cdot \optkz(P)$, and $B' = \{ \ball(c_1,r'),\cdots ,\ball (c_k,r')\}$.
% To prove that $P^*$ is an $(\eps, k, z)$-coreset of $P$ 
We show that the total weight of points of $P$ that are not covered by $B'$ is at most $z$.
Let $p\in P$ be an arbitrary point.
Note that if $p$ is not covered by a ball from~$B'$,
then its representative $q \in P^*$ cannot be covered by any ball from~$B$;
this follows from $\dist(p,q))\leq \eps \cdot \optkz(P)$ and the triangle inequality.
Thus the total weight of the point $p\in P$ not covered by $B'$ is at most
the total weight of the points $p\in P^*$ not covered by $B$, which is at most~$z$.

Next, we prove the first condition of the coreset also holds for mini-ball covering $P^*$.
% show that $(1-\eps)\cdot \optkz(P) \leq \optkz(P^*) \leq (1+\eps)\cdot \optkz(P)$.
Let $\mathcal{B^*}$ be an optimal set of balls for $P^*$, that is, a set of $k$ congruent
balls of minimum radius covering all points from $P^*$ except for some outliers of total weight at most~$z$.
It follows from the second condition of coreset which we just proved that holds for mini-ball covering $P^*$ that if we expand the radius of the balls in~$\mathcal{B^*}$
by $\eps\cdot\optkz(P)$, then the expanded balls are a feasible solution for~$P$.
Hence, $(1-\eps)\cdot \optkz(P) \leq \optkz(P^*)$.
To prove that $\optkz(P^*)\leq (1+\eps)\cdot \optkz(P)$, let $\B$ be an optimal set
of balls for $P$. Expand the radius of these balls by $\eps\cdot\optkz(P)$. It suffices
to show that the set $\B^*$ of expanded balls forms a feasible solution for~$P^*$. 
This is true by a similar argument as above: 
if $p^*\in P^*$ is not covered by~$\B^*$,
then it follows from triangle inequality and the fact that the distance between $p^*$ and each point represented by~$p^*$ is at most $\eps \cdot \optkz(P)$
that none of the points represented by~$p^*$ can be covered by $\B$, 
% \mdb{This is false, I think.} \lb{I think it is correct now.} 
and so the total
weight of the points $p^*\in P^*$ not covered by~$\B^*$ is at most~$z$.

This means that $P^*$ is an $(\eps, k, z)$-coreset of $P$.
\end{proof}

%------------------------------------------------------------------------------------------

\subsection*{Proof of Lemma \ref{lem:union:coreset}}

\newtheorem*{recall:lem:union:coreset}{Lemma~\ref{lem:union:coreset}}

\begin{recall:lem:union:coreset}[Union Property]
Let $P$ be a set of points in a metric space $(X,\dist)$. 
Let $k, z\in \Nats$ and $\eps \geq 0$ be parameters. 
Let $P$ be partitioned into disjoint subsets $P_1,\cdots,P_s$, and
let $Z = \{z_1,\cdots, z_s\}$ be a set of numbers such that 
$\opt_{k,z_i}(P_i) \leq \optkz(P)$ for each~$P_i$. 
If $P_i^*$ is an $(\eps, k, z_i)$-mini-ball covering of $P_i$ for each $1 \leq i \leq s$,
then $\cup_{i=1}^s P_i^*$ is an $(\eps,k,z)$-mini-ball covering of $P$. 
\end{recall:lem:union:coreset}

\begin{proof}
First of all, observe that the weight of the point set $P$ is preserved by 
the union of the mini-ball coverings. Indeed, $\sum_{p \in P} \weight(p)
= \sum_{i=1}^s \sum_{p \in P_i} \weight(p)
= \sum_{i=1}^s \sum_{q \in P^*_i} \weight(q)$. 

Next, consider an arbitrary point $p \in P$, and $P_i$ be the subset containing~$p$.
Then $p$ has a representative point $q$ in the $(\eps, k, z_i)$-mini-ball covering $P_i^*$ of~$P_i$. 
By Definition~\ref{def:miniball:covering}, $\dist(p, q) \leq \eps \cdot \opt_{k,z_i}(P_i) \leq \eps \cdot \optkz(P)$,
which proves that $\cup_{i=1}^s P_i^*$ is an $(\eps,k,z)$-mini-ball covering of $P$.
\end{proof}

%------------------------------------------------------------------------------------------

\subsection*{Proof of Lemma \ref{lem:merge:coreset}}

\newtheorem*{recall:lem:merge:coreset}{Lemma~\ref{lem:merge:coreset}}

\begin{recall:lem:merge:coreset}[Transitive Property]
Let $P$ be a set of $n$ points in a metric space $(X,\dist)$. 
Let $k, z \in \Nats$ and $\eps, \peps \geq 0$ be four parameters. 
Let $P^*$ be a $(\peps,k,z)$-mini-ball covering of $P$,
and let $Q^*$ be an $(\eps, k, z)$-mini-ball covering of $P^*$. 
Then, $Q^*$ is an $(\eps+\peps+\eps\peps, k,z)$-mini-ball covering of $P$.
\end{recall:lem:merge:coreset}

\begin{proof}
Note that the weight-preservation property of mini-ball covering implies that
$\sum_{p \in P} \weight(p) = \sum_{p^* \in P^*} \weight(p^*) = \sum_{q^* \in Q^*} \weight(q^*)$. 
It remains to show that any point $p \in P$ has a representative point 
$q^* \in Q^*$ so that $\dist(p, q^*) \leq (\eps+\peps+\eps\peps)\cdot \optkz(P)$. 

Since $P^*$ is an $(\peps,k,z)$-mini-ball covering for $P$, there is a representative point $p^* \in P^*$ 
for $p$ for which $\dist(p, p^*) \leq \peps\cdot \optkz(P)$. 
Similarly, since $Q^*$ is an $(\eps,k,z)$-mini-ball covering for $P^*$, there is a representative point $q^* \in Q^*$ 
for $p^*$ such that $\dist(p^*, q^*) \leq \eps\cdot \optkz(P^*)$.
Hence,
\[
\begin{array}{lll}
    \dist(p, q^*) 
    & \leq & \dist(p, p^*) + \dist(p^*, q^*)  \\
    & \leq & \peps\cdot \optkz(P) + \eps \cdot \optkz(P^*) \\
    & \leq & \peps\cdot \optkz(P) + \eps \cdot (1+\peps) \cdot \optkz(P) \hspace*{10mm} \mbox{ (by Definition~\ref{def:coreset})} \\
    & = & (\eps+\peps+\eps\peps) \cdot \optkz(P) \enspace .
\end{array}
\]
We conclude that $q^*$ is a valid representative for $p$, thus finishing the proof.
\end{proof}

\subsection*{Proof of Lemma \ref{lem:coreset:size}}

\newtheorem*{recall:lem:coreset:size}{Lemma~\ref{lem:coreset:size}}

\begin{recall:lem:coreset:size}
Let $P$ be a finite set of points in a metric space $(X,\dist)$ of doubling dimension $d$.
Let $0 < \delta \leq \optkz(P)$, and let $Q \subseteq P$ be a subset of $P$ such that for 
any two distinct points $q_1, q_2 \in Q$, $\dist(q_1, q_2) > \delta$. 
Then $|Q| \leq k\left( \frac{4\cdot\optkz(P)}{\delta} \right)^d+z$.
\end{recall:lem:coreset:size}

\begin{proof}
Consider an optimal solution for the $k$-center problem with $z$ outliers on $P$.
Let $\Bopt$ be the set of $k$~balls in this optimal solution.
% and let $\Pout = \{q_1,\cdots,q_z\}\subset P$ 
% be the outliers, that is, the points not covered by the balls in~$\Bopt$.
Since $(X,\dist)$ is a metric space of doubling dimension $d$, every ball in $\Bopt$ can be covered by at most $(2\cdot\frac{\optkz(P)}{\delta/2})^d = (4\cdot\frac{\optkz(P)}{\delta})^d$ mini-balls of radius $\delta/2$. 
Besides, as the distance between points of $Q$ is more than $\delta$, each mini-ball of radius $\delta/2$ contains at most one point of $Q$. Therefore, the number of points in $Q$ covered by $\Bopt$ is at most $(4\cdot\frac{\optkz(P)}{\delta})^d$. Besides, as $\Bopt$ is an optimal solution, at most $z$ points of $Q$ are not covered by $\Bopt$. Thus, 
$|Q| \leq k( 4\cdot\frac{\optkz(P)}{\delta})^d+z$.

\end{proof}

\subsection*{Proof of Lemma \ref{lem:rho}}

\newtheorem*{recall:lem:rho}{Lemma~\ref{lem:rho}}

\begin{recall:lem:rho}
After the point $p_t$ arriving at time~$t$ has been handled, we have:
for each point $p \in P(t)$
there is a representative point $q\in P^*$ such that $\dist(p, q) \leq \eps \cdot r$. 
\end{recall:lem:rho}

%------------------------------------------------------------------------------------------

\begin{proof}
We may assume by induction that the lemma holds after the previous point has been 
handled. (Note that the lemma trivially holds before the arrival of the first point.)
We now show that the lemma also holds after processing~$p_t$. 

% If there is a point $q\in P^*$ such that $\dist(p_t, q) \leq \frac{\eps}{2}\cdot r$, 
% we choose $q$ as the representative point of $p_t$. Otherwise, we add $p_t$ to $P^*$, 
% which means that $p_t$ is its representative point in distance~$0$. Hence, in both cases, 
% the distance between $p_t$ and its representative point is at most 
% $\frac{\eps}{2}\cdot r \leq \eps \cdot r$, and lemma holds
% after this step of algorithm. 
% \mdb{Explicitly refer to line numbers in the algorithm, instead of ``this step''. To save space, we could just write:}
It is easily checked that after lines~\ref{line:streaming:rep:itself}
of the algorithm, there is indeed a representative in $P^*$ within distance $\eps\cdot r$, namely the point~$q$ in line~\ref{line:streaming:rep:q} or $p$ itself in line \ref{line:streaming:rep:itself}.

In each iteration of the while-loop in lines \ref{line:streaming:threshold}, the value of~$r$ is doubled
and \UpdateCoreset is called. We will show that the lemma remains true after each iteration. 
%In the following of the algorithm, we may update the value of $r$ to $2\cdot r$ 
%and then update $P^*$ according to the new value of $r$. 
%We show that lemma still holds after updating $P^*$. 
Let $r^-$ and $r^+$ denote the value of $r$ just before
and after updating it in line~\ref{line:streaming:doubling:r}, respectively, so $r^+ = 2\cdot r^-$. Let $p$ be an arbitrary point in $P(t)$. 
Let $q^-$ be the representative point of $p$ before the call to \UpdateCoreset.
Because the statement of the lemma holds before the call, 
we have $\dist(p, q^-) \leq \eps\cdot r^-$.
Let $q^+$ denote the representative point of $q^-$ just after the call to \UpdateCoreset.
(Possibly $q^+=q^-$.)
Since the distance parameter~$\delta$ of \UpdateCoreset
in the call is set to $(\eps/2)\cdot r^+$, we know that
$\dist(q^-, q^+) \leq (\eps/2)\cdot r^+$.
Hence,
\[
\dist(p, q^+) 
\leq \dist(p, q^-)+\dist(q^-, q^+)
\leq \eps\cdot r^{-} + \frac{\eps}{2}\cdot r^+
\leq \eps \cdot \frac{r^+}{2} + \frac{\eps}{2}\cdot r^+
= \eps \cdot r^+ \enspace ,
\]
which finishes the proof of the lemma.
\end{proof}

%------------------------------------------------------------------------------------------

\subsection*{Proof of Lemma \ref{lem:bound:outliers:per:machine}}

\newtheorem*{recall:lem:bound:outliers:per:machine}{Lemma~\ref{lem:bound:outliers:per:machine}}

\begin{recall:lem:bound:outliers:per:machine}[\cite{DBLP:journals/pvldb/CeccarelloPP19}]
% For each machine $M_i$ we have 
$\Pr{\forall_{1\leq i \leq m} |P_i \cap \Pout|  \leq \outliersBound} \geq 1-1/n^2$.
\end{recall:lem:bound:outliers:per:machine}

\begin{proof}
Consider an optimal solution for the $k$-center problem with $z$ outliers on $P$.
Let $\Bopt$ be the set of $k$~balls in this optimal solution and let $\Pout = \{q_1,\cdots,q_z\}\subset P$ 
be the outliers, that is, the points not covered by the balls in~$\Bopt$.
Let us consider a random variable $X_i$ that corresponds to the number of outliers that 
are assigned to an machine $M_i$ for $i \in [m]$. 
First of all, observe that $\Ex{X_i} = \frac{z}{m}$. 
Next, we consider $X_i = \sum_{j=1}^z Y_{ij}$ where every random variable $Y_{ij}$ is an indicator 
random variable which is one of the outlier $q_j$ is assigned to the machine $M_i$ and zero otherwise. 
Now, we use the Chernoff bound to show that $X_i$ is concentrated around its expectation. 

\begin{lemma}[Multiplicative Chernoff bound]\cite{Cher,DBLP:journals/rsa/CzumajS07}
Let $X_1,\cdots,X_N$ be independent random variables, with $\Pr{X_i=1} = p$ 
and $\Pr{X_i=0}=1-p$ for each $i$ and for certain $0 \le p \le 1$.
Let $X = \sum_{i=1}^N X_i$. Then,
\begin{itemize}
\item for any $\tau \ge 6\cdot 	\Ex{X}$, we have
$
		\Pr{X \ge \tau} \le 2^{-\tau} \enspace .
$
\end{itemize}
\end{lemma}

We let $\tau = \frac{6z}{m}+3\log{n}$, then, $\tau \geq 6\cdot \Ex{X_i}$.
Thus, using the Chernoff bound we have 
\[
    \Pr{X_i \ge \frac{6z}{m}+3\log{n}} 
    = \Pr{X_i \ge \tau} 
    \le 2^{-\tau} \le 2^{-(\frac{6z}{m}+3\log{n})} 
    \le 2^{-3\log{n}} 
    = 1/n^3 \enspace .
\]

Now, since we have $m = \sqrt{n}$ machines, we use a union bound to obtain  
$$ 
\Pr{\exists_{i \in [m]} X_i \ge \outliersBound} \le \sum_{i+1}^m \Pr{X_i \ge \outliersBound} \le m/n^3\le 1/n^2 \enspace . \qedhere$$ 
\end{proof}

%------------------------------------------------------------------------------------------
\section{Omitted proofs of the $\Omega(k/\eps^d)$ lower bound for the streaming model}
\label{app:lb:streaming}
%------------------------------------------------------------------------------------------

This section provides the missing proofs for our lower bounds.
Recall that $P^*(t') \subseteq P^*(t) \cup P^- \cup P^+$ is the coreset maintained by the algorithm, which is supposed to be an $(\eps,k,z)$-coreset of the point set $P(t') = P(t) \cup P^- \cup P^+$.

%-----------------------------------------------------------------------------
\begin{lemma}
\label{lem:OPT:upper:bound}
The optimal $k$-center radius with $z$ outliers of the $(\eps,k,z)$-coreset $P^*(t')$ is 
at most $r$, i.e., $\optkz(P^*(t')) \leq r$. 
\end{lemma}
%-----------------------------------------------------------------------------
% \begin{proofof}{Lemma \ref{lem:OPT:upper:bound}}
\begin{proof}
To this end, we show that there exist $k$ balls centered at $k$ centers of radius $r$ 
that can cover all points in $(C_1\cup\ldots\cup C_{k-2d+1} \cup P^+ \cup P^- )\setminus \{p^*\}$. 
Recall that $p^*=(p_1^*,\ldots,p_d^*)$ is the point that is not explicitly stored in~$P^*(t)$ and 
we assume that $p^*$ belongs to a cluster $C_{i^*}$ for $i^* \in [k-2d+1]$. 

Observe that for every $i\neq i^*$, since the diameter of $C_i$ is $\sqrt{d}\lambda$, 
all points of $C_i$ can be covered by a ball of radius $\sqrt{d}\lambda/2$. 
We assumed that $\lambda := 1/(4d\eps)$ is an integer, $h := d(\lambda+2)/2$ and $r :=\sqrt{h^2-2h+d}$.
As $d \geq 1$, we  observe that $\sqrt{d}\lambda/2 \leq r$.

%-----------------------------------------------------------------------------
\begin{figure}[h]
\begin{center}
\hspace*{-3mm}
\includegraphics{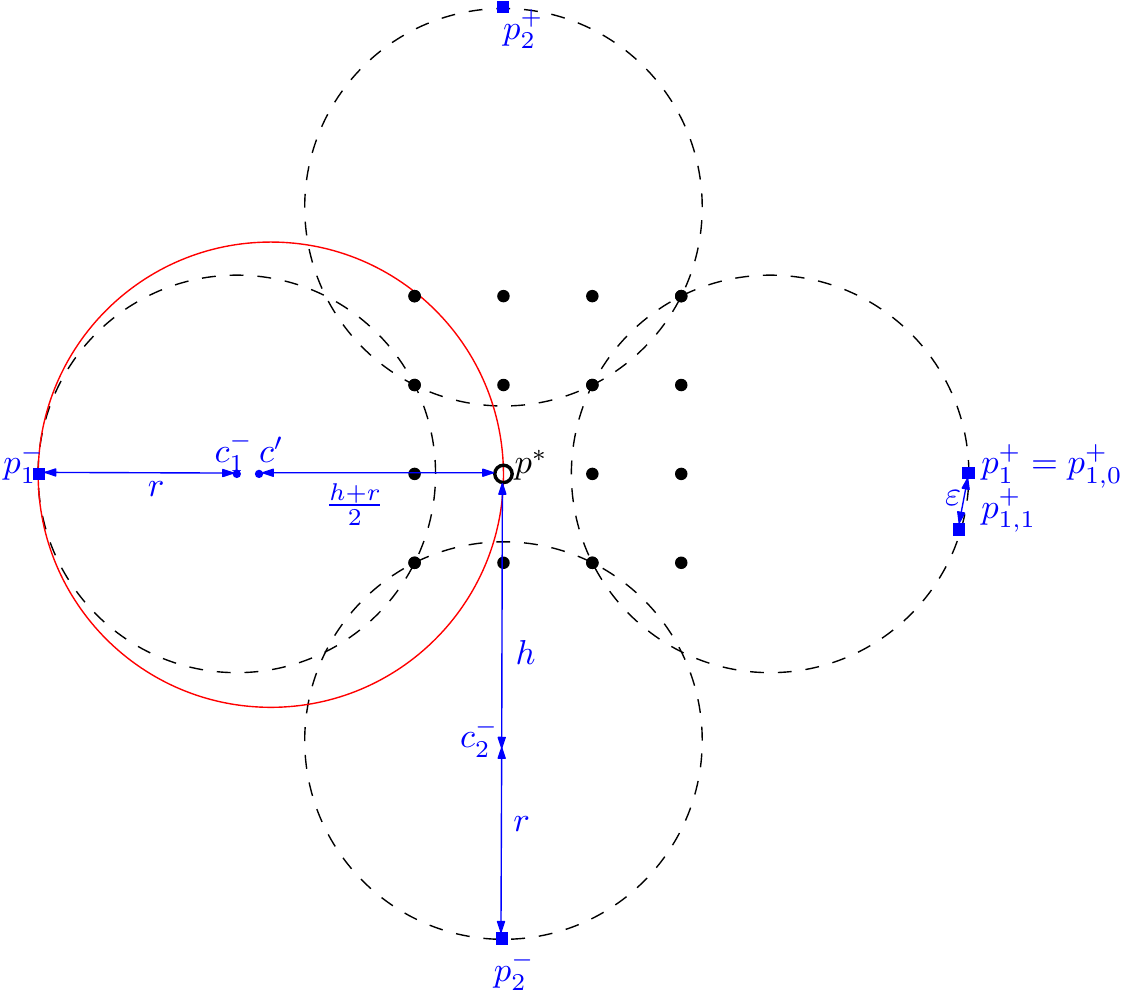}
\end{center}
\caption{Illustration of the lower bound for the streaming model. 
% We have $\lambda := 1/(4d\eps)$ is an integer, $h := d(\lambda+2)/2$ and $r :=\sqrt{h^2-2h+d}$.
Here, $P^*(t')$ underestimates $\optkz(P(t'))$ since $2d$ balls of radius $r$ can cover
$P^+\cup P^-\cup C_{i^*}\setminus\{p^*\}$ (dashed balls), 
and then $\optkz(P^*(t')) \leq r$.
However, $\optkz(P(t')) = (r+h)/2$ (the red ball).
} 
\label{fig:lb:streaming:missing}
\end{figure}
%-----------------------------------------------------------------------------

To cover $C_{i^*} \cup P^- \cup P^+$, we define $2d$ centers $c_1^+,\ldots,c_d^+$ and $c_1^-,\ldots,c_d^-$, 
where $c_j^+ = (c_{j,1}^+,\ldots,c_{j,d}^+)$ such that $c_{j,j}^+:= p_j^*+h$ and $c_{j,\ell}^+:=p_\ell^*$ for all $\ell\neq j$.
Similarly, $c_j^- = (c_{j,1}^-,\ldots,c_{j,d}^-)$ such that $c_{j,j}^-:= p_j^*-h$ and $c_{j,\ell}^-:=p^*_\ell$ for all $\ell\neq j$; see Figure \ref{fig:lb:streaming:missing}.
We claim that $2d$ balls centered at these $2d$ centers of radius $r$ cover all points in $P^+\cup P^- \cup C_{i^*} \setminus \{p^*\}$. 

For the moment suppose this claim is correct. 
The number of clusters $C_i$ where $i\neq i^*$ is $k-2d$. 
We just showed that all points of $C_i$ can be covered by a ball of radius $r$. 
We also claimed (which needs to be proved) 
that there exist $2d$ balls centered at these $2d$ centers of radius $r$ cover all points in $P^+\cup P^- \cup C_{i^*} \setminus \{p^*\}$. 
In addition, 
in Claim~\ref{claim:outlier_weight:<=z} (below) we prove that 
the total weight of the outlier points $o_1,\ldots,o_z$ in~$P^*(t')$ is at most~$z$.
Thus, $\optkz(P^*(t')) \leq r$ as we want to prove. 
%-----------------------------------------------------------------------------

Next, we prove the claim. 
Indeed, for each $p_j^+\in P^+$, we have $\dist(p_j^+, c_j^+) = r$. 
Similarly, for each $p_j^-\in P^-$, we have $\dist(p_j^-, c_j^-) = r$.
Let $q=(q_1,\ldots,q_d)$ be an arbitrary point in $C_{i^*}\setminus\{p^*\}$. 
We define $j_q := \arg\max_{\ell\in[d]}{|q_\ell-p^*_{\ell}|}$ to be the dimension along which 
$p^*$ and $q$ have the maximum distance from each other, and let $\mu_q:=|q_{j_q}-p_{j_q}^*|$ be 
the distance along the dimension $j_q$. Observe that $\mu_q \geq 1$. 

%-----------------------------------------------------------------------------

\begin{claim}
For an arbitrary point $q=(q_1,\ldots,q_d) \in C_{i^*}\setminus\{p^*\}$, we have the following bounds: 
\begin{itemize}
    \item If $q_{j_q}-p_{j_q}^* > 0$, then for the center $c_{j_q}^+ \in \{c_1^+,\ldots,c_d^+\}$ 
    we have $\dist(q, c_{j_q}^+) \leq r$. 
    \item If $q_{j_q}-p_{j_q}^* < 0$, then for the center $c_{j_q}^- \in  \{c_1^-,\ldots,c_d^-\}$ 
    we have $\dist(q, c_{j_q}^-) \leq r$.
\end{itemize}
\label{clm:arbitrary:covered}
\end{claim}

%-----------------------------------------------------------------------------

\begin{proofinproof}
First assume that $q_{j_q}-p_{j_q}^* > 0$. The other case is proven similarly. 
We let $c_{j_q}^+=(c_{j_q,1}^+,\ldots,c_{j_q,d}^+)$. 
Recall that $c_{j_q,j_q}^+=p_{j_q}^*+h$, and $c_{j_q,\ell}^+=p_{\ell}^*$ for all $\ell\neq j_q$.
Therefore, $|q_{j_q}-c_{j_q,j_q}^+|=|q_{j_q}-(p^*_{j_q}+h)| = |\mu_q-h|$ 
and for all $\ell\neq j_q$ we have $|q_\ell-c_{j_q,\ell}^+| = |q_\ell-p^*_{\ell}| \leq \mu_q$. 
Thus, 
%-----------------------------------------------------------------------------
\[
        \dist(q,c_{j_q}^+)
        =\sqrt{\sum_{\ell=1}^d |q_\ell-c_{j_q,\ell}^+|^2}
        \leq \sqrt{(d-1)\mu_q^2+(\mu_q-h)^2} 
        = \sqrt{h^2+d\mu_q^2-2\mu_q h} \enspace .
\]
%-----------------------------------------------------------------------------

Since $0<\mu_q\leq\lambda$ we have 
$h = \frac{d}{2}(\lambda+2) \geq \frac{d}{2}(\mu_q+1)$. 
As $\mu_q\geq 1$, we can multiply both sides of this inequality by $2(\mu_q-1)$ to obtain 
$h\cdot(2\mu_q-2) \geq \frac{d}{2}(\mu_q+1)\cdot2(\mu_q-1) = d(\mu_q^2-1)$ 
what yields $-2h+d \geq d\mu_q^2-2\mu_q h$. 
Finally, by adding $h^2$ to both sides we have $h^2-2h+d \geq h^2+d\mu_q^2-2\mu_q h$.
Recall that $r=\sqrt{h^2-2h+d}$ and $dist(q, c_{j_q}^+) \leq \sqrt{h^2+d\mu_q^2-2\mu_q h}$. 
Thus $dist(q, c_{j_q}^+) \leq r$ that proves this claim.
\end{proofinproof}

Next, we prove that $\optkz(P(t')) = (h+r)/2$. 

\begin{claim}
\label{clm:opt:bound:equality}
$\optkz(P(t')) = (h+r)/2$. 
\end{claim}

\begin{proofinproof}
We proved that $\optkz(P(t')) \ge (h+r)/2$. 
Now, we prove that $\optkz(P(t')) \le (h+r)/2$. 
In fact, in Claim~\ref{clm:arbitrary:covered}, we proved that 
the balls in $\{\ball(c_1^-, r),\ldots,\ball(c_d^-, r)\} \cup \{\ball(c_1^+, r),\ldots,\ball(c_d^+, r)\}$
cover all points in $P^+\cup P^- \cup C^{i^*}\setminus \{p^*\}$. 
%Let us consider one of these balls, say $\ball(c_1^-, r)$ as we see in Figure~\ref{fig:lb:streaming:config}. 
We define a center $c' = (c'_1,\ldots,c'_d)$ such that $c'_1:= p_1^*-(h+r)/2$ and 
$c'_{\ell}:=p^*_\ell$ for all $\ell\neq 1$. 
Note that $h\geq r$ and $c_1^- = (c_{1,1}^-,\ldots,c_{1,d}^-)$ such that $c_{1,1}^-:= p_1^*-h$ and $c_{j,\ell}^+:=p^*_\ell$ for all $\ell\neq 1$. 
Thus, $\dist(c_1^-,c') = |c_{1,1}^- - c'_1| = |p_1^*-h - (p_1^*-(h+r)/2)| = |(h-r)/2| = (h-r)/2$. 
Now, we observe that using the triangle inequality,  $\ball(c_1^-, r) \subseteq \ball(c', (h+r)/2)$.
(In Figure~\ref{fig:lb:streaming:missing}, the ball $\ball(c', (h+r)/2)$ is shown in red.) 
Indeed, let $q$ be an arbitrary point in $\ball(c_1^-,r)$.
Using the triangle inequality, we have $\dist(q,c') \le \dist(q,c_1^-) + \dist(c_1^-, c') 
\leq r + (h-r)/2 = (h+r)/2$, which means $q \in \ball(c', (h+r)/2)$. 
Therefore, $\ball(c_1^-,r) \subseteq \ball(c', (h+r)/2)$.
In addition, observe that the ball $\ball(c', (h+r)/2)$ covers the point $p^*$. 
Now, if we replace the ball $\ball(c_1^-,r)$ by $\ball(c', (h+r)/2)$, 
the union of balls $\{\ball(c',(h+r)/2)\} \cup \{\ball(c_2^-, r),\ldots,\ball(c_d^-, r)\} \cup \{\ball(c_1^+, r),\ldots,\ball(c_d^+, r)\}$ will 
cover $C_{i^*}\cup P^+ \cup P^-$. 
This essentially means that $\optkz(P+2d) \leq (h+r)/2$.
\end{proofinproof}
%-----------------------------------------------------------------------------
It remains to argue that the total weight of the outlier points $o_1,\ldots,o_z$ in~$P^*(t')$ is at most~$z$.

\begin{claim}
\label{claim:outlier_weight:<=z}
The total weight of the outlier points $o_1,\ldots,o_z$ in~$P^*(t')$ is at most~$z$.
\end{claim}

\begin{proofinproof}
% \begin{proofof}{Claim \ref{claim:outlier_weight:<=z}}
Suppose this is not the case. 
We consider an optimal set~$\B^*$ of $k$ balls that covers the weighted points in $P^*(t')$ except 
a total weight of at most $z$.   
Since $P^*(t')$ is an $(\eps,k,z)$-coreset for $P(t')$, 
the radius of balls in the set $\B^*$ is at most $(1+\eps)\cdot \optkz(P(t')) = (1+\eps)\cdot (h+r)/2$, 
where we use Claim~\ref{clm:opt:bound:equality} that shows $\optkz(P(t')) = (h+r)/2$. 

Suppose for the sake of contradiction, 
the total weight of the outlier points $o_1,\ldots,o_z$ in~$P^*(t')$ is more than~$z$, 
thus, at least one outlier point, say $o_1$, must be covered by a ball from~$\B^*$.
On the other hand, the nearest non-outlier point to $o_1$ is at distance at least~$4(h+r)$ from $o_1$. 
Thus, since the radius of balls in the optimal set $\B^*$ is at most $(1+\eps)\cdot (h+r)/2$, 
such an outlier $o_1$ must be a singleton in its ball.
Let $k'\geq 1$ be the number of outliers that are covered by singleton balls from~$\B^*$. 

As $P^*(t')$ is an $(\eps,k,z)$-coreset of $P(t')$, 
if we expand the radius of the balls in~$\B^*$ 
by~$\eps\cdot\optkz(P(t')) = \eps(h+r)/2$, then the total weight of points in $P(t')$ that
are not covered by these expanded balls is at most $z$. 
Therefore, the points in $C_1\cup\ldots\cup C_{k-2d+1} \cup P^+\cup P^-$  need to  
be covered by the remaining $k-k'$ expanded balls.

Consider the following $k$ sets:
$2d$ sets $\{p^+_1\},\ldots,\{p^+_d\}$ and $\{p^-_1\},\ldots,\{p^-_d\}$, 
as well as $k-2d$ sets $C_i$, where $i\neq i^*$.
Since the pairwise distances between these $2d+(k-2d)=k$ sets are at least $\sqrt{2}(h+r)$.
Besides, the radius of the expanded balls is at most $(1+2\eps)(h+r)/2$, 
where $1+2\eps < \sqrt{2}$ since we assume $\eps \leq \frac{1}{8d}$.
Hence, each of the remaining $k-k'$ expanded balls can cover at most one of these $k$ sets.
As $k-k'$ balls of $\B^*$ remained for these $k$ sets, then at least $k'$ sets cannot be covered by the expanded balls.

We assumed the weight of every point $p^+_i$ (or $p^-_i$) is two\footnote{
However, we can in fact change these weighted points to unweighted points 
by replacing each weighted point $p^+_i$ (similarly $p^-_i$) by two unweighted points $p^+_{i,0}$ and $p^+_{i,1}$, 
where $p^+_{i,0}$ is at the same place of $p^+_i$, and $p^+_{i,1}$ is on the boundary of the ball $\ball(c^+_i,r)$ 
at distance $\eps$ of $p^+_{i,0}$ ;see Figure \ref{fig:lb:streaming:missing}. 
It is simple to see that our arguments still hold for this unweighted case as 
$\dist(p^-_{i,0},c^+_i)=\dist(p^-_{i,1},c^+_i) = r$, and $\dist(p^-_{i,0},p^+_i) = \eps$.}.
In addition, the number of points in every  $C_i$ is at least $2$ since $\lambda \geq 2$ and $|C_i|=(\lambda+1)^d$.
Recall that the union of balls in $\B^*$ does not cover $z-k'$ points in $\{o_1,\ldots,o_z\}$.
Therefore, the total weight of the points that the expanded balls do not cover
is at least $(z-k')+2k' = z+k'$. As the total weight of outliers must be at most $z$,
we must have $k'=0$, which is a contradiction to the assumption that 
at least one outlier point must be covered by a ball from~$\B^*$. 
That is, the total weight of the outlier points $o_1,\ldots,o_z$ in~$P^*(t')$ is at most~$z$ 
as we want.
\end{proofinproof}

%------------------------------------------------------------------------------------------

% and let $Q'\subset \{q_1,\ldots,q_z\}$ be the uncovered outlier points.

% The expanded balls do not cover the points from $Q'$, and so at most~$k'$ points
% from $P^-\cup \{p_1,\ldots,p_{(k-1)\ell}\} \cup P^+$ are uncovered. The remaining $k\ell+1-k'$ points
% must all be covered by the remaining $k-k'$ expanded balls. Thus there must be
% a ball that covers at least~$\ell+2$ points. Indeed, if this is not the case then 
% % the number of points covered by the $k-k'$ balls is at most

\end{proof}

%------------------------------------------------------------------------------------------

%------------------------------------------------------------------------------------------
\begin{lemma}
\label{lem:r:bound}
Let $0 < \eps \leq \frac{1}{8d}$. 
Let $\lambda := 1/(4d\eps)$ be an integer, 
$h := d(\lambda+2)/2$ and $r :=\sqrt{h^2-2h+d}$. 
Then, $r < (1-\eps)(r+h)/2$.
\end{lemma}
%------------------------------------------------------------------------------------------
\begin{proof}
We start with the statement that we want to prove, and then derive a sequence of equivalent statements, 
until we arrive at a statement that is easily seen to be true. 
Indeed, assume that $r < (1-\eps)(r+h)/2$ is correct. 
Then, we have $2r  <  r+h-\eps r - \eps h $ which means that 
$ r(1+\eps)  <  h(1-\eps) $. 
Since both sides of the inequality are non-negative, 
we can raise them to the power of two to obtain 
$r^2(1+\eps)^2  <  h^2(1-\eps)^2$. 
Since $r=\sqrt{h^2-2h+d}$ and $h^2-2h+d \geq 0$, 
we have $(h^2-2h+d)(1+2\eps+\eps^2)  < h^2(1-2\eps+\eps^2)$ 
which means that 
$ h(-4h\eps + 1) + (h-d) + \eps^2 (2h-d) + 2\eps(2h-d) > 0 $.

We observe that all these four terms of this inequality are non-negative which proves the claim. 
First of all, since $\eps \leq \frac{1}{8d}$, we observe that $-4h\eps + 1 \geq 0$, 
and so, the first term is at least zero. 
Next, as $\lambda \geq 1$, we have $h=d(\lambda+2)/2 > d$ that implies that 
$2h-d > h - d > 0$. Hence, the second, third, and forth terms are greater than zero.
Thus the left-hand side of the above inequality is greater than zero and this finishes the proof.
\end{proof}

\end{document}